%% file: main_v2.tex
\documentclass[
  reprint,
  aps,
  pra,
  superscriptaddress,
  amsmath,amssymb,
  floatfix,
  a4paper
]{revtex4-2}

\usepackage[T1]{fontenc}
\usepackage[utf8]{inputenc} 
\usepackage{graphicx}
\usepackage{bm}
\usepackage[pdfusetitle]{hyperref}
\usepackage{xcolor}
\definecolor{measurebg}{HTML}{F2EFEA}
\usepackage{physics} 
\usepackage{siunitx} 
\usepackage{mathtools}
\usepackage{todonotes}
\usepackage{tikz}
\usetikzlibrary{quantikz2}
\usepackage{amsthm}
\usepackage{nicematrix}
\usepackage{algorithm, algpseudocode}
\usepackage{cleveref}
\AddToHook{cmd/appendix/before}{%
    \crefalias{section}{appendix}%
    \crefalias{subsection}{appendix}
    \crefalias{subsubsection}{appendix}
}
\usepackage{booktabs, multirow}
\usetikzlibrary{arrows.meta, positioning}

\definecolor{palRed}{HTML}{E6194B}
\definecolor{palOrange}{HTML}{F58231}
\definecolor{palLime}{HTML}{BFEF45}
\definecolor{palGreen}{HTML}{3CB44B}
\definecolor{palBlue}{HTML}{4363D8}
\definecolor{palPurple}{HTML}{911EB4}
\definecolor{palCyan}{HTML}{42D4F4}

\colorlet{figGrayFill}{black!7}
\colorlet{figGrayLine}{black!60}
\colorlet{figContribFill}{palGreen!15}
\colorlet{figContribLine}{palGreen!80!black}
\colorlet{figCorrFill}{palOrange!18}
\colorlet{figCorrLine}{palOrange!85!black}

\hypersetup{
  colorlinks=true,
  linkcolor=blue,
  citecolor=blue,
  urlcolor=blue,
  pdftitle={Exploring the landscape of compact magic-state distillation factories},
  pdfauthor={Hugo Jacinto, Xavier Valcarce, Victor Barizien, Élie Gouzien and Nicolas Sangouard}
}

\theoremstyle{definition}
\newtheorem{definition}{Definition}
\theoremstyle{plain}
\newtheorem{theorem}{Theorem}

\newtheorem{lemma}{Lemma}
\newtheorem*{theorem*}{Theorem}

\newtheorem{proposition}{Proposition}
\newtheorem*{proposition*}{Proposition}

\newcommand{\C}{\mathord{\mathrm{C}}}

\newcommand{\CS}{\mathord{\C S}}
\newcommand{\CT}{\mathord{\C T}}

\newcommand{\CCZ}{\mathord{\C\C Z}}
\newcommand{\CCS}{\mathord{\C\C S}}
\newcommand{\CCCZ}{\mathord{\C\C\C Z}}

\begin{document}

\title{Exploring the landscape of compact magic-state distillation factories}
\author{Hugo Jacinto}
\email{hugo.jacinto@alice-bob.com}
\affiliation{Université Paris--Saclay, CEA, CNRS, Institut de Physique Théorique, \num[detect-all]{91191} Gif-sur-Yvette, France}
\affiliation{Alice \& Bob, 49 boulevard du Général Martial Valin, \num[detect-all]{75015} Paris, France}

\author{Xavier Valcarce}
\affiliation{Université Paris--Saclay, CEA, CNRS, Institut de Physique Théorique, \num[detect-all]{91191} Gif-sur-Yvette, France}

\author{Victor Barizien}
\affiliation{Université Paris--Saclay, CEA, CNRS, Institut de Physique Théorique, \num[detect-all]{91191} Gif-sur-Yvette, France}
\affiliation{Department of Applied Physics, University of Geneva, CH-1205 Geneva, Switzerland}

\author{Élie Gouzien}
\affiliation{Alice \& Bob, 49 boulevard du Général Martial Valin, \num[detect-all]{75015} Paris, France}

\author{Nicolas Sangouard}
\affiliation{Université Paris--Saclay, CEA, CNRS, Institut de Physique Théorique, \num[detect-all]{91191} Gif-sur-Yvette, France}
\date{\today}

\begin{abstract}
Producing high-fidelity magic states using the smallest possible number of physical qubits and operations stands as a very important challenge to achieve fault-tolerant quantum computation at scale.
Besides emerging proposals for alternative methods such as cultivation, magic state distillation remains essential for achieving very low error rates.
Known distillation protocols are usually built through quantum codes derived from triorthogonal matrices.
Here, exploiting the specific noise structure present in magic state distillation protocols, we show that classical error-correcting codes offer a simpler framework for deriving these protocols.
This formulation is particularly well suited to systematic numerical and analytical studies of distillation protocols involving a fixed number of qubits.
Specifically, we use a SAT solver to derive a series of no-go theorems that relate key figures of merit, including the number of qubits and the factory distance.
In particular, considering the standard implementation as a circuit of Z-diagonal Pauli product rotations followed by measurements, we show that no $T$-to-$T$ distillation protocol on fewer than eight qubits can exceed distance 3, and no $T$-to-$\CCZ$ protocol distance 2.
Our results also include new such distillation protocols with the smallest number of qubits for a given distance in the literature, namely distance 4 and 5 $T$-to-$T$ protocols supported on 10 and 11 qubits, as well as distance 3 and 4 $T$-to-$\CCZ$ distillation protocols supported on 9 and 10 qubits.
Finally, going beyond unitary circuits by recycling qubits through mid-circuit measurement and reinitialization, we obtain an implementation of the distance-5 $49T$-to-$1T$ protocol on only 5 active qubits.

\end{abstract}


\maketitle

\input{Introduction_v2}
\input{Preliminaries_v2}

\input{Canonical_family_v2}
\input{Numerical_v2}
\input{Recycling}

\input{Conclusion_v2}

\section*{Note added}
Shortly after releasing the first version of this manuscript, we became aware of the independent and concurrent work~\cite{singh2026borrowed}, which develops a unified numerical search framework for distillation based on a borrowed-identity condition rather than triorthogonality conditions.
Their search targets the multi-output ($k>1$) regime at distance 2, whereas our work focuses on single-output protocols at higher distance $d \geq 3$.

\begin{acknowledgments}

The authors would like to thank J. Guillaud and D. Ruiz for fruitful discussions at an early stage of the project and V. Londe for critical reading of
the manuscript draft as well as for encouraging discussions during the project.
We are additionally grateful to C. Gidney for feedback that prompted a second version of this preprint.
We acknowledge funding
by Agence Nationale de la Recherche in the framework
of France 2030 with the reference ANR-22-PETQ-0007
and project name EPiQ. This work was also partially
supported by the French National program Programme
d’investissement d’avenir, IRT Nanoelec, with the reference ANR-10-AIRT-05.
\end{acknowledgments}

\bibliography{sample}

\clearpage
\appendix
\crefalias{section}{appendix}
\input{Commutations_relations}

\clearpage
\input{Algebras_v2}
\clearpage
\input{Matrices}
\clearpage
\input{Recycling_app}

\end{document}

%% file: Introduction_v2.tex
\section{Introduction}

Quantum computing promises to solve computational problems that are hard to solve classically, with potential applications spanning cryptography, materials science, and optimization~\cite{beverland_resource_estimation}.
However, the practical realization of large-scale quantum computation faces many difficulties, primarily due to the fragility of quantum states and their susceptibility to noise~\cite{Preskill_2025}.
Quantum error correction provides a way of dealing with these errors, enabling reliable quantum processing even in the presence of noise~\cite{Campbell_2017_Terhal}.
While error-correcting codes can, in principle, suppress errors to arbitrarily low levels, their implementation may require substantial physical and computational resources~\cite{Gidney_2021,gidney2025factor2048bitrsa, litinski2023compute256bitellipticcurve, Gouzien_2021,Gouzien_2023, beverland_resource_estimation,Zhou_2025,Jacinto_2026, Review_applications}.

A fundamental constraint on quantum error-correcting codes (QECC) arises from the Eastin--Knill theorem~\cite{eastin2009}, which forbids the transversal implementation of a universal gate set at the logical level for any binary QECC with distance $d > 1$.
Additionally, strict limitations on the set of gates that can be efficiently performed via unitary circuits exist on topological QECC~\cite{Bravyi_2013} as well as for architectures of strongly noise-biased qubits concatenated with classical error-correcting codes~\cite{barizien2025accessiblequantumgatesclassical}.
As a result, most practical QECCs, such as the widely studied surface code~\cite{Dennis_2002,Litinski_2019_GoS,fowler2019lowoverheadquantumcomputation} and bicycle bivariate code~\cite{Bravyi_2024}, have efficient implementations of fault-tolerant logical Clifford gates~\cite{yoder2025tourgrossmodularquantum, Horsman_2012,serraperalta2025decodingtransversalcliffordgates} but logical non-Clifford gates, which are required for universality, are performed by consuming \textit{magic states} through an injection circuit composed of Clifford gates~\cite{Bravyi_2005,fowler2019lowoverheadquantumcomputation}.
For example, the magic state $\ket{T}=\left(\ket{0}+e^{i\pi/4}\ket{1}\right)/\sqrt{2}$ can be consumed to implement the non-Clifford $T$ gate, a $\pi/4$ rotation around the $Z$-axis of the Bloch sphere.

Errors in the preparation of magic states translate directly to errors on the gates implemented through the injection channel.
Useful quantum computation typically requires on the order of $10^8$ or more $T$ gates~\cite{Lee_2021,beverland_resource_estimation,schrottenloher2026optimizedpointadditioncircuits}, demanding magic state error rates below $10^{-10}$ to avoid spoiling the final result.
In practice, however, magic states can be realistically prepared with error rates $p_\mathrm{in}$, similar to the physical qubit error rate, in the order of $\approx 10^{-3}$~\cite{Li_2015,MS_injection_Surface_code}.
\textit{Magic state distillation} is a well studied approach to close this error rate gap~\cite{Bravyi_2005,Bravyi_2012,campbell_magic-state_2012,jones_multilevel_2013,Campbell_2017,Haah_2018,Litinski_2019,londe2026localdistillationreedmuller}.
A $T$-state distillation protocol consumes $n$ low-quality $\ket{T}$ states with input error $p_\mathrm{in}$ to produce $k$ high-quality $\ket{T}$ states with output error $p_\mathrm{out} = C p_\mathrm{in}^d$ at first order, for some distance $d>1$ and prefactor $C>0$.
Similarly, $\CCZ$-state distillation protocols~\cite{Haah_2018}, often called synthillation protocols~\cite{Campbell_2017}, consume $\ket{T}$ states to produce high-quality $\ket{\CCZ}=\CCZ\ket{+++}$ states, which are injected to perform Toffoli gates.
Recent work on magic-state cultivation has emerged as a promising, resource-efficient alternative to magic-state distillation~\cite{gidney2024magicstatecultivation,itogawa_efficient_2025,google_cultivation_experiment}.
However, current protocols hit a floor at $p_\mathrm{out}\approx 10^{-7}$ for input error rates of $p_\mathrm{in}\approx 10^{-3}$.
Therefore, in recent resource estimation, magic state cultivation is used as a first step to produce less-noisy magic states, then processed in a final distillation phase~\cite{gidney2025factor2048bitrsa,cain_shors_2026}.

\begin{figure*}[ht]
    \centering
    \includegraphics[width=\linewidth]{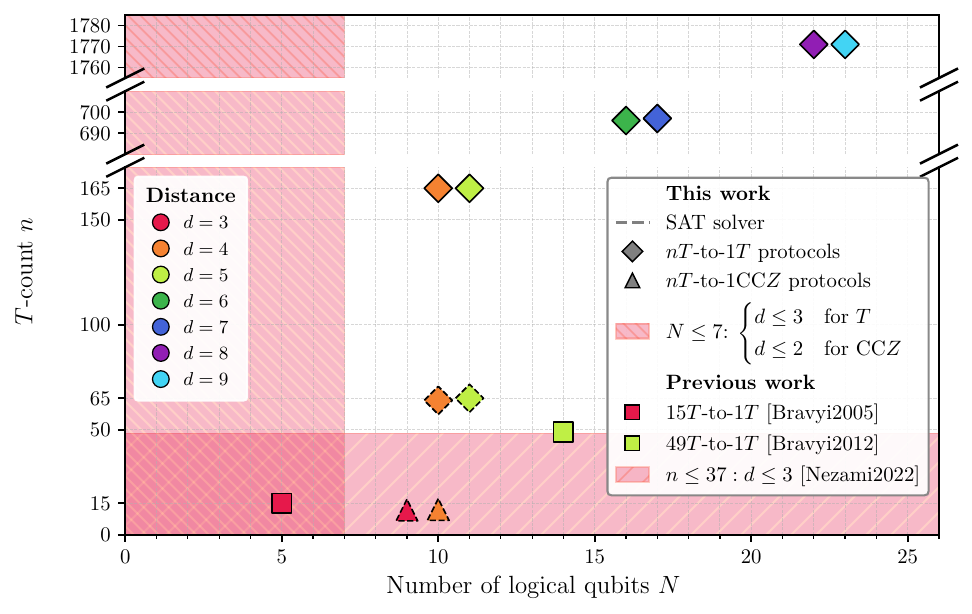}
   \caption{\textbf{Summary of the exploration of compact magic-state distillation protocols.}
    Each dot represents a distillation protocol implemented on $N$ logical qubits, using standard $Z$-pauli product rotations circuits, that consumes $n$ noisy magic states and outputs a single magic state with output error scaling as $p_{\mathrm{out}} \propto p_{\mathrm{in}}^d$, where $d$ is the distillation distance.
    Note that we here exclude mid-circuit measurements and re-initialization.
    Previously known protocols include the 15-to-1 protocol of Ref.~\cite{Bravyi_2005}, referred to as [Bravyi2005] on the figure, implemented on $N=5$ qubits, and the 49-to-1 protocol of Ref.~\cite{Bravyi_2012}, [Bravyi2012], implemented on $N=14$ qubits.
    Ref.~\cite{Nezami_2022}, referred to as [Nezami2022] on the figure, proved a no-go result showing that any protocol with $n+k\leq 38$ satisfies $d\leq 3$.
    In this work, we introduce several new families of protocols.
    First, using a SAT-based search, we find a 64-to-1 protocol on $N=10$ qubits with output error rate $p_{\mathrm{out}} = 495\,p_{\mathrm{in}}^4$, and a 65-to-1 protocol on $N=11$ qubits with $p_{\mathrm{out}} = 7947\,p_{\mathrm{in}}^5$.
    Second, we construct a canonical family of protocols, a few examples of which are shown here for distances up to $9$.
    Third, we present $\ket{\mathrm{CC}Z}$ distillation protocols achieving distances $d=3$ and $d=4$ with $n=47$ and $n=48$, respectively. For comparison with $\ket{T}$ state distillation protocols, the corresponding $T$-counts are divided by four, since a Toffoli gate can be implemented using either one $\ket{\mathrm{CC}Z}$ state or four $\ket{T}$ states.
    Finally, we prove that no $\mathrm{CC}Z$ (resp. $T$) distillation protocol achieving distance $d\geq 3$ (resp. $d\geq 4$) can be implemented on $N=7$ or less logical qubits.
    } 
    \label{fig:Main_results}
\end{figure*}

Distillation protocols are typically constructed from $[[n,k,d]]$ quantum codes, commonly referred to as triorthogonal codes~\cite{Bravyi_2012}, that support a transversal implementation of the $T$ gate.
The most widely used example is the quantum Reed--Muller $[[15,1,3]]$ code, which leads to a 15 noisy $\ket{T}$ to 1 cleaner $\ket{T}$ distillation (15-to-1 protocol) with an output error $p_\mathrm{out}=35p_\mathrm{in}^3$~\cite{Bravyi_2005}.
For higher distances, numerical exploration has identified the $[[49,1,5]]$ code~\cite{Bravyi_2012} yielding a 49-to-1 protocol characterized by an error rate reduction of $p_\mathrm{out}=1411p_\mathrm{in}^5$.
In general, distillation protocols based on a triorthogonal code consuming $n$ noisy $\ket{T}$ states and involving the measurement of $m_x$ $X$-type stabilizers can be implemented on a circuit of only $N=k+m_x$ qubits, at the cost of a circuit depth of $n$~\cite{Haah_2018,Litinski_2019,Litinski_2019_GoS}.
This reduction is made possible by the fact that distillation protocols suffer only from $Z$-errors on $\ket{T}$ states, allowing the qubits that would otherwise store the $Z$-check values to be eliminated.
This construction yields protocols of the following form: qubit preparation in $\ket{+}$, a sequence of commuting $Z$ Pauli-product rotations, and finally, qubit measurements in $X$-basis where the number of qubits corresponds to the number of rows in the original weakly triorthogonal matrix.
Concretely, construction of the 15-to-1 protocol can be implemented using only $5$ qubits, drastically reducing the resources needed.
Following a similar logic, one can implement the 49-to-1 scheme using $14$ qubits.

In this scope, finding good distillation protocols amounts to identifying codes that simultaneously achieve low output error, dictated primarily by $d$ and secondarily by $C$, small qubit footprint $N$, and low $T$-count $n$.
Many works attempt to bound the code distance $d$ for different values of $n$ and $k$~\cite{Shi_2024,baldelli2026constructingdecodingquantumtriorthogonal}, notably showing that $d\leq3$ for $n+k<38$~\cite{Nezami_2022}.
However, to our knowledge, no bound related to the number of qubits $N$ has been derived.

This work explores the set of achievable values of distance $d$, $T$-count $n$, and prefactor $C$ for fixed numbers of qubits $N$.
We first frame distillation protocols as the implementation of a fault-tolerant logical $T$ gate on the repetition code, which allows us to derive a natural characterization of distillation protocols in terms of the number of qubits $N$, bypassing the standard description via $X$-type stabilizers.
Using this framework, we construct a canonical family of distillation protocols analytically, yielding one concrete protocol for each value of qubit count $N$ with explicit $T$-count and distance.

Encoding our framework as a \textit{Boolean satisfiability} (SAT) problem suggests that this family may minimize the number of required qubits to achieve a given distance $d$.
Specifically, for every $d\leq 5$, the SAT solver has not been able to find working protocols using fewer qubits than the canonical family.
Relaxing the SAT instances by restricting to symmetric subfamilies of protocols, we explore distillation schemes up to $N=23$ logical qubits, uncovering new constructions along the way.
Most notably, we find a $64$-to-$1$ distance-$4$ protocol on $N=10$ qubits and a $65$-to-$1$ distance-$5$ protocol on $N=11$ qubits, the lowest known qubit footprints for these distances.
Additionally, we establish no-go theorems showing that distillation at distance $d\geq4$ requires at least $N\geq 8$ qubits.
Extending the SAT approach to $\CCZ$-state distillation, we prove that any distance-$3$ protocol requires at least $8$ qubits, and report the first such protocol, a $47T$-to-$1\CCZ$ scheme on $N=9$ qubits with output error rate $236\,p_\mathrm{in}^3$.
We further propose a $48T$-to-$1\CCZ$ distance-$4$ protocol on $N=10$ qubits, reducing the qubit footprint by $7$ compared to the previously most compact known construction~\cite{Jones_2013_64_to_2,Haah_2018}.
The results of the exploration of this landscape are gathered in \Cref{fig:Main_results}.

Finally, we attempt to further reduce the qubit cost of distillation factories by recycling qubits~\cite{ibm_qbit_recycling}.
Using mid-circuit measurements and qubit reinitialization, we bring the implementation of the $15$-to-$1$ protocol down from $N=5$ to $4$ qubits, and that of the $49$-to-$1$ protocol from $N=14$ to $5$ qubits.
To our knowledge, this is the first time a distance-$5$ protocol admits an implementation on $5$ qubits.

This paper is organized as follows.
In \Cref{sec:prelim}, we recall the necessary background on fault-tolerant quantum computing and magic state injection, before describing a first construction for a logical $T$ gate on the repetition code.
In \Cref{sec:FT_T_on_rep_code}, we propose two constructions of fault-tolerant logical $T$ gate on the repetition code and link these constructions to distillation schemes. 
In \Cref{sec:canonical}, we derive an analytical family of distillation protocols with an explicit link between the qubit count and the distance. In \Cref{sec:exploring}, we use
a SAT encoding to explore distillation and synthillation circuits, uncovering no-go theorems and discovering new schemes. 
In \Cref{sec:recycling}, we compress the qubit footprint of distillation protocols through qubit recycling.
Finally, we summarize all the discovered distillation and synthillation results in \Cref{sec:results} and conclude in \Cref{sec:conclusion}.

%% file: Preliminaries_v2.tex
\section{Preliminaries}
\label{sec:prelim}

\subsection{Fault-tolerant quantum computing}

Quantum error-correcting codes protect logical information by encoding it redundantly across many physical qubits.
Prominent examples include Calderbank--Shor--Steane (CSS) codes, such as the surface code and bivariate bicycle (BB) codes.
A key feature of CSS codes is that they admit a transversal implementation of the CNOT gate.
More generally, resource-efficient implementations of Clifford gates, which belong to the second level of the Clifford hierarchy, are available on such codes~\cite{Horsman_2012,yoder2025tourgrossmodularquantum,Bravyi_2024,fowler2019lowoverheadquantumcomputation,Litinski_2019_GoS}.

\begin{definition}[Pauli group]
	The \emph{Pauli group} $\mathcal{P}$ is the group of $n$-qubit operators of the form $e^{i\phi} P_1 \otimes \cdots \otimes P_n$, where each $P_k \in \{\mathbb{I}, X, Y, Z\}$ and $\phi \in \{0, \pi/2, \pi, 3\pi/2\}$.
\end{definition}

\begin{definition}[Clifford hierarchy]
The $j$-th level of the Clifford hierarchy ($j \geq 1$) is defined as
\begin{equation}
    \mathfrak{C}_j = \bigl\{ U \in \mathrm{U}(2^n) \;\big|\; U P U^\dagger \in \mathfrak{C}_{j-1},\; \forall P \in \mathcal{P} \bigr\},
\end{equation}
with $\mathfrak{C}_1 = \mathcal{P}$.
\end{definition}

Circuits composed entirely of Clifford gates can be efficiently simulated classically by the Gottesman--Knill theorem, as they map stabilizer states to stabilizer states~\cite{Gottesman1998,aaronson2004}.
Universal quantum computation therefore requires at least one non-Clifford gate, i.e., a gate belonging to $\mathfrak{C}_j$ for $j \geq 3$.
A canonical example is the universal gate set $\{H, \mathrm{CNOT}, T\}$, where the $T$ gate lies in $\mathfrak{C}_3$.

However, the Eastin--Knill theorem prohibits any QECC from implementing a universal gate set transversally~\cite{eastin2009}.
In particular, the codes mentioned above do not admit a transversal $T$ gate.
The standard approach to circumvent this restriction is \emph{magic state distillation and injection} where high-fidelity resource states, called magic states, are prepared offline and consumed via gate teleportation to perform the desired non-Clifford operation at the logical level.

\subsection{Non-Clifford operations via state injection}

Gate teleportation is a circuit gadget that implements a unitary $U$ on a target state $\ket{\psi}$ by applying $U$ to one half of an entangled pair and then teleporting $\ket{\psi}$ onto it.
We first recall the canonical quantum teleportation protocol in the circuit model of quantum computation

\begin{center}
\begin{quantikz}
\lstick{$\ket{\psi}$} &  \phase{X}   & \meter{}  & \cwbend{1}\setwiretype{n} & \\
\lstick{$\ket{+}$}    & \gate{Z}\wire[u]{q}  & \qw       & \targ{}                  & \qw \rstick{$\ket{\psi}$}
\end{quantikz}.
\end{center}

The first qubit is entangled with the ancilla via a CNOT gate, which we equivalently wrote as a $Z$ gate controlled in the $X$ basis as one can easily check that $\ketbra{0}{0}\mathbb{I}+\ketbra{1}{1}X=\ketbra{+}{+}\mathbb{I}+\ketbra{-}{-}Z$.
Then, the first qubit is measured in the computational basis.
The classical outcome controls an $X$ correction on the ancilla, recovering $\ket{\psi}$.
Here, we consider the states $\ket{\psi}$ and $\ket{+}$ to be encoded in an \textit{inner} QECC, protecting against Clifford errors occurring during the teleportation circuit.

To teleport $U\ket{\psi}$, one prepares the ancilla in $U\ket{+}$ and conjugates the correction operators by $U$.
This quantum circuit reads

\begin{center}
\begin{quantikz}
\lstick{$\ket{\psi}$} & \qw       & \phase{X}                      & \meter{} & \cwbend{1}\setwiretype{n} & \\
\lstick{$\ket{+}$}    & \gate{U}  & \gate{UZU^{\dagger}}\wire[u]{q} & \qw      & \gate{UXU^{\dagger}}      & \qw \rstick{$U\ket{\psi}$}
\end{quantikz}.
\end{center}

To keep the output on the first wire, one can prepend a SWAP, yielding the in-place gadget

\begin{center}
\begin{quantikz}
\lstick{$\ket{\psi}$} & \phase{X}                              & \ctrl{1} & \qw      & \gate{UXU^{\dagger}X}         & \qw \rstick{$U\ket{\psi}$} \\
\lstick{$\ket{U}$}    & \gate{ZUZ U^{\dagger}}\wire[u]{q}    & \targ{}  & \meter{} & \cwbend{-1}\setwiretype{n}    &
\end{quantikz}
\end{center}
where we write $\ket{U} \coloneqq U\ket{+}$ for the resource state.
When $U$ implements a non-Clifford gate, $\ket{U}$ is called a \emph{magic state}.
If $U$ is diagonal in the $Z$-basis, the controlled-$ZUZU^\dagger$ gate on the first wire acts trivially and can be omitted.
Furthermore, when $U \in \mathfrak{C}_3$, the classically controlled correction $UXU^\dagger X$ belongs to $\mathfrak{C}_2$.
Therefore, the entire teleportation gadget can be implemented fault-tolerantly, as it relies only on Clifford operations plus the preparation of $\ket{U}$. For the $T$ gate specifically, the circuit reduces to
\begin{equation}\label{circ:Tinject}
    \begin{quantikz}
\ket{\psi} \ & \ctrl{1} & \gate{S} & \ T\ket{\psi} \\
\ket{T} \ & \targ{} & \meter{} \wire[u][1]{c}\\
\end{quantikz}
\end{equation}

This circuit is commonly known as a $T$-state injection circuit~\cite{Campbell_2017_Terhal}.
The input magic state has to be of high quality, as imperfections will propagate as errors through the circuit to the final state. 
Indeed, a $Z$-type error on $\ket{T}$ produces the state $Z T \ket{\psi}$.
Meanwhile, an $X$-type error leads to an erroneous state
\begin{equation}
    ST\ket{\psi} = \frac{e^{i\pi/4}}{\sqrt{2}}\left(\mathbb{I} - iZ\right)T\ket{\psi}.
\end{equation}
Interestingly, this decomposition highlights that only $Z$-type errors occur on the output state.
Indeed, measuring an $X$-type stabilizer of the QECC encoding $\ket{\psi}$ at the logical level projects the output to either $T\ket{\psi}$ or $ZT\ket{\psi}$, each with probability $1/2$, depending on the measurement outcome.

Therefore, assuming that all Clifford operations are noiseless, as protected by the QECC in the injection scheme, only $Z$-type errors on the output state need to be corrected. 
This observation naturally motivates the design of protocols that use $n$ injection schemes of noisy $T$-states to prepare $k<n$ $T$-states specifically protected against $Z$-type errors.
Note that this logic generalizes to any unitary diagonal in the $Z$-basis, such as $\CCZ$ or $\sqrt{T}$.

\section{Fault-tolerant logical $T$ gate on the repetition code}\label{sec:FT_T_on_rep_code}

Here, we show that a protocol for $T$-state distillation can be implemented with a pair of complementary codes: an inner code which provides fault-tolerant Clifford gates and $T$ gates, implemented by magic state injections, affected by $Z$-errors only and an outer classical code used as a distillation code.
The basic idea is to prepare a logical $\ket{+}$ state of the distillation code from qubits encoded in the inner code. A logical $T$-state of the outer code is then obtained by performing a $T$ gate on the distillation code using multiple $T$-state injections of the inner code. The decoding circuits and the measurements of the distillation code stabilizers are finally used to map the $T$-state on one of the inner code qubits and identify $Z$-errors.
Cases with no detected errors are kept and provide one $T$ state encoded in the inner code that is cleaner than the initially injected $T$-states.

\subsection{Logical $T$ gate on the repetition code}
\label{sec:LogicalT_repCode}

We propose a construction of a logical $T$ gate on the simplest classical code: the repetition code.
The $[N,1,N]$ repetition code encodes a single logical qubit into $N$ physical qubits, and has distance $N$ against a single type of error.
Throughout this section, we consider that each physical qubit is itself encoded in an inner QECC (e.g., the surface code) to protect against the errors that may arise during the Clifford operations of the distillation protocol. 
We remind the reader that when $T$ gates are performed by state injection on this inner code, they are only affected by $Z$-errors

We focus on a repetition code protecting against $Z$-type errors, namely the standard repetition code recast in the $X$-basis.

\begin{definition}[Phase-flip repetition code]
    The phase-flip repetition code on $N$ physical qubits is defined by the logical codewords
    \begin{equation}
    \begin{aligned}
        \ket{+}_L &= \ket{+}^{\otimes N}, \\
        \ket{-}_L &= \ket{-}^{\otimes N}.
        \end{aligned}
    \end{equation}
\end{definition}

The code has stabilizer generators $\{X_i X_{i+1}\}_{i=1}^{N-1}$, logical $\bar{X} = X_1$ and logical $\bar{Z} = Z_1 \cdots Z_N$.

To protect a state in the repetition code, we encode it using a circuit $\mathcal{E}$ and decode it with $\mathcal{E}^\dagger$.

\begin{definition}[Encoding and decoding circuits]
    Let $\ket{\psi} = \alpha\ket{+} + \beta\ket{-}$ be an arbitrary single-qubit state.
    The \emph{encoding circuit} $\mathcal{E}$ is the Clifford circuit mapping
    \begin{equation}
        \ket{\psi} \otimes \ket{+}^{\otimes N-1} \;\longmapsto\; \ket{\psi}_L = \alpha\ket{+}_L + \beta\ket{-}_L.
    \end{equation}
	Equivalently, $\mathcal{E}$ maps the stabilizers $X_i \mapsto X_i X_{i+1}$ for $i = 1, \ldots, N-1$, and the logical operators $X_1 \mapsto \bar{X}$, $Z_1 \mapsto \bar{Z}$.
    The \emph{decoding circuit} is $\mathcal{E}^\dagger$, mapping $\ket{\psi}_L$ back to $\ket{\psi} \otimes \ket{+}^{\otimes N-1}$.
\end{definition}

The encoding circuit is implemented by a descending CNOT cascade

\begin{equation}
\begin{quantikz}
& \targ{}  && \ \ldots \ && \\
& \ctrl{-1} & \targ{} & \ \ldots \ &&\\
&& \ctrl{-1} & \ \ldots \ &&\\
\setwiretype{n} & \vdots && \ \ldots \ \\
&&& \ \ldots \ & \targ{}&\\
&&& \ \ldots \ & \ctrl{-1} &
\end{quantikz} 
=: \begin{quantikz}
 && \gate[3]{\mathcal{E}} &&\\
 \setwiretype{n} & \vdots && \vdots &\\
 && &&\\
\end{quantikz}
\end{equation}

and decoding by the reverse, ascending CNOT cascade

\begin{equation}\label{circ:non-FT_T_gate}
    \begin{quantikz}
&  &\ \ldots\ && \targ{}  \\
&&\ \ldots\ &\targ{} & \ctrl{-1}  \\
&&\ \ldots\ & \ctrl{-1}&\\
\setwiretype{n} & \vdots && \ \ldots \ \\
& \targ{} & \ \ldots \ && \\
& \ctrl{-1} & \ \ldots \ && 
\end{quantikz} =: \begin{quantikz}
 && \gate[3]{\mathcal{E}^\dagger} &&\\
 \setwiretype{n} & \vdots && \vdots &\\
 && &&\\
\end{quantikz}
\end{equation}

A naive implementation of a logical $T$ gate on $\ket{\psi}_L$ starts with decoding using $\mathcal{E}^\dagger$, applying $T$ on the first qubit, and re-encoding with $\mathcal{E}$. 
We label this operation $T_L$
\begin{equation}
    \begin{quantikz}
        & \gate[4]{\mathcal{E}^\dagger} & \gate{T} & \gate[4]{\mathcal{E}} & \\
        &&&&\\
        \setwiretype{n} & & \vdots & & \\
        & & & & \\
    \end{quantikz} =: \begin{quantikz}
        & \gate[3]{T_L} & \\
        \setwiretype{n} &  & \\
        & &  \\
    \end{quantikz}
\end{equation}
However, this naive implementation strategy is not fault-tolerant.
Indeed, any error on the $T$ gate propagates through the encoding circuit to all $N$ qubits and is therefore undetectable by the repetition code.
We thus look for a genuinely fault-tolerant implementation of $T_L$, in which errors remain local and correctable using the repetition code.

\subsection{A \texorpdfstring{$\{T,\CS^\dagger,\CCZ\}$}{T,CS,CCZ} implementation}

In the previously proposed naive implementation of the logical $T$ gate, errors are propagated due to the encoding circuit $\mathcal{E}$.
To find a fault-tolerant implementation, we use the fact that $\mathcal{E}^\dagger \mathcal{E} = \mathbb{I}$.
Commuting $T$ on the first qubit to the left with the decoding circuit, therefore, produces a new circuit where the encoding and decoding cancel out.
In \Cref{App:Commutations_relations}, we give details on how this commutation yields the circuit
\begin{equation}
 \begin{quantikz}
    & \gate{T}& \gate[3]{\prod_{i,j} (\CS^{\dagger})_{ij}} & \gate[3]{\prod_{i,j,k} (\CCZ)_{ijk}} &&\\
 \setwiretype{n} & \vdots &&& \vdots &\\
 & \gate{T} &&& &\\ 
\end{quantikz}
\label{circ:commutations}
\end{equation}
composed solely of $T$ gates, $\CS^\dagger$ gates, and $\CCZ$ gates, defined as
\begin{equation}
 \begin{aligned}
T = \begin{pmatrix} 1 & 0 \\ 0 & e^{i\pi/4} \end{pmatrix}, \quad
\CS^\dagger = \text{diag}(1, 1, 1, -i) \\[10pt]
\CCZ = \text{diag}(1, 1, 1, 1, 1, 1, 1, -1).
\end{aligned}   
\end{equation}

In words, this circuit is made up of one T gate for every qubit, one $\CS^\dagger$ for every pair of qubits, and a $\CCZ$ for every triplet.
Interestingly, these gates can be exactly decomposed into Clifford+T circuits: using 3 T gates for $\CS^\dagger$ and 7 T gates for $\CCZ$~\cite{Nielsen2010,jones_low-overhead_2013}.
Hence, only $Z$ errors affect these gates, making the phase-flip repetition code still relevant for detecting faulty gates.

Considering such circuits, with $N\leq 3$, errors will propagate to logical errors as at least one gate acts on every qubit.
However, for $N \geq 4$, we obtain a fault-tolerant implementation of $T_L$.
Indeed, with $N=\{4,5,6\}$ qubits, two gate error patterns are required to affect all $N$ qubits, e.g.~for $N=4$, a $\CCZ$ on three qubits and a $T$ gate on the remaining one. 
Hence, for these qubit footprints, we have a fault-tolerant implementation of the logical $T$ gate that is able to detect any error on one of the $T$, $\CS^\dagger$ or $\CCZ$ gates.
Extending this logic to $N\geq7$, the repetition code detects weight-two error patterns as well, offering a circuit distance $d=3$ protection against erroneous magic states.
Interestingly, these circuits directly link the number of qubits $N$ with the distance $d$, and with the $T$-count n, scaling as
\begin{equation}
    n = N + 3\binom{N}{2}+7\binom{N}{3}.
\end{equation}

In the distillation picture, this approach is equivalent to distillation protocols based on $n$-to-$1$ protocols with distance $d$ implemented on $N$ qubits.
In particular, our construction on $N=4$ qubits leads to a $50$-to-$1$ distillation protocol of distance $2$, while on $N=7$ qubits, we obtain a $315$-to-$1$ protocol of distance 3.
Comparing to the well-known $15$-to-$1$ Reed-Muller protocol of distance $3$ that can be implemented on $N=5$ qubits shows the need to find alternative circuits of the logical $T$ gate with improved $T$-count and distances.

\subsection{Pauli product rotations implementations}
\label{sec:alaVictor}

To go beyond this costly implementation, we focus on a more practical approach, attempting to implement $T_L$ using the gate set of $\pi/8$ Pauli product rotations on $N$ qubits, that we label $P_{\pi/8}$.
Such rotations are defined as
\begin{equation}
	P_{\pi/8} = \exp(-i\frac{\pi}{8}P) = \cos(\frac{\pi}{8})\mathbb{I}-i\sin(\frac{\pi}{8})P
	\label{eq:pi8}
\end{equation}
where $P$ is a Pauli string $P\in\{I,Z\}^{\otimes N}$.

This family of rotations is interesting as each gate consumes exactly one magic state $\ket{T}$ via teleportation followed by Clifford corrections, as depicted in \Cref{fig:Pcirc}.
Following the argument in \Cref{sec:prelim}, these gates propagate only $Z$ errors on every qubit where the element of $P$ is $Z$.
Moreover, as seen previously, the teleportation protocol can be implemented fault-tolerantly, provided a perfect magic-state, as these rotations are members of $\mathfrak{C}^2$, the second level of the Clifford hierarchy.
Furthermore, these gates can be implemented with clearly defined overhead on any QECC supporting lattice surgery techniques, such as the surface code~\cite{Litinski_2019_GoS} or the BB code~\cite{ibm_qbit_recycling}.
On these two codes, entire distillation factories have been proposed using Pauli product rotations.

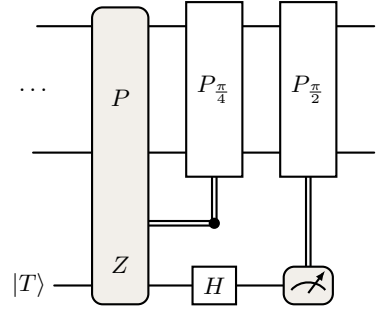
\begin{figure}
    \begin{quantikz}
    	\ & \gate[5, disable auto height, style={fill=measurebg, rounded corners}]
        {\begin{array}{c}
     \\  \\ P \\ \\  \\ \\ \\ \\ Z
    \end{array}} & \gate[3]{P_{\frac{\pi}{4}}} &\gate[3]{P_{\frac{\pi}{2}}}&  \\
     \ldots\setwiretype{n} &          &              &   & \\
                         &           &              &   &\\
    \setwiretype{n}&   &\cwbend{-1}     &\\
	\ket{T} \ &   & \gate{H}&\meter[style={fill=measurebg, rounded corners}]{} \wire[u][2]{c}
    \end{quantikz}
\label{circ:multi_qubit_T}
\caption{Gate teleportation and Clifford correction circuit to perform a $P_{\pi/8}$ rotation~\cite{Litinski_2019_GoS}. First a $\ket{T}$ state is injected and the entire system is jointly measured in $P\otimes Z$. A first Pauli correction consisting of a $P_{\pi/4}$ gate is applied to the $N$ qubits if the outcome of the measurement is $-1$. The injected qubit undergoes a Hadamard operation before being measured in the Z basis. Another Clifford correction, namely $P_{\pi/2}$, is finally performed on the first N
	qubits.}
\label{fig:Pcirc}
\end{figure}

Equipped with this new gate set, we attempt to reconstruct $T_L$.
Specifically, we aim to find a circuit composed of a subset of $n$ Pauli product rotations, $\{P^k_{\pi/8}\}_{k=1}^n$, whose joint action implements the logical $T$ gate on the repetition code.
To each Pauli string $P^k$ of the $k$-th rotation, we associate a binary vector $\alpha^k \in \mathbb{F}_2^N$ where the $i$-th element is defined as
\begin{equation}\label{eq:gate_sequence_P_from_alpha}
    \alpha^k_i = \begin{cases} 0 & \text{if } (P^k)_i = I, \\ 1 & \text{if } (P^k)_i = Z. \end{cases}
\end{equation}
The full circuit can thus be represented as a $N\times n$ binary matrix $\mathcal{C}$, whose rows correspond to qubits and whose columns map to Pauli product rotations, with $\mathcal{C}_{ij}=\alpha^i_j$. For simplicity, we henceforth refer to $\mathcal{C}$ as both the matrix and the circuit itself, as $\mathcal{C} = \prod_{k=1}^n P_{\pi/8}^k$.

Note that a $P_{\pi/8}$ gate is nothing else but a logical $T$ gate on a smaller repetition code.
Indeed, conjugating $T=\cos(\frac{\pi}{8})\mathbb{I}-i\sin(\frac{\pi}{8})Z$ by an ascending and a descending cascade of CNOT that maps $Z$ onto $P$ gives $P_{\pi/8}$.
From the construction \Cref{circ:commutations}, we know that a logical $T$ gate decomposes to a circuit with one $T$ gate on every qubit, one $\CS^\dagger$ gate on every pair of qubits and a $\CCZ$ between every triplets of qubits.
Thus, in terms of $\mathcal{C}$, each column $\alpha^k$ adds a $T$ gate on each qubit $i$ where $\alpha_i^k=1$, a $\CS^\dagger$ gate on each pair of qubits $i,j$ where $\alpha_{i}^k\alpha_j^k=1$, and a $\CCZ$ gate on triplets on each triplet of qubits $i,j,l$ where $\alpha_{i}^k\alpha_j^k\alpha_{l}^k=1$.
Therefore, we can impose constraints on $\mathcal{C}$ to recover $T_L$ as,
\begin{subequations}\label{eq:strong_cons_T}
    \begin{align}
        & \forall i\in[\![1,N]\!], \sum_k \alpha_i^k \equiv 1 \mod 8, \\
        & \forall i<j \in[\![1,N]\!]^2, \sum_k \alpha_{i}^k\alpha_j^k \equiv 1 \mod 4, \\
        & \forall i<j<l\in[\![1,N]\!]^3, \sum_k \alpha_i^k\alpha_j^k\alpha_l^k \equiv 1 \mod 2. 
    \end{align}
\end{subequations}
The modulo arguments come from the fact that applying a $T$ gate 8 times on an individual qubit is equivalent to applying nothing, as $T^8 = \mathbb{I}$. Likewise, $(\CS^{\dagger})^4 = \mathbb{I}^{\otimes 2}$ and $(\CCZ)^2 = \mathbb{I}^{\otimes 3}$.

Considering that Clifford corrections are free, these constraints can be further simplified.
Since $T^2 = S$ and $(\CS^{\dagger})^2 = CZ$ are both Clifford gates, any even number in the first two conditions reduces to a Clifford operation.
Therefore, the constraints can be all set to modulo 2:
\begin{subequations}\label{eq:cons_T}
    \begin{align}
        & \forall i\in[\![1,N]\!], \sum_k \alpha_i^k \equiv 1 \mod 2, \\
        & \forall i<j \in[\![1,N]\!]^2, \sum_k \alpha_i^k\alpha_j^k \equiv 1 \mod 2, \\
        & \forall i<j<l\in[\![1,N]\!]^3, \sum_k \alpha_i^k\alpha_j^k\alpha_l^k \equiv 1 \mod 2. 
    \end{align}
\end{subequations}

While a circuit fulfilling the conditions \Cref{eq:cons_T} implements a logical $T$ gate up to Clifford corrections, this implementation may not be fault-tolerant. 
What we denote as the distance $d$ of a protocol $\mathcal{C}$ corresponds to the minimum number of faulty gates which, together, lead to a logical error in the repetition code.
A logical error here is characterized by an odd number of $Z$ errors on every qubit composing the repetition code.
As a Pauli product rotation propagates $Z$ errors where $(P^k)_i=Z$~\cite{Litinski_2019}, or $\alpha^k_i=1$, a set $K$ of faulty-gates leads to a logical error if, for all $i\in [1,N]$, $\sum_{k\in K}\alpha_i^k=1 \mod 2$. 
Formally, the distance is therefore given by
\begin{equation} \label{eq:dparam}
    d = \min |K| \text{ s.t. }  \sum_{k\in K} \alpha_i^k \equiv \left(\begin{array}{c}
         1   \\
         \vdots\\
         1 
    \end{array}\right) \ \mod 2.
\end{equation}

In \Cref{sec:canonical} and \Cref{sec:sat}, we provide ways to find circuits $\mathcal{C}$, on $N$ qubits, implementing the fault-tolerant logical $T$ gate on the repetition code with a guaranteed minimum distance $d$.

\subsection{From fault-tolerant logical $T$ gate to distillation protocols}
\label{sec:logicalT-to-distill}

Given a circuit $\mathcal{C}$ implementing a fault-tolerant logical $T$ gate on the repetition code with distance $d$, we show how to derive a magic state distillation circuit $\mathcal{G}$ that outputs a high-fidelity $\ket{T}$ state on the first qubit, with output error $\propto p_\mathrm{in}^d$, where $p_\mathrm{in}$ is the physical error rate of the injected $\ket{T}$ states, and with the same $T$-count.
The distillation process proceeds in two steps. First, applying $\mathcal{C}$ to the logical state $\ket{+}_L = \mathcal{E}\ket{+}^{\otimes N}$ of the repetition code produces $\ket{T}_L$.
Then, decoding $\ket{T}_L$ via $\mathcal{E}^\dagger$ leads to a high quality $\ket{T}$ state on the first qubit, with error $\propto p_\mathrm{in}^d$, and with the other $N-1$ qubits serving as checks.
The distillation circuit $\mathcal{G}$ is therefore obtained by commuting $\mathcal{C}$ past $\mathcal{E}^\dagger$. 
Note that when considering fault-tolerant logical $T$ gate up to Clifford correction, these corrections must also be commuted past $\mathcal{E}^\dagger$, resulting in a modified correction remaining in the Clifford group.

As explained in \Cref{App:V_to_L}, $\mathcal{G}$ is also composed of $n$ Pauli product rotations $\{\tilde P^k\}_{k=1}^n$. Each rotation $\tilde P^k$ is associated with a binary vector $\beta^k$, which are columns of $\mathcal{G}$ and whose $i$-th element is
\begin{equation}\label{eq1:from_alpha_to_beta}
\begin{aligned}
\beta^k_1 &= \alpha^k_1\\
\forall i \in [\![2,N]\!],~ \beta^k_i &= \alpha^k_i \oplus\alpha^k_{i-1}
\end{aligned}
\end{equation}
where $\alpha^k$ is the $k$-th column of $\mathcal{C}$.
Therefore, a distillation circuit $\mathcal{G}$ computed from $\mathcal{C}$ of size $ N\times n$ has the same $T$-count $n$. 
Note that when Clifford corrections are required, the depth of the distillation process is lower bounded by the $T$-count.
Indeed, if the Clifford correction happens to be a $P_{\theta}$ gate for some $P$ such that $P_{\pi/8}$ is already in $\mathcal{C}$, then it can be absorbed by adapting the injection circuit of the former $P_{\pi/8}$ gate~\cite{Litinski_2019}.
However, if $P$ is not the support of one of the Pauli product $\pi/8$ rotations in $\mathcal{C}$, then, some additional Clifford gates are necessary before measuring the checks.

With the same logic, we can rewrite the constraints \Cref{eq:cons_T} that $\mathcal{G}$ must fulfill to be a $T$-distillation circuit of the form $nT$-to-$1T$:
\begin{subequations}\label{eq:triorth_cons}
    \begin{align}
        &  \sum_k \beta_1^k \equiv 1 \mod 2, \\
        & \forall i\in[\![2,N]\!], \sum_k \beta_i^k \equiv 0 \mod 2, \\
        & \forall i<j \in[\![1,N]\!]^2, \sum_k \beta_i^k\beta_j^k \equiv 0 \mod 2, \\
        & \forall i<j<l\in[\![1,N]\!]^3, \sum_k \beta_i^k\beta_j^k\beta_l^k \equiv 0 \mod 2. 
    \end{align}
\end{subequations}
The distance $d$ of our distillation protocols corresponds to the amount of error required to produce a $Z$ error on the first qubit without producing any detectable $Z$ errors on the remaining $N-1$ qubits, serving as check qubits. Formally, the distance of the protocol implemented by $\mathcal{G}$ is given by
\begin{equation} \label{eq:dparam_triorth_cons}
    d = \min |K| \text{ s.t. } \sum_{k\in K} \beta^k \equiv \left(\begin{array}{c}
         1  \\
         0  \\
         \vdots\\
         0 
    \end{array}\right) \mod 2.
\end{equation}

These expressions recover the known construction of magic state distillation based on $[[n,1,d]]$-triorthogonal codes~\cite{Bravyi_2012}.
Such codes are built from punctured triorthogonal matrices, which are matrices on $\mathbb{F}_2^{N\times n}$ that fulfill the constraints \Cref{eq:triorth_cons}~\cite{Bravyi_2012,Litinski_2019_GoS}.
Distillation protocols based on triorthogonal codes have later been compressed~\cite{Haah_2018, Litinski_2019_GoS} to be implemented using $m_X$ qubits, the number of $X$-type stabilizer generators in the code, instead of $n$ qubits. 
These compressed protocols start by preparing the qubits in $\ket{+}$ applying a sequence of $n$ pauli product rotations $P_{\pi/8}$ and then measure the $n-1$ last qubits in the $X$-basis. 
Our construction circumvent the use of a triorthogonal code and directly yields a protocol of this form.
Finally, this correspondence allows us to compute the prefactor $C$ of the output error $p_\mathrm{out}=C p_\mathrm{in}^d$ using the MacWilliams identity, as detailed in \cite{Bravyi_2012}.

%% file: Canonical_family_v2.tex
\section{Construction of a canonical family of distillation protocols}
\label{sec:canonical}

In this section, we explore constructions of distillation protocols from an analytical analysis of the logical $T$ gate on the repetition code. 
In particular, we build a \textit{canonical} family of distillation protocols $\mathcal{F}^0_N$ distilling $\ket{T}$ states by using $N$ logical qubits.
This construction recovers known protocols such as the $15$-to-$1$ on $N=5$ qubits and provides new ones, such as a $165$-to-$1$ protocol with distance $4$ and another $165$-to-$1$ protocol with distance $5$, using $N=10$ and $N=11$ qubits, respectively.
To our knowledge, these new distillation schemes require the smallest qubit footprint to distill at an order greater than $3$.
Furthermore, our construction generalizes to any single qubit rotation along the $Z$-axis at the level $L\geq 3$ of the Clifford hierarchy, i.e.~any rotation $R_Z(\pi/2^L)$ for $L\geq 3$. 
We illustrate this by providing distillation protocols for the $\ket{\sqrt{T}}$ magic state in \Cref{App:sqrt(T)}.

Let $N$ be a positive integer; the spatial footprint of our distillation protocol.
We denote $[N]=[\![1,N]\!]$ the set of indices of the qubits in the protocol.
We want to build a set $\mathcal{F}_N\subset \mathcal{P}([N])$ such that the circuit composed of the gates $\{P_{\frac{\pi}{8}} , P=\otimes_{i\in S} Z_i \text{ for } S \in \mathcal{F}_N\}$ performs a logical $T$ gate on the repetition code for each value of $N$.
A sequence $(\mathcal{F}_N)_{N\in\mathbb{N}}$ of such sets is called here a family of distillation protocols.
Note that the distance of a protocol $\mathcal{F}_N$ is the minimal number of erroneous gates in $\mathcal{F}_N$ combining into the undetected pattern of $Z$ errors on every qubits.
We denote $d(\mathcal{F}_N)$ the distance of $\mathcal{F}_N$.

To each set $\mathcal{F}_N\subset \mathcal{P}([N])$, we associate an indicator function $f:\mathcal{P}([N])\rightarrow\mathbb{F}_2$ as 
    \[f(A)=\begin{cases}
        1 &\text{ if } A\in \mathcal{F}_N\\
        0 &\text{ else}
    \end{cases}.\]
Conversely, this function fully defines the protocol associated with $\mathcal{F}_N$. For any possible gate $P_{\frac{\pi}{8}}$ with support $P=\otimes_{i\in A} Z_i$, $f(A)$ indicates whether the gate is in the protocol or not.
Now we want to impose the constraints of \Cref{eq:cons_T} on our set of selected gates.
To do so, we introduce $g$ as
\[\forall B \subseteq [N], g(B) = \sum_{A\supseteq B}f(A).\]
Intuitively, $g$ gives the parity of the number of Pauli product rotations in the protocol whose support includes the qubits of $B$.

\begin{lemma}[Constraints]\label{lemma:constraints}
     $\mathcal{F}_N\subset \mathcal{P}([N])$ defines a protocol for implementing a logical T gate on the repetition code if and only if
    \[\forall B \subseteq [N] \text{ s.t. }1\leq |B|\leq 3,~  g(B) = 1\]
    Additionally, we have that 
    \[\forall A \subseteq [N], f(A) = \sum_{B\supseteq A}g(B).\]
\end{lemma}
\begin{proof}
    The proof is given in \Cref{Proof_lemma_1}.
\end{proof}

The remaining degrees of freedom to define the family $\mathcal{F}_N$ is to fix the variables $g(B)$ for $B\subseteq [N]$ such that $|B|>3$.

Surprisingly, fixing all the free variables $\{g(B), B\in \mathcal{P}([N]), |B|>3\}$ to 0 yields an interesting family of protocols. 
In particular, without imposing any conditions related to the distance, this family still yields protocols with interesting distances relative to what can be achieved with a fixed number of qubits $N$.
We study this canonical family explicitly through the next theorem. 

\begin{theorem}[Canonical family of logical T circuits]\label{th:canonical_T}
    We call the canonical protocol and denote $\mathcal{F}_N^0$ the protocol defined by $$g_0(B)=\begin{cases}
        1 & \text{for all } B\subseteq [N]\text{ such that }|B|\leq 3,\\
        0 &\text{if } |B|>3.
    \end{cases}$$
This defines a circuit on $N$ qubits that implements a logical $T$ gate on the repetition code. 

    This family is defined by the following equality:
    \[
        \forall A\subseteq[N], f_0(A)\equiv\begin{cases}
            0 \mod 2&\text{ if } |A|>3,\\
            1 \mod 2&\text{ if } |A|= 3,\\
            N-1 \mod 2&\text{ if } |A|= 2,\\
            N+\binom{N-1}{2} \mod 2&\text{ if } |A|= 1.\\
        \end{cases}
    \]
    
    The distance of this family of protocols is 
    \[d(\mathcal{F}_N^0)=\begin{cases}
        \lceil\frac{N}{3}\rceil & \text{if $N$ even,}\\
        \text{smallest odd integer}\geq \frac{N}{3}& \text{if $N$ odd,}
    \end{cases}\]
    and the $T$-count is 
    \[n= \begin{cases}
        \binom{N}{1}+\binom{N}{2}+\binom{N}{3} & \text{if } N\equiv 0 \mod 4,\\
        \binom{N}{1}+\binom{N}{3}& \text{if } N\equiv 1 \mod 4,\\
        \binom{N}{2}+\binom{N}{3} & \text{if } N\equiv 2 \mod 4,\\
        \binom{N}{3} & \text{if } N\equiv 3 \mod 4.\\
    \end{cases}\]
\end{theorem}
\begin{proof}
    The proof is provided in \Cref{Proof_th1}
\end{proof}

This family of circuits reproducing the logical $T$ gate on the repetition code obeys additional symmetries.
It is made of multi-qubit $P_{\pi/8}$ supported on 1, 2, or 3 qubits, and whenever a gate appears, all the gates with the same support cardinality appear as well: a permutation group symmetry over the qubits is respected.
However, when brought back to the distillation framework, the gates no longer support such an obvious symmetry due to the mapping explained in \Cref{sec:logicalT-to-distill}, \Cref{eq1:from_alpha_to_beta}.

Surprisingly, even under these highly restrictive constraints, we find a distance-$3$ protocol on $5$ qubits, which can be proven optimal using a SAT solver, as explained in \Cref{sec:sat}.
The same family also yields a distance-$4$ protocol on $N=10$ qubits and a distance-$5$ protocol on $N=11$ qubits. We provide the parameters of protocols in this family for small $N$ in \Cref{tab:canonical_family_T}.

Strikingly, translating directly the triorthogonal constraints of \Cref{eq:triorth_cons} in this algebraic formulation (instead of going through the logical $T$ gate in the repetition code formulation), would give $g(\{1\})\equiv 1 \mod 2$ and $g(B)\equiv0 \mod 2$ for any other $B\subseteq[N]$ such that $|B|<4$.
Therefore, fixing $g(B)\equiv 0 \mod 2$ for $|B|>3$ would only yield the trivial protocol that performs a physical $T$ gate on the first qubit, providing no error suppression at all.
The repetition code (or any other classical code) spreads the logical operator over the $N$ qubits so that building an undetected error pattern out of atomic $Z$ error patterns of small support requires a lot of them.

\begin{table}[]
    \centering 
    \renewcommand{\arraystretch}{1.25} \setlength{\tabcolsep}{10pt} 
    \begin{tabular}{@{}ccc@{}} \toprule \textbf{Number of qubits} $N$ & \textbf{$T$-count} $n$ & \textbf{Distance} $d$ \\ \midrule 4 & 14 & 2 \\ 5 & 15 & 3 \\ 6 & 35 & 2 \\ 7 & 35 & 3 \\ 8 & 92 & 3 \\ 9 & 93 & 3 \\ 10 & 165 & 4 \\ 11 & 165 & 5 \\ 16 & 696 & 6 \\ 17 & 697 & 7 \\ 22 & 1\,771 & 8 \\ 23 & 1\,771 & 9 \\ 
    \bottomrule 
    \end{tabular}
    
    \caption{Parameters of the $\ket{T}$ distillation protocols from the canonical family $\mathcal{F}_N^0$ for small values of $N$.}
    \label{tab:canonical_family_T}
\end{table}

To further explore the landscape of $nT$-to-$1T$ distillation protocols, we use a SAT solver, as described in the following~\Cref{sec:exploring}.
This way, we either find other protocols not belonging to this canonical family with potentially lower $T$-count if the outcome is SAT, or if the solver provides UNSAT, we can state that no distillation protocol of distance $d$ exists on $N$ qubits.

Finally, using the same ideas to implement a logical ${\sqrt{T}}$ gate on a repetition code yields a set of constraints analogous to \Cref{eq:cons_T} with four $\mod 4$ equalities.
This way, using a function $g$ taking values in $\mathbb{Z}/4\mathbb{Z}$ to translate these constraints, we similarly build a canonical family of protocols for distilling $\ket{\sqrt{T}}$ states in \Cref{App:sqrt(T)}.
More generally, as explained in \Cref{App:sqrt(T)}, one can extend the framework to ${\sqrt[L]{T}}$ gates for any positive integer $L$.

%% file: Numerical_v2.tex
\section{Exploration of distillation factories with SAT solvers}
\label{sec:exploring}

In this section, we further explore distillation protocols using numerical methods. 
Specifically, we encode the constraints that a circuit $\mathcal{C}$ implementing the logical-$T$ gate on the repetition code must satisfy into a \textit{Boolean satisfiability} (SAT) problem.
Using the equivalence between logical $T$ gate circuits and distillation protocols established in \Cref{sec:logicalT-to-distill}, this allows us to exhaustively explore the relations between the number of qubits $N$, the $T$-count $n$, and the code distance $d$.
We first perform an exhaustive search of distillation factories over $N\leq 7$ qubits, and further search for specific codes up to $N=11$.
We then relax the SAT instance to explore restricted families of protocols at large number of qubits, up to $N=20$.
Finally, we generalize this approach to the $\CCZ$-state distillation.

\subsection{Logical \texorpdfstring{$T$}{T} gate on the repetition code as a SAT problem}\label{sec:sat}

\subsubsection{Encoding the SAT problem}

We show how finding circuits $\mathcal{C} \in \mathbb{F}_2^{N\times n}$ that implement the fault-tolerant logical $T$ gate on the repetition code can be expressed as a SAT problem with binary constraints.
As seen in \Cref{sec:alaVictor}, a circuit $\mathcal{C}$ consists of $n$ columns, each being a binary vector $\alpha^k\in \mathbb{F}_2^N$ for column $k$.
Since there are $2^N$ such possible vectors, we assign a boolean variable $v_k \in \{0,1\}$, accounting for whether the column $\alpha^k$ appears in $\mathcal{C}$.
We restrict each column to appearing at most once, since any Pauli product rotation that appears more than once is equivalent to a Clifford rotation and thus can be discarded.
This formulation in terms of boolean variables does not account for column ordering in $\mathcal{C}$, thus it is invariant under column permutation, encoding this symmetry directly in the problem structure.

The constraints $\mathcal{C}$ must satisfy to implement a valid logical $T$ gate are given by \Cref{eq:cons_T} and can be written in terms of boolean variables $v_k$ as
\begin{subequations} \label{eq:SAT_formulation}
\begin{align}
         \forall i \in [N],\quad & \bigoplus_{k=1}^{2^N} v_k.\alpha^k_i \equiv 1 \mod 2, \\
         \forall i_1<i_2 \in[N]^2,\quad &\bigoplus_{k=1}^{2^N} v_k.\alpha^k_{i_1}.\alpha^k_{i_2} \equiv 1 \mod 2, \\
         \forall i_1<i_2<i_3 \in[N]^3,\quad &\bigoplus_{k=1}^{2^N} v_k.\alpha^k_{i_1}.\alpha^k_{i_2}.\alpha^k_{i_3} \equiv 1 \mod 2.
\end{align}
\end{subequations}
where each constraint reduces to a \textsc{XOR} over all variables $v_k$ whose associated column $\alpha^k$ has support on the target singlet, pair, or triplet of qubits. 
Explicitly,
\begin{subequations}
\label{eq:SAT_XOR}
\begin{align}
    \forall i \in [N], \quad
        & \bigoplus_{\substack{k=1 \\ \alpha^k_i = 1}}^{2^N} v_k \equiv 1 \mod{2}, \\
    \forall i_1 < i_2 \in [N]^2, \quad
        & \bigoplus_{\substack{k=1 \\ \alpha^k_{i_1} = \alpha^k_{i_2} = 1}}^{2^N} v_k \equiv 1 \mod{2}, \\
    \forall i_1 < i_2 < i_3 \in [N]^3, \quad
        & \bigoplus_{\substack{j=1 \\ \alpha^k_{i_1} = \alpha^k_{i_2} = \alpha^k_{i_3} = 1}}^{2^N} v_k \equiv 1 \mod{2}.
\end{align}
\end{subequations}

The code distance can be enforced directly as an additional SAT constraint.
Specifically, the distance condition \Cref{eq:dparam} requires that no subset of fewer than $d_\text{min}$ selected columns \textsc{XOR}s to the all-ones vector $\mathbf{1}$, i.e. leading to a logical error.
Formally, that is the constraint
\begin{multline}
    \label{eq:dist_sat} 
    \forall K\subset [2^N]:|K|<d_\text{min}, \\\left[\bigoplus_{k\in K} \alpha^k = \mathbf{1} \mod 2 \right] \Longrightarrow \prod_{k\in K}v_k =0.
\end{multline}
Intuitively, this forbids any combination of at most $d_\text{min}-1$ columns whose direct sum yields a logical error from being simultaneously selected.

We implement this SAT problem both using z3~\cite{z3} and using Google's CP-SAT solver from the OR-Tools suite~\cite{cpsatlp}.
We made our code available on GitHub~\cite{compact_distillation}.
Note that in principle it would be more efficient to encode the problem in the variables $g(B)$ for $B\subseteq [N]$ described in the previous \Cref{sec:canonical}, as this consists of $\sum_{k=4}^N\binom{N}{k}$ variables, which is fewer than one boolean variable per element in $\mathcal{P}([N])$. 
However, this alternative is more complex in its design and does not extend to larger $N$, as the problem size still grows combinatorially with $N$.

\subsubsection{Exploring the space of small \texorpdfstring{$N$}{N} codes}

The SAT formulation introduced above enables a systematic exploration of logical $T$ circuits along several axes. 
A key feature is that the distance constraints allow us to determine whether a code of distance $d$ exists on $N$ physical qubits. Indeed, fixing $N$ and $d_{\min}$, an UNSAT result certifies that no code of distance $d_{\min}$ or higher can be implemented on $N$ qubits.

We begin by probing small values of $N$, incrementing $d$ until UNSAT is reached. 
This allows us to prove that the maximum achievable distance for $N=4$ is $d=2$, while $5 \leq N\leq 7$ admits codes of distance at most $d \leq 3$.
For $N \geq 8$, the SAT instance becomes too large to certify UNSAT within a practical time budget.
Specifically, the solver did not converge on the $N=8,~d_\text{min}=4$ problem, after two months of runtime.
Nevertheless, our experiments provide evidence that distance $d\geq 4$ may require at least $N=10$ qubits.
Indeed, whenever a satisfying assignment exists in the instances we tested, the solver typically finds it within seconds; no such solution was found for $N\in{8,9}$ at $d_\mathrm{min}=4$, while solutions were found for $N=10$.

Exploring larger $N$, we search for protocols of increasing distance. 
The number of clauses in \Cref{eq:dist_sat} grows combinatorially with $d$, so each increment carries a significant computational cost. 
The first $d=4$ code we identify is a $165$-to-$1$ protocol of distance $4$ on $N=10$ qubits, with an output error rate of $18900\,p^4$.
Imposing $d_{\min}=4$ and moving to $N=11$ qubits, we find a $165$-to-$1$ protocol of distance $5$ with an output error rate of $784245\,p^5$.
For $N>11$, we could not build distance constraints above $d_\text{min}=3$ due to the size exceeding the available memory of our computer.

We also investigate codes with a reduced $T$-count by adding an explicit upper bound on the number of columns,
\begin{equation}
    \sum_j v_j \leq n_{\max}.
\end{equation}
For $N \leq 7$, this constraint lets us certify that no $d=3$ $T$-state distillation protocol exists with fewer than $n = 15$ columns, as reported in~\cite{koutsioumpas2022smallestcodetransversalt}.
For $N = 10$, we identify a $80$-to-$1$ distance $4$ protocol with an output error rate of $1259\,p^4$, while below $n = 80$, the solver finds no $d=4$ solution within an hour of runtime.
For $N=11$, we were not successful in further reducing the $T$-count within reasonable runtime.

Although these results already reveal interesting structure, the size of the SAT instance grows rapidly with both $N$ and $d$. 
To explore the high-$(N,d)$ regime, we therefore turn to restricted subfamilies of the problem.

\subsection{Symmetric subfamilies \texorpdfstring{$\mathcal{F}$}{F}}
\label{sec:relax}

The repetition code is invariant under arbitrary permutations of its qubits.
In terms of $\mathcal{C}$, this means that permuting rows leads to the same protocol. 
Imposing this symmetry would therefore drastically reduce the search space of the SAT instance.
However, encoding full row-permutation symmetry while simultaneously preserving column-permutation invariance proved elusive. 
We therefore consider restricted families of distillation circuits that are constructed to be invariant under row permutations $S_N$, or under weaker, less global symmetries that still substantially constrain the problem while allowing access to broader families of codes.

\subsubsection{\texorpdfstring{$\mathcal{F}_{S_N}$}{F_{S_N}}: Permutation symmetry}

In order to reduce the complexity of the SAT instance, we consider a subfamily of distillation circuits that are symmetric under all permutations of the $N$ qubits, $S_N$, of the repetition code.
Within this family $\mathcal{F}_{S_N}$, the only relevant parameters are whether, for each weight $k \in \{1,\ldots,N-1\}$, all $\binom{N}{k}$ columns of Hamming weight $k$ appear in the decomposition.
We associate a Boolean variable $w_k$ to each weight class, so that $w_k = 1$ means every Pauli product rotation with support on exactly $k$ qubits is included. 
This reduces the SAT instance to only $N-1$ variables.

If $w_k = 1$, the number of $T$ gates acting on any given qubit is $\binom{N-1}{k-1}$, the number of $\CS^\dagger$ on any given pair of qubits is $\binom{N-2}{k-2}$, and the number of $\CCZ$ on any given triplet is $\binom{N-3}{k-3}$.
Since the conditions \Cref{eq:cons_T} to implement a logical $T$ gate on the repetition code only depend on the parities of these counts, we define
\begin{equation}
    \begin{aligned}
    l_k \equiv \binom{N-1}{k-1}& \mod{2}, \quad
    p_k \equiv \binom{N-2}{k-2} \mod{2}, \\[10pt]
    &t_k \equiv \binom{N-3}{k-3} \mod{2},
    \end{aligned}
\end{equation}
and the conditions~\Cref{eq:SAT_XOR} reduce to three \textsc{xor} constraints,
\begin{equation}
    \bigoplus_{k=1}^{N} w_k l_k \;=\;
    \bigoplus_{k=1}^{N} w_k p_k \;=\;
    \bigoplus_{k=1}^{N} w_k t_k \;=\; 1.
\end{equation}

Within this symmetric family, a logical error of weight $r$ exists if and only if there exist $r$ active weight classes $k_1 \leq \cdots \leq k_r$ satisfying $k_1 + \cdots + k_r = N$. 
The distance of a circuit $\mathcal{C}$ in this family is therefore
\begin{equation}
\begin{aligned}
    d = \min \bigl\{ r : \exists\, &k_1 \leq \cdots \leq k_r \\
    &\text{ s.t. }
        {\textstyle\sum_{j=1}^r} k_j = N \\
    &\text{ and }
        w_{k_1} = \cdots = w_{k_r} = 1 \bigr\}.
\end{aligned}
\end{equation}
Imposing a minimum distance $d_{\min}$ is then equivalent to forbidding all such sets of size $r < d_{\min}$.
Formally, this translates into the clauses
\begin{equation}
    \begin{aligned}
    \forall\, r& < d_{\min},\quad \\
    &\forall\, k_1 \leq \cdots \leq k_r\, \text{ s.t.\ } \, \textstyle\sum_{j=1}^r k_j = N, \\
    & \prod_{j=1}^{r} w_{k_j} = 0.
    \end{aligned}
\end{equation}
The number of such clauses grows combinatorially with $d_{\min}$, but remains tractable
given the small number of Boolean variables $w_k$.

Finally, the total number of columns in the circuit is $\sum_{k=1}^{N-1} w_k \binom{N}{k}$,
and an upper bound $n_{\max}$ on the circuit size translates directly into the linear constraint
\begin{equation}
    \sum_{k=1}^{N-1} w_k \binom{N}{k} \leq n_{\max}.
\end{equation}

For this construction, we explore circuits on $N\in[4,23]$ qubits.
For each $N$, we increase the distance until we reach UNSAT, excluding circuits of higher distance. Once the maximum distance is known, we attempt to reduce the $T$-count. Fixing the distance, we reduce $n_\text{max}$ incrementally : for each new run of the solver, if a solution is found with $n$ columns, we impose a new constraint $n_\text{max}<n$.
The valid constructions found are available in \Cref{tab:all_results}.

Interestingly, for each $N$, the family $\mathcal{F}_{S_N}$ achieves the same maximum distance as the canonical family $\mathcal{F}_N^0$ developed in \Cref{sec:canonical}.
The main difference between both families is the capability to use columns of weights $w_{k\geq4}$ in $\mathcal{F}_{S_N}$. However, including all columns of higher weight is detrimental for the distance, as producing logical errors will require fewer faulty gates. 
This explains why these two families essentially produce the same distillation circuits.

Similarly to the canonical family, these symmetric circuits form a strict subset of all possible distillation schemes. 
For instance, no symmetric scheme with $d = 3$ exists for $N = 6$, whereas a non-symmetric scheme with these parameters is easily obtained as a mere extension of the 15-to-1 protocol on $N=5$ qubits, showing the limit of this approach.

\subsubsection{\texorpdfstring{$\mathcal{F}_\mathrm{Y}$}{F_{Y}}: Young subgroup symmetry}
\label{sec:F_Y}

We now consider a strictly larger family $\mathcal{F}_\mathrm{Y}$ of distillation circuits, which we construct to be symmetric under the Young subgroup. 
Given a $K$-part partition $\lambda = (\lambda_1, \ldots, \lambda_K)$ of $N$ with $\sum_i \lambda_i = N$, the associated Young subgroup is $S_\lambda = S_{\lambda_1} \times \cdots \times S_{\lambda_K}$, corresponding to the independent permutations of qubits within each block.

Under $S_\lambda$, the gate columns in $\mathbb{F}_2^N$ are grouped into orbits characterized by their weight vector $\mathbf{w} = (w_1, \ldots, w_K)$, where $w_i$ is the Hamming weight of the column restricted to block $i$. 
The orbit of type $\mathbf{w}$ has size $\prod_i \binom{\lambda_i}{w_i}$, and we exclude the all-zero and all-one columns, i.e.~the orbits $\mathbf{w} = \mathbf{0}$ and $\mathbf{w} = \lambda$. 
We associate a Boolean variable $v_\mathbf{w} \in \{0,1\}$ to each remaining orbit, with $v_\mathbf{w} = 1$ indicating that all columns of type $\mathbf{w}$ are included in the circuit.

Ensuring a logical $T$ gate is implemented, the conditions \Cref{eq:cons_T} reduce, under $S_\lambda$ symmetry, to one XOR constraint per orbit of qubit subsets of size $t \in \{1,2,3\}$. Such an orbit is characterized by a type vector $\tau = (\tau_1, \ldots, \tau_K)$ with $\sum_i \tau_i = t$ and $\tau_i \leq \lambda_i$, where $\tau_i$ records how many qubits of the subset fall in block $i$.
For a fixed type $\tau$, the logical $T$ gate conditions associated with any qubit subset of type $\tau$ involves counting the columns in each orbit $\mathbf{w}$ which have their support containing that subset, modulo 2.
This count depends only on the orbit label $\mathbf{w}$ and the type $\tau$, and not on the specific subset chosen within the orbit.
Indeed, in block $i$, the number of ways to extend $\tau_i$ fixed positions to a support of size $w_i$ is $\binom{\lambda_i - \tau_i}{w_i - \tau_i}$. 
The parity contribution of orbit $\mathbf{w}$ to the constraint of type $\tau$ is therefore
\begin{equation}
    c_{\tau,\mathbf{w}} = \prod_{i=1}^K \binom{\lambda_i - \tau_i}{w_i - \tau_i} \mod 2.
\end{equation}
The constraint of type $\tau$ then reads
\begin{equation}
    \bigoplus_{\mathbf{w} \,:\, c_{\tau,\mathbf{w}} = 1} v_\mathbf{w} = 1.
\end{equation}
We impose this constraint for all $\tau$ with $|\tau| \in \{1, 2, 3\}$.

The distance condition \Cref{eq:dparam} generalizes naturally in this framework. 
Recall that a logical error corresponds to a set of faulty gates whose combined action is an undetectable $Z$-error. 
Under $S_\lambda$ symmetry, this translates into a condition on weight vectors.
Indeed, a multiset of $r$ orbit labels $\{\mathbf{w}^{(1)}, \ldots, \mathbf{w}^{(r)}\}$ constitutes a logical error of weight $r$ if, in every block $i$, the weights $w_i^{(1)}, \ldots, w_i^{(r)}$ can be combined to flip all $\lambda_i$ qubits in that block. 
Imposing minimum distance $d_\mathrm{min}$ can thus be done by forbidding all such multisets of size $r < d_\mathrm{min}$, each translated into a clause
\begin{equation}
    \prod_{j=1}^r v_{\mathbf{w}^{(j)}} = 0
\end{equation}
ensuring that at least one orbit in the multiset must be absent from the circuit.

Finally, the total $T$-count is 
\begin{equation}
    n = \sum_{\mathbf{w}} v_\mathbf{w} \prod_{i=1}^K \binom{\lambda_i}{w_i}.
\end{equation}
We can thus impose $n \leq n_\mathrm{max}$ to bound the $T$-count.

For each $N$, we enumerate all partitions of $N$ into $K \geq 2$ parts and run the SAT instance for each partition $\lambda$ to determine the maximum achievable distance $d^*(\lambda)$. We then restrict attention to partitions attaining the global maximum $d^* = \max_\lambda d^*(\lambda)$. Among these, we search for the minimum $T$-count by iteratively reducing $n_\mathrm{max}$. Cycling through these partitions, whenever the solver finds a solution with $n$ columns, we add the constraint $n_\mathrm{max} < n$ and rerun until the instance becomes unsatisfiable or times out after 360 seconds. The minimum $T$-count for $N$ is thus the minimum $T$-count over all partitions.

This strategy subsumes $\mathcal{F}_{S_N}$ as any circuit found under full $S_N$ symmetry is recovered in $\mathcal{F}_\mathrm{Y}$ via the single-block partition $\lambda = (N)$. Moreover, the $\mathcal{F}_\mathrm{Y}$ family grants access to asymmetric constructions. 
In particular, it recovers the $15$-to-$1$ protocol on any $N \geq 5$ qubits via the partition $\lambda = (1, \dots,1, 5)$, circumventing a structural limitation of $\mathcal{F}_{S_N}$.
It further helps reducing $T$-count in some specific case, for example, improving the $T$-count from $165$ to $164$ for the distance $d=4$ protocol over $N=10$ qubits.
In \Cref{tab:all_results}, we display these results.

\subsubsection{\texorpdfstring{$\mathcal{F}_\mathrm{C}$}{F_{C}}: Cyclic subgroup symmetry}

We construct a third subfamily $\mathcal{F}_\mathrm{C}$ of distillation circuits, obtained by replacing the Young subgroup symmetry with the weaker requirement of invariance under cyclic permutations within each block. Given a $K$-part partition $\lambda = (\lambda_1, \ldots, \lambda_K)$ of $N$, the associated symmetry group is the product of cyclic groups $C_\lambda = C_{\lambda_1} \times \cdots \times C_{\lambda_K}$, acting by independent cyclic shifts of qubits within each block. 
Since $C_{\lambda_i} \subsetneq S_{\lambda_i}$ for $\lambda_i \geq 3$, this symmetry is strictly weaker than the Young subgroup symmetry.
We construct this subfamily such that $\mathcal{F}_\mathrm{Y} \subseteq \mathcal{F}_\mathrm{C}$. 
The larger family therefore grants access to a broader class of distillation circuits, potentially reducing the $T$-count for a given distance, at the cost of a larger SAT instance.

Under $C_\lambda$, the columns of a circuit $\mathcal{C}$ are partitioned into orbits, with a structure finer than in the Young case. 
In this case, two columns belong to the same orbit if and only if one can be obtained from the other by independently shifting each block cyclically. 
Formally, the orbit of a column $\alpha \in \mathbb{F}_2^N$ under $C_\lambda$ is
\begin{equation}
    \mathrm{Orb}(\alpha) = \left\{ \sigma \cdot \alpha \;\middle|\; \sigma \in C_{\lambda_1} \times \cdots \times C_{\lambda_K} \right\}.
\end{equation}
Distinct orbits are no longer characterized by weight vectors alone.
Indeed, columns with the same Hamming weights, within each block, may now belong to different orbits if their internal binary patterns differ. 
We associate a Boolean variable $v_s \in \{0,1\}$ to each orbit $s$, excluding the all-zero and all-one columns as before, with $v_s = 1$ indicating that all columns of orbit $s$ are included in the circuit.

The parity constraints follow the same structure as in \Cref{sec:F_Y}. 
For each orbit of qubit subsets of size $t \in \{1,2,3\}$, we impose one XOR constraint. 
The parity contribution of orbit $s$ to the constraint associated with a qubit subset $B$ is
\begin{equation}
    c_{B,s} = \left|\left\{ \alpha \in \mathrm{Orb}_s \;\middle|\; B \subseteq \mathrm{supp}(\alpha) \right\}\right| \mod 2,
\end{equation}
i.e.\ the number of columns in orbit $s$ whose support contains $B$, taken modulo $2$. The corresponding XOR constraint reads
\begin{equation}
    \bigoplus_{s \,:\, c_{B,s} = 1} v_s = 1.
\end{equation}
Since qubit subsets within the same $C_\lambda$-orbit lead to identical constraints, it is sufficient to impose one constraint per orbit of qubit subsets, using canonical representatives of weight $t \in \{1,2,3\}$.

The distance condition follows the same logic as that of \Cref{sec:F_Y}. 
A logical error of weight $r$ corresponds to a multiset of $r$ orbit labels $\{s_1, \ldots, s_r\}$ such that one can select one column from each orbit so that their XOR equals the all-one vector $\mathbf{1} \in \mathbb{F}_2^N$. 
Imposing minimum distance $d_\mathrm{min}$ forbids all such multisets of size $r < d_\mathrm{min}$, each yielding a clause
\begin{equation}
    \prod_{j=1}^r v_{s_j} = 0.
\end{equation}
The total $T$-count is 
\begin{equation}
    n = \sum_s v_s \,|\mathrm{Orb}_s|,
\end{equation}
and we impose $n \leq n_\mathrm{max}$ to bound the circuit size.

The number of orbits under $C_\lambda$ grows faster with $N$ than under $S_\lambda$, making the SAT instance larger. Accessing large qubit numbers therefore requires substantially more computational resources than in $\mathcal{F}_\mathrm{Y}$, though the problem remains smaller than the fully unrestricted SAT of \Cref{sec:sat}.

The exploration strategy mirrors that of $\mathcal{F}_\mathrm{Y}$. For each $N$, we enumerate all partitions of $N$ into $K \leq 5$ parts. 
For each partition $\lambda$, we first determine the maximum achievable distance $d^*(\lambda)$ by increasing $d_\mathrm{min}$ until the instance becomes unsatisfiable. 
We then fix $d_\mathrm{min} = d^*(\lambda)$ and minimize the $T$-count iteratively.
This is, whenever a solution with $n$ gates is found, we impose $n_\mathrm{max} < n$ and rerun the solver until the instance becomes unsatisfiable or times out after 360 seconds. 
The result for qubit number $N$ is the minimum $T$-count across all partitions achieving the global maximum distance $d^* = \max_\lambda d^*(\lambda)$.

For $d=3$, the 15-to-1 protocol is the minimum $T$-count we found for any $N\geq 5$.
Considering $d=4$, we obtained a protocol on $N=10$ with a $T$-count of $n=64$ and output error rate $p_\text{out}=495p_\text{in}^4$, reducing the requirement compared to both the unrestricted SAT, which timed out attempting to find low $T$-count, and $\mathcal{F}_{S_N}$ which exclude this protocol.
For $d=5$, we obtained a protocol on $N=11$ with a $T$-count of $n=65$ and output error rate $p_\text{out} = 7947p_\text{in}^5$, again improving on previous constructions.
These results can be found in \Cref{tab:all_results}, and the distillation circuit of the $64$-to-$1T$ and $65$-to-$1T$ are given in \Cref{App:G_matrices}.
Instances of larger $N$ remain out of reach, as enumerating and storing the distance constraints alone exceeds available memory and time.

\subsection{Beyond \texorpdfstring{$\ket{T}$ distillation : $\ket{\CCZ}$ factories}{T distillation : CCZ factories}}
\label{sec:ccz}

A wide class of quantum subroutines implements arithmetic on integers.
This class of algorithms uses many Toffoli-based classical reversible circuits.
Implementing a Toffoli gate requires either $7$ $T$-states or $4$ T-states and one auxiliary qubit in the context of the \textsc{AND} gate~\cite{jones_low-overhead_2013,gidney_halving_2018}.
As a consequence, it is often more efficient to distill $\mathrm{CC}Z$ states rather than $T$ states, as implementing a Toffoli gate consumes a single $\mathrm{CC}Z$ state.

To do so, state-of-the-art architectures~\cite{gidney2025factor2048bitrsa,low2026denserplanarsurfacecode} employ the $8T\rightarrow1\CCZ$ protocol supported on $N=4$ logical qubits with an output error rate of $p_{\text{out}}=28 p_{\text{in}}^2$~\cite{jones_low-overhead_2013}.
This protocol is often used on top of either a first level of $T$ state distillation or a first level of $T$ state cultivation~\cite{gidney2024magicstatecultivation,ibm_qbit_recycling}, so that with $p_{\text{in}}\approx 10^{-7}$ one can reach a regime of $p_{\text{out}}\approx 10^{-12}$.
Such concatenation of protocols yielding $\mathrm{CC}Z$ states saturates the output logical error rate due to quadratic error suppression.
For very deep algorithms requiring a logical error rate below $10^{-12}$ the $8T\rightarrow1\CCZ$ is replaced by the $15T\rightarrow1T$ protocol with an output error rate of $p_{\text{out}}=35 p_{\text{in}}^3$, see Figure 1 of~\cite{Gidney_2021}.
This concatenation then saturates at a much lower error rate of $p_{\text{out}}\approx 10^{-19}$.
However, it requires compiling the Toffoli gate over a Clifford+$T$ gate set, which induces some overhead. 
Here we explore the possibility of building small distillation factories of the form $nT\rightarrow1\CCZ$ with an output error rate of $p_{\text{out}}\propto p_{\text{in}}^d$ with $d>2$.
To our knowledge, the closest form of such factories are reported in~\cite{Haah_2018,Campbell_2017, jones_low-overhead_2013, londe2026localdistillationreedmuller}.
In \cite{Haah_2018}, the family of $T\rightarrow \mathrm{CC}Z$ protocols is based on Reed-Muller code $\mathrm{RM}(r,m)$.
They provide explicit protocols for error suppression of order $d=2^{m/3}$, thus finding a $64T\rightarrow 2\mathrm{CC}Z$ with error suppression $2944p^4$.
This code has 11 independent $X$-stabilizer generators, so that the associated circuit with Pauli product rotations $P_{\pi/8}$ is supported on $2\times 3+11 =17$ qubits.
This is quite similar to the $64T\rightarrow 2\mathrm{CC}Z$ protocol of \cite{Jones_2013_64_to_2}, exhibiting error suppression in $3072p^4$ using $17$ qubits as well.
Here, we leverage a SAT solver approach to fill the gap between distance $2$ and $3$, extending the repetition-code framework developed for $\ket{T}$-state distillation in \Cref{sec:LogicalT_repCode} to the distillation of $\ket{\CCZ}$-states.

Recall that $T$ gate injection produces only $Z$-type errors on the output state (see \Cref{sec:prelim}).
We want to reproduce the logical effect of a $\CCZ$ gate using a Pauli product rotations circuit $\mathcal{G}$ while ensuring that errors arising from these gates are detected.
One way to do so is to use a repetition code against phase flips on one of the three qubits of the $\CCZ$ gate while forbidding the use of Pauli product rotations with support on unprotected qubits only.
Following the same decode-gate-encode strategy as in \Cref{sec:LogicalT_repCode}, the imitated logical circuit, that we abusively denote $\CCZ_L$, is
\begin{equation}
    \begin{quantikz}
        & & \ctrl{2} & & \\
        & & \control{} & & \\
        & \gate[4]{\mathcal{E}^\dagger} & \control{} & \gate[4]{\mathcal{E}} & \\
        &&&&\\
        \setwiretype{n} & & \vdots & & \\
        & & & & \\
    \end{quantikz} =: \begin{quantikz}
        & \gate[3]{\CCZ_L} & \\
         &  & \\
       &\qwbundle{}  &   \\
    \end{quantikz}
\end{equation}
Commuting the $\CCZ$ gate with the decoding circuit $\mathcal{E}^\dagger$ yields a circuit composed of $\CCZ$ gates between qubits $1$, $2$, and each of the $N-2$ remaining qubits of the repetition code.
Applying the same construction as in \Cref{sec:alaVictor}, we obtain an encoded $\CCZ_L$ gate circuit $\mathcal{G}$ on $N$ qubits decomposed into $n$ Pauli product rotations $\{P^k\}_{k=1}^n$, each associated with a binary vector $\alpha^k \in \mathbb{F}_2^N$ as in \Cref{eq:gate_sequence_P_from_alpha}.
For $\mathcal{G}$ to correctly implement $\CCZ_L$, the vectors $\{\alpha^k\}$ must satisfy the constraints 
\begin{subequations}\label{eq:strong_cons_CCZ}
    \begin{align}
        & \forall i\in[\![1,N]\!], \sum_k \alpha_i^k \equiv 0 \mod 8, \\
        & \forall i<j \in[\![1,N]\!]^2, \sum_k \alpha_{i}^k\alpha_j^k \equiv 0 \mod 4, \\
        & i=1,j=2,\forall l\in[\![3,N]\!], \sum_k \alpha_1^k\alpha_2^k\alpha_l^k \equiv 1 \mod 2, \\ 
        & \forall i<j<l\in[\![3,N]\!]^3, \sum_k \alpha_i^k\alpha_j^k\alpha_l^k \equiv 0 \mod 2,
    \end{align}
\end{subequations}
which reduce to
\begin{subequations}\label{eq:cons_CCZ}
    \begin{align}
        & \forall i\in[\![1,N]\!], \sum_k \alpha_i^k \equiv 0 \mod 2, \\
        & \forall i<j \in[\![1,N]\!]^2, \sum_k \alpha_{i}^k\alpha_j^k \equiv 0 \mod 2, \\
        & i=1,j=2,\forall l\in[\![3,N]\!], \sum_k \alpha_1^k\alpha_2^k\alpha_l^k \equiv 1 \mod 2, \\ 
        & \forall i<j<l\in[\![3,N]\!]^3, \sum_k \alpha_i^k\alpha_j^k\alpha_l^k \equiv 0 \mod 2. 
    \end{align}
\end{subequations}
when Clifford corrections are treated as free.

For a fault-tolerant implementation, we need distance $d > 1$ circuits.
The notion of distance is more complex than in the $T$ gate case, as logical error corresponds to any $Z$-error pattern on qubits $1$, $2$, or $3$ that goes undetected by the repetition code. 
Concretely, there are seven such undetectable error patterns,
\begin{equation}
    \mathbf{E}= \left\{ 
    \left(\begin{smallmatrix}0 \\ 1 \\ 0 \\ \vdots \\  0 \end{smallmatrix}\right), 
    \left(\begin{smallmatrix}1 \\ 0 \\ 0 \\ \vdots \\  0 \end{smallmatrix}\right), 
    \left(\begin{smallmatrix}1 \\ 1 \\ 0 \\ \vdots \\  0 \end{smallmatrix}\right), 
    \left(\begin{smallmatrix}0 \\ 0 \\ 1 \\ \vdots \\  1 \end{smallmatrix}\right), 
    \left(\begin{smallmatrix}0 \\ 1 \\ 1 \\ \vdots \\  1 \end{smallmatrix}\right), 
    \left(\begin{smallmatrix}1 \\ 0 \\ 1 \\ \vdots \\  1 \end{smallmatrix}\right), 
    \left(\begin{smallmatrix}1 \\ 1 \\ 1 \\ \vdots \\  1 \end{smallmatrix}\right)
    \right\},
\end{equation}
Adapting \Cref{eq:dparam}, the distance can thus be written,
\begin{equation} \label{eq:dparam_CCZ}
    d = \min_{K}\, |K| \quad \text{s.t.} \quad \sum_{k\in K} \alpha^k \equiv e \pmod{2}, \quad e \in \mathbf{E}.
\end{equation}

Finally, the binary transformation of \Cref{sec:logicalT-to-distill} is applied to the $N-2$ last bits of the columns to build a $\CCZ$-state distillation protocol.
Using the SAT formulation of \Cref{sec:sat}, adapted to \Cref{eq:cons_CCZ} and \Cref{eq:dparam_CCZ}, we find that no $nT \rightarrow 1\CCZ$ protocol with $d > 2$ exists for $N < 8$ qubits. 
Searching over larger qubit numbers, we discover a $47T \rightarrow 1\CCZ$ protocol with error suppression $p_{\mathrm{out}} = 236\, p_{\mathrm{in}}^3$ on $N = 9$ qubits, and a $48T \rightarrow 1\CCZ$ protocol with $p_{\mathrm{out}} = 2816\, p_{\mathrm{in}}^4$ on $N = 10$ qubits.
The corresponding $\mathcal{G}$ matrices are provided in \Cref{App:G_matrices}.
Interestingly, this contrasts with the $T$-state case. 
Indeed, for $T$-state distillation, we empirically observed plateaus at odd distances, whereas for $\CCZ$ distillation the plateaus happen at even distances.

%% file: Recycling.tex
\section{Compressing distillation protocols with qubit recycling}
\label{sec:recycling}


Throughout this work, we explored the landscape of compact distillation protocols through its conventional mapping with weakly triorthogonal matrices with small number of rows~\cite{Haah_2018, Litinski_2019}.
By design, our framework implements distillation factories as a sequence of $Z$-basis Pauli product rotations followed by check-qubit measurements, which conveniently maps the number of rows to the qubit count $N$.
However, this count can be further lowered by recycling qubits which are done participating in the rotations, and that thus can be measured and reinitialized~\cite{ibm_qbit_recycling}.
Indeed, a check qubit needs only to be present from its first Pauli product rotation until its last, whereas the output qubit needs to be present from its first rotation until the end of the protocol.
In terms of the distillation matrix $\mathcal{G}$, the $i$-th qubit needs to be present at every column index $j$ between the first and last $1$ of the $i$-th row, and until the last column for the output row.
At the cost of mid-circuit measurements and initialization, a distillation protocol can thus be compressed down to the maximum number of qubits that need to be simultaneously in use, or \emph{active}, at any point in the circuit.
For a distillation matrix $\mathcal{G}$, this maximum active-qubit count reads
\begin{equation}
  A(\mathcal{G}) = \max_{j} \bigl|\{\, i : f_i \le j \le \ell_i \,\}\bigr|,
\end{equation}
where $f_i$ and $\ell_i$ is the column indices of the first and last $1$ of row $i$, with the convention that $\ell_i$ is the last column for the output row, i.e. $\ell_1 = n$.

The active-qubit count is not fixed by a $nT \to T$ distillation protocol itself, but by the particular matrix $\mathcal{G}$ that implements it.
Indeed, a given protocol admits many equivalent implementations, related to one another by column permutations, which reorder the commuting Pauli product rotations, and by row additions that preserve the triorthogonality conditions, which change the generator basis without altering the distilled state (see \Cref{app: recycling}).
Therefore, to compress a distillation protocol, we search over all matrices $\mathcal{G}$ equivalent to it for the one that minimizes $A(\mathcal{G})$.

To reduce the maximum active-qubit count, we first focus on the output row, whose qubit is the most costly as it stays active until the end of the protocol.
By adding check rows to it, we lower its Hamming weight, so that the output qubit is initialized as late as possible (see \Cref{app: min generator weight}).
In a similar manner, we reduce the weights of the check rows.
We then tackle the column ordering with a branch-and-bound algorithm that builds $\mathcal{G}$ column by column while tracking the set of currently active qubits, discarding any partial ordering whose active count already reaches the best value found so far. We describe this method in detail in \Cref{app: bnb}.
While this search is exhaustive, certifying optimality scales exponentially with the number of columns $n$.
We therefore complete the search by framing the column ordering as an integer linear program (ILP).
Rather than optimizing over permutations directly, we assign each column to a position in the sequence and introduce, for every qubit, indicator variables marking whether it has been initialized and whether it is still to be used at that position, from which its being active follows.
Warm-started with the branch-and-bound ordering, the program minimizes $A(\mathcal{G})$ and returns a certificate that the reported active-qubit count is optimal over column permutations.
\Cref{app: ilp} provides more explanation on the ILP.
We provide a numerical implementation of this compression method in~\cite{compact_distillation}.

We first apply qubit recycling to compress the $15$-to-$1$ protocol, which brings the implementation down from $N=5$ to $4$ active qubits.
Targeting the $49$-to-$1$ protocol, we reduce its footprint from $N=14$ to $5$ active qubits, the column ordering being certified optimal by the ILP.
To our knowledge, this is the first time a distance $d=5$ protocol can be implemented on $5$ qubits, going beyond the $7$ active qubits implementation reported in ~\cite{ibm_qbit_recycling} and the $6$ qubits one of ~\cite{gidney_2024_13777072}.
We further applied our recycling pipeline to $\CCZ$ distillation protocols, recovering this way an implementation of the $64T$-to-$2\CCZ$ protocol on $10$ active qubits as already reported in~\cite{ibm_qbit_recycling}.
For every other distillation and synthillation protocol reported in this paper, however, we did not find any reduction in the active-qubit count.

%% file: Conclusion_v2.tex
\section{Summary of main results}\label{sec:results}

\begin{table}[htbp]
\centering

\renewcommand{\arraystretch}{1.25}
\setlength{\tabcolsep}{7pt}
 
\begin{tabular}{@{}clcccc@{}}
\toprule
\multicolumn{2}{c}{\textbf{Method}} & $N$ & $n$ & $d$ & $C$ \\
\midrule
 
\multicolumn{6}{c}{\textit{$T$-state distillation}} \\[2pt]
 
\multirow{3}{*}{\rotatebox[origin=c]{90}{SAT}}
  & & \textbf{5--7} & \textbf{15}  & \textbf{3} & \textbf{35}        \\
  & & 8--9 & 15  & 3 & 35        \\
  & & 10   & 80  & 4 & 1\,259    \\
  & & 11   & 165 & 5 & 784\,245  \\[4pt]
 
\multirow{19}{*}{\rotatebox[origin=c]{90}{$\mathcal{F}_{S_N}$}}
  & & 5  & 15     & 3 & 35                \\
  & & 6  & 15     & 2 & 35                \\
  & & 7  & 35     & 3 & 105               \\
  & & 8  & 92     & 3 & 280               \\
  & & 9  & 45     & 3 & 120               \\
  & & 10 & 165    & 4 & 18\,900           \\
  & & 11 & 165    & 5 & 784\,245          \\
  & & 12 & 298    & 4 & 15\,400           \\
  & & 13 & 299    & 5 & 2\,002\,000       \\
  & & 14 & 455    & 5 & 1\,401\,400       \\
  & & 15 & 455    & 5 & 1\,401\,400       \\
  & & 16 & 696    & 6 & 442\,842\,400     \\
  & & 17 & 697    & 7 & 78\,189\,711\,600 \\
  & & 18 & 969    & 6 & ---               \\
  & & 19 & 969    & 7 & ---               \\
  & & 20 & 1\,350 & 7 & ---               \\
  & & 21 & 1\,351 & 7 & ---               \\
  & & 22 & 1\,771 & 8 & ---               \\
  & & 23 & 1\,771 & 9 & ---               \\[4pt]
 
\multirow{3}{*}{\rotatebox[origin=c]{90}{\parbox{1.8cm}{\centering $\mathcal{F}_\mathrm{Y}$}}}
  & & ${\leq}\,9$ & 15  & 3 & 35  \\
  & & 10          & 164 & 4 & --- \\
  & & 11          & 165 & 5 & 784\,245 \\[4pt]
 
\multirow{3}{*}{\rotatebox[origin=c]{90}{$\mathcal{F}_\mathrm{C}$}}
  & & ${\leq}\,9$ & 15 & 3 & 35     \\
  & & \textbf{10}          & \textbf{64} & \textbf{4} & \textbf{495} \\
  & & \textbf{11}          & \textbf{65} & \textbf{5} & \textbf{7\,947} \\
 
\midrule\midrule
 
\multicolumn{6}{c}{\textit{$\CCZ$-state distillation}} \\[2pt]
 
\multirow{3}{*}{\rotatebox[origin=c]{90}{SAT}}
  & & \textbf{4--8}             & \textbf{8}          & \textbf{2}          & \textbf{28}             \\
  & & \textbf{9}              & \textbf{47} & \textbf{3} & \textbf{236}   \\
  & & \textbf{10}             & \textbf{48} & \textbf{4} & \textbf{2\,816}\\
 
\bottomrule
\end{tabular}
\caption{
  T and $\CCZ$ distillation protocols found by each numerical method.
  The methods are the unrestricted SAT for $T$-distillation of \cref{sec:sat}, the three subfamilies presented in \cref{sec:relax}, and the extension of SAT for $\CCZ$ distillation detailed in \cref{sec:ccz}.
  Distillation protocols are characterized by $N$, the number of qubits, $n$ the $T$-count, and have output error rate $p_\mathrm{out}=C p_\mathrm{in}^d$.
}\label{tab:all_results}
\end{table}

We establish a direct link between the number of qubits $N$, the $T$-count $n$, and the distance $d$ for distillation protocols made of conventional $Z$-diagonal Pauli product rotation circuits
--the standard compressed implementation of triorthogonal codes, without mid-circuit measurement--, by framing the construction of fault-tolerant logical $T$ and $\CCZ$ gates within a repetition code framework.
We illustrate our pipeline in ~\Cref{fig:pipeline}.
\begin{figure*}
    \centering
    \includegraphics{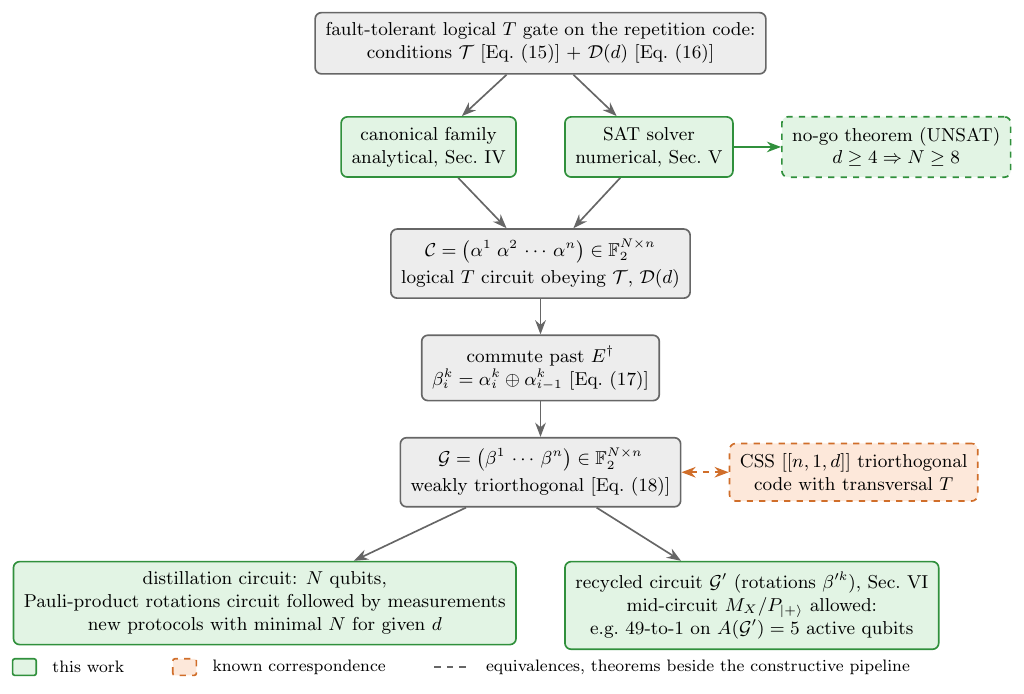}
    \caption{\textbf{Logical pipeline used to explore distillation protocols.} The conditions defining a
fault-tolerant logical $T$ gate on the repetition code from~\cref{sec:LogicalT_repCode} can either be solved analytically as in~\cref{sec:canonical} or via SAT as in~\cref{sec:sat}, whose UNSAT instances yield no-go theorems.
As explained in ~\cref{sec:logicalT-to-distill} commuting a solution $\mathcal{C}$ through the decoding circuit $E^{\dagger}$ gives weakly triorthogonal matrix $\mathcal{G}$, the type of matrix used to define CSS $[[n,1,d]]$ triorthogonal code in \cite{Litinski_2019_GoS, Haah_2018}.
Distillation protocols can then be implemented either directly as a distillation circuit on $N$ qubits or, with qubit recycling explained in~\cref{sec:recycling}, on $A(\mathcal{G}')\le N$ active qubits.
Here the pipeline is illustrated for $T$ state distillation protocols, the same applies to $\CCZ$ or $\sqrt{T}$ distillation protocols.}
    \label{fig:pipeline}
\end{figure*}
In \cref{tab:all_results}, we report all the results found using this framework, with some highlighted in \cref{fig:Main_results}.
To express our findings in a natural manner, we define the following quantity:
\[N_{T}(d)= \min_{\substack{\text{$nT$-to-$1T$ protocols}\\
\text{ of distance } \geq~d \\
\text{ on }N \text{ qubits}}} N\]
We limited our exploration to protocols producing a single magic state.
If we denote $N_{T}^{(k)}(d)$ the same quantity for protocols producing $k$ magic states, we immediately have that for all $k\geq 2$,  $N_{T}(d)\leq N_{T}^{(k)}(d)-(k-1)$ because one can always simply remove the qubits of the $k-1$ other magic states from the protocol.
Analogously, we will use $N_{\CCZ}(d)$ for $nT$-to-$1 \CCZ$ protocols.

For $\ket{T}$-state distillation, an exhaustive search based on SAT instances establishes that no distillation protocol with distance $d>3$ can be implemented on $N<8$ qubits, while the first instance of $d\geq4$ has been found for $N=10$ qubits, namely $8 \leq N_T(d=4) \leq 10$.
Restricting to the small subfamily $\mathcal{F}_{S_N}$ of distillation protocols, we uncover protocols with distances up to 9 on $N=23$ qubits.
In \Cref{tab:all_results} we provide the prefactor for these protocols up to $N=17$, after which computing $C$ becomes too resource intensive.
Further exploring the $\mathcal{F}_\mathrm{Y}$ family leads to an improvement of protocols, notably finding a $N=10,d=4$ protocol with a $T$-count of $164$.
Finally, the $\mathcal{F}_\mathrm{C}$ family helped optimizing $T$-count with two new distillation protocols.
This way, at $N=10$, we obtain a $64$-to-$1$ protocol with output error $495p_\mathrm{in}^4$.
For one more qubit in the repetition code, $N=11$, we find a $65$-to-1 protocol with an output error rate of $7947p_\mathrm{in}^5$.
To our knowledge, these are the smallest qubit implementations with unitary circuits of distance-4 and distance-5 $T$-state distillation protocols.
For comparison, the 49-to-1 protocol, with a lower $T$-count, comes at the price of requiring an implementation over $N=14$ qubits.

Beyond these unitary-circuit implementations, we further reduce the qubit footprint through qubit recycling, allowing check qubits to be measured out and their qubits reused once they are done participating in the rotations.
For the $15$-to-$1$ protocol this lowers the implementation from $N=5$ to $4$ active qubits, and for the $49$-to-$1$ protocol of~\cite{Bravyi_2012} from $N=14$ to $5$ active qubits.
This is, to our knowledge, the first implementation of a distance-$5$ protocol on $5$ qubits.

\medbreak

For $\CCZ$-distillation, an exhaustive exploration based on SAT instances show that no $d\geq 3 $ protocols exist for $N<8$ qubits, i.e.~$N_{\CCZ}(d=3)\geq 8$.
To our knowledge, we report the first instance of $nT\rightarrow 1\CCZ$ protocols for distance 3. 
In particular, using $N=9$ qubits, we provide a $47T \rightarrow 1\CCZ$ protocol with error suppression $p_{\mathrm{out}} = 236\, p_{\mathrm{in}}^3$, stating that $N_{\CCZ}(d=3)\leq 9$.
For $N=10$, we obtain a $48T \rightarrow 1\CCZ$ scheme with $p_{\mathrm{out}} = 2816\, p_{\mathrm{in}}^4$, requiring $7$ fewer qubits than previously known $d=4$ protocols, demonstrating that $N_{\CCZ}(d=4)\leq 10$.
Indeed, it was previously identified that a $64T$-to-$2\CCZ$ protocol distilling in $2944p^4$ supported on $N=17$ existed~\cite{Haah_2018,Jones_2013_64_to_2} which provides a better ratio of check qubits per output state, at the cost of more qubits.
As $N_{\CCZ}(d=3)\geq 8$ we have that $N_{\CCZ}(d=4)\geq 8$, but it was known that $N_{\CCZ}^{(2)}(d=4)\leq 17$, so it is very likely that $N_{\CCZ}(d=4)\geq\frac{N_{\CCZ}^{(2)}(d=4)}{2}$ which would prove that it is structurally worthwhile to look for distillation factories producing a higher number of magic states.

This exploration of $\ket{\CCZ}$ and $\ket{T}$ distillation protocols of the form $nT$-to-$1 T$ or $nT$-to-$1\CCZ$ both exhibit a pattern in terms of $N_{\CCZ}(d)$ and $N_{T}(d)$.
The protocols gathered in \Cref{tab:all_results} suggest that for $m\geq 1$, 
\[N_{T}(d=2m+1)=N_{T}(d=2m)+1\]
while the gap between $N_{T}(d=2m)$ and $N_{T}(d=2m-1)$ is larger.
Indeed, the first protocol of distance $4$ (resp. $6$ and $8$) we have found comes with $5$ additional qubits compared to the smallest protocol of distance $3$ (resp. $5$ and $7$).
Interestingly, for $\CCZ$ protocols, the pattern looks inverted, the few protocols discovered with the SAT solver suggest that for $m\geq 2$, 
\[N_{\CCZ}(d=2m)=N_{\CCZ}(d=2m-1)+1\]
while reaching distance $3$ required $5$ additional qubits compared to distance $2$.
Qualitatively, this pattern suggests that it is better to distill $T$ states with odd distances and $\CCZ$ with even distances. 

Finally, we discuss how these findings compare to state-of-the-art logical $\ket{T}$-state preparation protocols and how that information complements our understanding of efficient magic state preparation protocols.
This systematic exploration provides useful information regarding the gap to overcome in order to distill in $p^5$ instead of $p^3$.
Distilling in $p^5$ would yield a state with an error rate below $10^{-11}$ for $p=10^{-3}$ and $C\approx10^4$ which is precisely in the target of chemistry use cases such as FeMoco~\cite{beverland_resource_estimation} or Shor's algorithm~\cite{gidney2025factor2048bitrsa}.
Reaching this $p^5$ suppression in a single distillation stage has, until now, come with a prohibitive space overhead.
In the unitary-circuit setting the $49$-to-$1$ protocol requires $N=14$ qubits, or $N=11$ for the most compact distance-$5$ construction of \Cref{tab:all_results}, and each of these qubits must be encoded at a high code distance to preserve the quality of the distilled state.
This is what has encouraged two-level distillation factories, either by the concatenation of two $15T$-to-$1T$ protocols or the cultivation~+~$15T$-to-$1T$~\cite{gidney2024magicstatecultivation, google_cultivation_experiment}.
In a two-level factory, the first level can run at a small inner code distance, so that only the few qubits of the second level need to be encoded at a high distance~\cite{Litinski_2019}, keeping the overall footprint small.
Allowing mid-circuit measurements, however, the $49T$-to-$1T$ protocol can be implemented on only $5$ active qubits, the same number as the usual single $15T$-to-$1T$ factory.
A single-stage $p^5$ factory now fits within the same footprint as the high-distance second level of a two-level factory.
A full spacetime-volume comparison across physical error rates, folding in circuit depth and the post-selection overhead, is left to future work, but our results indicate that single-stage $p^5$ distillation is no longer disqualified on spatial grounds, at least for surface-code-based quantum computation.
Beyond the surface code, for qLDPC codes of very high encoding rates~\cite{kasai2026breakingorthogonalitybarrierquantum,quera_qldpc_VHER}, single-stage distillation factories may become even more relevant.


\section{Conclusion}
\label{sec:conclusion}

In this work, we emphasize the link between the qubit count $N$ and the distance $d$ of distillation protocols.
To do so, we employ a mapping of the distillation problem onto the design of a circuit of Pauli product rotations that acts as a logical $T$ gate on a classical error correcting code.
Using this construction, we have been able to build numerical and analytical tools to generate distillation protocols with constraints on the number of qubits $N$ employed in the circuits, as well as on the distance $d$ and the $T$-count $n$.
Using a \textit{SAT} formulation of the problem, we have been able to certify that no $nT$-to-$1T$ distillation protocol with distance $\geq 4$ can be implemented on $N<8$ qubits using standard $Z$-pauli product rotation circuits.
Besides numerical exploration, we built an analytical approach to design a canonical family of distillation protocols.
This way, we figured out that distance $d=4$ could be achieved with as few as $N=10$ qubits, and $d=5$ can be achieved with $N=11$ qubits, while the previous smallest unitary circuit distillation protocol achieving distance $d=5$ was the $49T$-to-$1T$ of~\cite{Bravyi_2012} on $N=14$ qubits.
We further improved our numerical exploration by considering subfamilies constructed to be invariant under various symmetries, namely from the strongest to the weakest: a full permutation symmetry, a Young subgroup symmetry, and a cyclic symmetry over a partition.
This targeted numerical exploration scheme improved on the parameters $n$ and $C$ of our previous most compact protocols.
For $d=4$, we have found a $64T$-to-$1T$ with a reduction of $p \to 495p^4$ on $N=10$ qubits, while for $d=5$, we have found a $65T$-to-$1T$ with a reduction of $p \to 7947p^5$ on $N=11$ qubits.
Allowing for mid-circuit measurements and qubit re-initializations, we have been able to compress some distillation factories. Most notably , we found an implementation of $49T$-to-$1T$ protocol, with output error $p \to 1411p^5$, on $5$ active qubits instead of $N=14$ qubits.

Extending our approach to $\CCZ$-state distillation, we have been able to show that any distance-$3$ protocol requires at least $8$ qubits.
We reported one of the most compact such protocols we have found, a $47T$-to-$1\CCZ$ scheme on $N=9$ qubits with output error rate $236\,p_\mathrm{in}^3$.
We further propose a $48T$-to-$1\CCZ$ distance-$4$ protocol on $N=10$ qubits, reducing the qubit footprint by $7$ compared to the previously most compact known construction~\cite{Jones_2013_64_to_2,Haah_2018} that was however producing $2$ $\CCZ$-states at a time.

Even though the qubit-cycle volume of a magic state factory remains the gold standard figure of merit for large-scale quantum computation resource estimation, near to medium term hardware capabilities will likely impose strong constraints on the qubit count.
In this regime, our exploration provides concrete insight into what can be achieved with few logical, and hence physical, qubits.
A natural next step toward lowering the overall spacetime volume is the regime of factories that distill many magic states at once.
This dense regime has so far been explored mostly through quantum triorthogonal codes built from punctured Reed--Muller codes~\cite{Haah_2018}, and the framework introduced here offers a route to investigate it differently.
In particular, a canonical family of \Cref{sec:canonical} can be built from any classical code $[n_c, k_c, d_c]$: a high encoding rate $k_c/n_c$ then yields distillation protocols of distance $\geq \lceil d_c/3 \rceil$ producing $k = k_c$ output states on $N = n_c$ qubits.
Finally, we show that qubit recycling is central to compressing the spatial footprint of these protocols.
Future work could bias the search toward protocols that recycle well, e.g. by incorporating into the SAT instances objectives such as a low output-row weight, a low overall weight, or few $Z$ operators per column, and to assess which of these best predict a low active qubit count.

%% file: Commutations_relations.tex
\section{Commutation relations}\label{App:Commutations_relations}
In this Appendix, we detail and prove the commutation relation used to commute the T gate through the encoding circuit of the repetition code. 

\subsection{T gate commutation}
More precisely, we prove for a given $N\geq 3$:
\begin{widetext}
\begin{equation}
    \begin{quantikz}
&  &\ \ldots\ &&\targ{}& \gate{T} &\targ{}  && \ \ldots \ && \\
&&\ \ldots\ &\targ{} & \ctrl{-1}&& \ctrl{-1} & \targ{} & \ \ldots \ &&\\
&&\ \ldots\ & \ctrl{-1}&&&& \ctrl{-1} & \ \ldots \ &&\\
\setwiretype{n}  \vdots && \ \ldots \ &&&&&&\ \ldots \ && \vdots \\
& \targ{} & \ \ldots \ &&&&&& \ \ldots \ & \targ{}&\\
& \ctrl{-1} & \ \ldots \ &&&&&& \ \ldots \ & \ctrl{-1} & 
\end{quantikz} =
\begin{quantikz}
    & \gate{T}& \gate[3]{\prod_{i,j} (\CS^{\dagger})_{ij}} & \gate[3]{\prod_{i,j,k} (\CCZ)_{ijk}} &&\\
 \setwiretype{n} & \vdots &&& \vdots &\\
 & \gate{T} &&& &\\ 
\end{quantikz}
\end{equation}\label{circ:commutations_appendix}
\end{widetext}
\begin{proof}
To do things carefully, we work in the computational basis $\ket{\epsilon_1,\ldots,\epsilon_N}$ and show that both circuits have the same effect on this state, no matter the value of $\epsilon_1, \ldots,\epsilon_N\in\{0,1\}^N$.

We start by computing the effect of the left hand-side circuit.
We fix $N\geq3$ and $\epsilon_1, \ldots,\epsilon_N\in\{0,1\}^N$.
In what follows, we write $\oplus$ for modulo 2 sums.
At the beginning, the state is $\ket{\psi_{i}}=\ket{\epsilon_1,\ldots,\epsilon_N}$
After the first decoding circuit, the state is transformed to $\ket{\epsilon_1',\ldots,\epsilon_N'}$ with the $\epsilon'$ defined as:
\begin{align*}
    &\epsilon_N' = \epsilon_N\\
    &\forall i\in [\![1,N-1]\!], \epsilon_i' = \epsilon_i\oplus\epsilon_{i+1}\oplus \ldots \oplus\epsilon_N  \
\end{align*}

Then the T gate applies a phase $e^{i\frac{\pi}{4}}$ if and only if $\epsilon_1'=1$,
so the state becomes $e^{i\frac{\pi}{4}\epsilon_1'}\ket{\epsilon_1',\ldots,\epsilon_N'}$.
Then the last CNOT sequence brings it back to 
\begin{align*}
    &\forall i\in [\![1,N-1]\!], \epsilon_i'' = 
    \epsilon_i'\oplus\epsilon_{i+1}' =  \epsilon_i \\
    &\epsilon_N'' = \epsilon_N =\epsilon_N\
\end{align*}

The final state is
$$\ket{\psi_f}_l=e^{i\frac{\pi}{4}\bigoplus_{i}\epsilon_i}\ket{\epsilon_1,\ldots,\epsilon_N}$$

Now, for the right hand side circuit, the first layer of $T$ gates adds a phase that is $e^{i\frac{\pi}{4}\sum_i\epsilon_i}$.
Then each $(\CS^\dagger)_{ij}$ gate adds a $e^{-i\frac{\pi}{2}}$ phases if and only if $\epsilon_i=\epsilon_j=1$ so the overall phase added by the $\CS^\dagger$ layer is $e^{-i\frac{\pi}{2}\sum_{i,j, i<j}\epsilon_i\epsilon_j}$.
Finally, the last $\CCZ$ layer adds a $e^{i\pi}=-1$ phase for each triplet $i,j,k$ such that $\epsilon_i\epsilon_j\epsilon_k=1$.
The final state is 
$$\ket{\psi_f}_r=e^{i\frac{\pi}{4}\left(\sum_i\epsilon_i-2\sum_{i,j, i<j}\epsilon_i\epsilon_j+4\sum_{i,j,k, i<j<k}\epsilon_i\epsilon_j\epsilon_k\right)}\ket{\epsilon_1,\ldots,\epsilon_N}$$.
Both states are the same if and only if 
\[
\bigoplus_{i}\epsilon_i
\equiv
\sum_i\epsilon_i
-2\sum_{i<j}\epsilon_i\epsilon_j
+4\sum_{i<j<k}\epsilon_i\epsilon_j\epsilon_k
\pmod{8}.
\]

Let 
\[
S_r=\sum_{1\le i_1<\cdots<i_r\le N}\epsilon_{i_1}\cdots\epsilon_{i_r}
\qquad (r=1,\dots,N).
\]

Consider the product
\[
P=\prod_{i=1}^N(1-2\epsilon_i).
\]
Since each $\epsilon_i$ is either $0$ or $1$, each factor is equal to $1$ or $-1$.
If $m=\sum_i\epsilon_i$, then
\[
P=(-1)^m.
\]
Hence
\[
\bigoplus_i \epsilon_i =
\begin{cases}
0, & m \text{ even},\\
1, & m \text{ odd},
\end{cases}
= \dfrac{1-(-1)^m}{2}
= \dfrac{1-P}{2}.
\]

Expanding the product $P$, we obtain
\[
\prod_{i=1}^N(1-2\epsilon_i)
= 1-2S_1+4S_2-8S_3+16S_4-\cdots+(-2)^N S_N.
\]
Therefore,
\[
\frac{1-P}{2}
= S_1-2S_2+4S_3-8S_4+\cdots+(-2)^{N-1}S_N.
\]
This is an exact identity of integers:
\[
\bigoplus_i\epsilon_i
= S_1-2S_2+4S_3-8S_4+\cdots.
\]

Reducing modulo $8$, all terms from $-8S_4$ onward vanish. Hence
\[
\bigoplus_i\epsilon_i
\equiv S_1-2S_2+4S_3 \pmod{8}.
\]
In expanded form,
\[
\bigoplus_{i}\epsilon_i
\equiv
\sum_i\epsilon_i
-2\sum_{i<j}\epsilon_i\epsilon_j
+4\sum_{i<j<k}\epsilon_i\epsilon_j\epsilon_k
\pmod{8}.
\]
\end{proof}
\subsection{General case of \texorpdfstring{$\sqrt[L]{T}$}{L-th square root of T}}
When commuting a $\sqrt{T}$ gate through the CNOT ladder decoding circuit of a repetition code, one obtains \Cref{circ:commutations_sqrt_T}.
More generally, for any $L\geq1$ and $N\geq L+3$ integers, commuting $\sqrt[L]{T}$ through the CNOT ladder yields the circuit of \Cref{circ:commutations_sqrtK_T}.
\begin{widetext}
    \begin{equation}
    \begin{quantikz}
        & \gate[4]{\mathcal{E}^\dagger} & \gate{\sqrt{T}} & \gate[4]{\mathcal{E}} & \\
        &&&&\\
        \setwiretype{n} & & \vdots & & \\
        & & & & \\
    \end{quantikz}=
     \begin{quantikz}
        & \gate{\sqrt{T}}& \gate[3,disable auto
    height]{\prod_{i,j} (\CT^{\dagger})_{ij}} & \gate[3]{\prod_{i,j,k} (\CCS)_{ijk}} & \gate[3]{\prod_{i,j,k,l} (\CCCZ)_{ijkl}}&&\\
     \setwiretype{n} & \vdots &&&& \vdots &\\
        & \gate{\sqrt{T}} &&&& &\\ 
    \end{quantikz}
    \label{circ:commutations_sqrt_T}
    \end{equation}
    \begin{equation}
    \begin{quantikz}
        & \gate[4]{\mathcal{E}^\dagger} & \gate{\sqrt[L]{T}} & \gate[4]{\mathcal{E}} & \\
        &&&&\\
        \setwiretype{n} & & \vdots & & \\
        & & & & \\
    \end{quantikz}=
     \begin{quantikz}
        & \gate{\sqrt[L]{T}}& \gate[3,disable auto
    height]{ \prod_{\text{all}} \C{\sqrt[L-1]{T}}^{\dagger}} & \gate[3]{ \prod_{\text{all}} \C\C{\sqrt[L-2]{T}}} && \gate[3]{\prod_{\text{all}} (\underbrace{\C\ldots \C}_{L+2}Z)}&&\\
     \setwiretype{n} & \vdots &&&\ldots&& \vdots &\\
        & \gate{\sqrt[L]{T}} &&&&& &\\ 
    \end{quantikz}
    \label{circ:commutations_sqrtK_T}
    \end{equation}
\end{widetext}
In \Cref{circ:commutations_sqrtK_T}, layers of gates consist of all possible $\underbrace{\C\ldots \C}_{r}{\sqrt[L-r]{T}}^{\dagger}$ gates for odd values of $r\in [L+2]$ and $\underbrace{\C\ldots \C}_{r}\sqrt[L-r]{T}$ for even values of $r\in [L+2]$.

\begin{proof}
The proof is the same as for the $T$ gate case above, and comes from the fact that for any $\epsilon_1, \ldots \epsilon_N \in \mathbb{F}_2^N$, the left hand side circuits acts as \[\ket{\epsilon_1, \ldots \epsilon_N}\rightarrow e^{i \frac{\pi}{2^{L+2}}\bigoplus_i\epsilon_i}\ket{\epsilon_1, \ldots \epsilon_N}\]
whereas the right hand side circuits does
\begin{equation}
\begin{aligned}
\ket{\epsilon_1, \ldots \epsilon_N} & \rightarrow \\ 
& e^{i \frac{\pi}{2^{L+2}} \left(\sum_{r=1}^{L+3}\sum_{1\le i_1<\cdots<i_r\le N}(-2)^{r-1}\epsilon_{i_1}\cdots\epsilon_{i_r}\right)}\ket{\epsilon_1, \ldots \epsilon_N} 
\end{aligned}
\end{equation}
as each layer of the circuit adds a phase $\sum_{1\le i_1<\cdots<i_r\le N}(-2)^{r-1}\epsilon_{i_1}\cdots\epsilon_{i_r}$.
Therefore, the equality of the two circuits relies on the integer equality proven above
\[
\bigoplus_i\epsilon_i
= \sum_{r=1}^{N}(-2)^r\sum_{1\le i_1<\cdots<i_r\le N}\epsilon_{i_1}\cdots\epsilon_{i_r} \mod{2 ^{L+3}} 
\]
because phases proportional to $2\pi$ can be neglected.
\end{proof}
\section{From logical gate to distillation} 
\subsection{Gate sequence transformation}\label{App:V_to_L}

Here, we derive the commutation relation of a $P_{\pi/8}$ gate through a repetition code decoding circuit.
As in the main text, we denote by $\mathcal{E}^\dagger$ the sequence of CNOT gates implementing the decoding circuit. 
We recall that 
$P_{\pi/8} = \cos(\frac{\pi}{8}) \mathbb{I}-i \sin(\frac{\pi}{8})P$, and we can associate a vector $\alpha\in\{0,1\}^N$ such that $P = \bigotimes_i Z_i^{\alpha_i}$.
We prove that there exists another $\pi/8$ rotation gate $\tilde P_{\pi/8}$ such that \[\tilde P_{\pi/8}=\mathcal{E}P_{\pi/8}\mathcal{E}^{\dagger}.\]
To do so we first use the fact that $\mathcal{E}\mathbb{I}\mathcal{E}^\dagger = \mathbb{I}$, thus preserving the $\cos(\pi/8)$ part of $P_{\pi/8}$. For the $\sin(\pi/8)$ part, we use the following properties of the commutation between $Z$ and CNOT:
\begin{equation*}
     \begin{quantikz}
&\targ{}  &  & \gate{Z^{\alpha_i}} &       &\targ{}&\\
&\ctrl{-1}&  & \gate{Z^{\alpha_j}} &       & \ctrl{-1}&
\end{quantikz} = \begin{quantikz}
  & \gate{Z^{\alpha_i}} &      \\
 & \gate{Z^{\alpha_j\oplus \alpha_i}} &      
\end{quantikz}
\end{equation*}
It becomes clear how to commute $\mathcal{E}$ through $P$ as this generalizes to
\begin{widetext}
\begin{equation}
\begin{quantikz}
&\targ{}  & \ldots  &        &   \gate{Z^{\alpha_1}} &       &\ldots& \targ{}&\\
&\ctrl{-1}& \ldots  &        &   \gate{Z^{\alpha_2}} &       &\ldots& \ctrl{-1}&\\
   &  \vdots \setwiretype{n}&  && \vdots                &       &      &  \vdots &\\
&         & \ldots  &\targ{} &\gate{Z^{\alpha_{N-1}}}&\targ{}&\ldots&&\\
&         & \ldots  &\ctrl{-1}&\gate{Z^{\alpha_{N}}} &\ctrl{-1}&\ldots&&\\
\end{quantikz} \quad  = \quad \begin{quantikz}
 && \gate{Z^{\alpha_1}} &&\\
 \\
 &&\gate{Z^{\alpha_{2}\oplus\alpha_{1}}}&&\\
 \setwiretype{n} & \vdots && \vdots &\\
 \\
&&\gate{Z^{\alpha_{N-2}\oplus\alpha_{N-1}}}&&\\
 &&\gate{Z^{\alpha_{N-1}\oplus\alpha_N}}&&\\
\end{quantikz}
\end{equation}
\end{widetext}
If we associate a bit-sequence $\beta$ to $\tilde P$, i.e.~$\tilde P = \bigotimes_i Z_i^{\beta_i}$, it is now explicit that
\begin{align}\label{eq:from_alpha_to_beta}
\beta_1 &= \alpha_1\\
\forall 2\leq i\leq N, \beta_i &= \alpha_i \oplus\alpha_{i-1}  
\end{align}
This relation can be inverted as
\begin{align}
\alpha_1 &= \beta_1 \label{eq:from_beta_to_alpha1}\\
\label{eq:from_beta_to_alpha2} \forall 2\leq i\leq N, \alpha_i &= \beta_1 \oplus \ldots \oplus \beta_i  
\end{align}

\subsection{SAT formulation equivalence}\label{App:SAT_L_problem}

A matrix $V=\begin{pNiceMatrix}
\alpha^0 & \alpha^1 & \cdots & \alpha^n
\end{pNiceMatrix}$ that obeys 
\begin{align*}
        & \forall i\in[\![1,N]\!], \sum_k \alpha_i^k \equiv 1 \mod 2, \\
        & \forall i<j \in[\![1,N]\!]^2, \sum_k \alpha_i^k\alpha_j^k \equiv 1 \mod 2, \\
        & \forall i<j<l\in[\![1,N]\!]^3, \sum_k \alpha_i^k\alpha_j^k\alpha_l^k \equiv 1 \mod 2. 
\end{align*} 

provides a gate sequence $P^1,\ldots, P^n $ (defined from the $\alpha$'s in \Cref{eq:gate_sequence_P_from_alpha}) that implements a logical $T$ gate (up to Clifford corrections) on the repetition code.
To use this logical $T$ gate to perform a distillation protocol, one can choose to apply it to the $\ket{+}_L$ state and further perform the decoding circuit, which is equivalent to mapping the gate sequence $ \alpha^1, \ldots, \alpha^n$ over the gate sequence $ \beta^1, \ldots, \beta^n$ as explained in \Cref{App:V_to_L}.

Injecting the relations \Cref{eq:from_beta_to_alpha1} and (\ref{eq:from_beta_to_alpha2}), the first condition becomes
\[ \forall i\in[\![1,N]\!], \sum_k \beta^k_1 \oplus \ldots \oplus \beta^k_i  \equiv 1 \mod 2.\]
For $i=1$, this directly gives 
\[\sum_k \beta^k_1  \equiv 1 \mod 2.\]
For $i=2$, it gives \[\sum_k \beta^k_1\oplus \beta^k_2 =  \left(\sum_k \beta^k_1\right)\bigoplus \left(\sum_k\beta^k_2\right)\equiv 1 \mod 2,\]
and by summing both, we obtain that 
\[\sum_k \beta^k_2  \equiv 0 \mod 2.\]
By induction on $i$, it is easy to show that 
\[\forall i\in[\![2,N]\!],\sum_k \beta^k_i  \equiv 0 \mod 2.\]

Similarly, the second condition (the one on pairs of lines) becomes 
\begin{multline*}
    \forall i<j\in[\![1,N]\!], \\\sum_k (\beta^k_1 \oplus \ldots \oplus \beta^k_i)(\beta^k_1 \oplus \ldots \oplus \beta^k_j)  \equiv 1 \mod 2\\
    \Leftrightarrow \sum_{\substack{1\leq l \leq i\\ 1\leq m \leq j}}\sum_k \beta_l^k\beta_m^k \equiv 1 \mod 2,
\end{multline*} 
which for $i=1, j=2$ gives 
\[\sum_k\beta_1^k\oplus \beta_1^k\beta_2^k= 1 \mod 2.\]
Summing with the condition on line 1, we obtain 
\[\sum_k \beta_1^k\beta_2^k= 0 \mod 2.\]
Fixing $i=1$ and doing an induction on $j$, it is easy to show that 
\[\forall j \in [\![2,N]\!], \sum_k \beta_1^k\beta_j^k= 0 \mod 2.\]
For $i=2,j=3$, the condition reads 
\begin{multline*}
    \sum_k \beta_1^k \oplus \beta_1^k \beta_2^k\oplus \beta_1^k \beta_3^k\oplus \beta_2^k \beta_1^k\oplus \beta_2^k \oplus \beta_2^k \beta_3^k\equiv 1 \mod 2
\end{multline*}
Using the relation above, we can directly deduce that 
\[\sum_k \beta_2^k\beta_3^k\equiv 0 \mod 2.\]
Similarly, we can proceed by induction on $j$ to cover all terms of the form $\beta_2\beta_j$.

By doing so on all $i$ in increasing order, it is straightforward to prove by induction on $i$ that 
\[\forall i<j,\sum_k \beta_i^k\beta_j^k\equiv 0 \mod 2. \]

For the triplet of lines condition, the proof goes the same way as for the pairs, but instead of interlocking two inductions, one has to interlock 3 inductions (respectively on $i,j,l$) using 
\[\forall i<j<l, \sum_{\substack{1\leq p \leq i\\1\leq q \leq j\\1\leq r \leq l}}\beta_p^k\beta_q^k\beta_r^k\equiv 1 \mod 2.\]
One can prove this way that 
\[\forall i<j<l,\sum_k \beta_i^k\beta_j^k\beta_l^k\equiv 0 \mod 2.\]

For the distance condition, by applying the transformation \Cref{eq:from_alpha_to_beta} on both sides of  
\begin{equation}
    d = \min |K| \text{ s.t. } \sum_{k\in K} \alpha^k \equiv \left(\begin{array}{c}
         1 \ \mod 2  \\
         1 \ \mod 2 \\
         \vdots\\
         1 \ \mod 2
    \end{array}\right) 
\end{equation}
we get that 
\begin{equation}
    d = \min |K| \text{ s.t. } \sum_{k\in K} \beta^k \equiv \left(\begin{array}{c}
         1 \ \mod 2  \\
         0 \ \mod 2 \\
         \vdots\\
         0 \ \mod 2
    \end{array}\right).
\end{equation}
Said in another way, a logical error in the repetition code is a $Z$ on every line, when commuting this through the decoding CNOT ladder operator, it becomes a single $Z$ on the first qubit.

%% file: Algebras_v2.tex
\section{Canonical families of protocols : General framework}\label{App:General_Algebra}

\subsection{General preliminaries for the proofs}\label{App:analytical_framework}

In this section, we work with the partially ordered set $(\mathcal{P}([N]), \subseteq)$ and leverage the fact that a multi-qubit $Z$ rotation can be seen as an element $A \in \mathcal{P}([N])$ where the indices in $A$ are the qubits in the support of the gate.
We gather a series of definitions and reminders from the algebra of order theory, which will be used in the next section \Cref{App:Proofs_of_lemmas_and_ths}.

\begin{definition}[Convolution product]
    Given the partially ordered set $(\mathcal{P}([N]), \subseteq)$ and a ring $\mathcal{R}$, one can define the incidence algebra $I$, which is the set of functions from $\mathcal{P}([N]) \times \mathcal{P}([N])$ to $\mathcal{R}$.
    The addition on the incidence algebra is pointwise, and the product is the convolution product: 
    \[\forall \nu, \eta \in I, \ (\nu\ast \eta)(x,z) = \sum_{x\leq y \leq z} \nu(x,y)\eta(y,z).\]
\end{definition}

\begin{definition}[zeta function]
    We define the zeta function on $\mathcal{P}([N])\times \mathcal{P}([N])$ as 
    \[\zeta(A,B) = \begin{cases}
        1 \text{ if }A \subseteq B,\\
        0 \text{ otherwise.}
    \end{cases}\]
\end{definition}
This function will help us translate the conditions \Cref{eq:cons_T} in a natural manner.
Indeed, the conditions read as ``every singlet, pair, triplet of lines should appear an odd number of times in the support of the gates in the circuit''.
The zeta function can be seen as a filter returning 1 when the gate of support $B$ includes the set of lines $A$ in its support and 0 otherwise.
To express the conditions, it is natural to apply this filter over a boolean function $f$ characterizing the selected gates (hereby represented by the selected elements of $\mathcal{P}([N])$).
To do so, we define the Star product.
\begin{definition}[Star product]
    We define the $\star$ product as the natural action of the incidence algebra $I$ over the set of functions $F$ from $\mathcal{P}([N])$ to $\mathcal{R}$: 
    For $f\in F$, $\nu\in I$ and $x \in\mathcal{P}([N])$:
    \[(\nu\star f)(x)=\sum_{y \geq x}\nu(x,y)f(y).\]
\end{definition}

For a set $A\subseteq[N]$, $\zeta \star f(A)$ is the sum of $f(B)$ over all the gates $B$ that contain $A$ in their support.
Using this star product, we will express the constraints (e.g \Cref{eq:cons_T}) by fixing some values of $\zeta \star f$.
We will now introduce some tools to invert those relations and deduce how the constraints over $\zeta \star f$ can be translated as constraints over $f$.

\begin{definition}[Möbius function]\label{def:mobius}
    The Möbius function $\mu$ is defined as the inverse of the zeta function for the convolution product:
    \begin{equation}
        \mu \ast \zeta = \zeta \ast \mu = \delta
    \end{equation}
    where \[\delta (A,B)= \begin{cases}
        1 &\text{ if A=B},\\
        0 &\text{ otherwise}.
    \end{cases}\] is the neutral element for $\ast$ and $\star$.
    
     For the Boolean lattice $(\mathcal{P}([N]), \subseteq)$, the Möbius function is 
    \begin{equation}
        \mu(A,B)=\begin{cases}
        (-1)^{|B|-|A|} &\text{ if } A \subseteq B,\\
        0 &\text{ otherwise.}
    \end{cases}
    \end{equation}
\end{definition}

\begin{proof}
    In our context, we always work on a commutative ring so the incidence algebra is commutative as well, it is sufficient to show that $\mu \ast \zeta = \delta$.
    As the inverse is unique, it is enough to show that the function defined by the above formula works.
    First, for any $A\subseteq B$,
    if $A=B$, it is straightforward that $ \mu \ast \zeta (A,B)=1$.
    Otherwise, we start by using the binomial theorem to recall that 
    $$\sum_{x \subseteq [N]}(-1)^{|x|}=\sum_{k=0}^N\binom{N}{k}(-1)^k=(1-1)^N=0.$$
    Then we compute,
    \begin{align*}
        \mu \ast \zeta (A,B)&=\sum_{A\subseteq C\subseteq B }\mu(A,C)\\
        &=(-1)^{|A|}\sum_{A\subseteq C\subseteq B }(-1)^{-|C|}\\
        &= \sum_{C'\subseteq B\backslash A }(-1)^{-|C'|}\\
        &=0
    \end{align*}
    Therefore, $\mu \ast \zeta =\delta$.
\end{proof}

\begin{proposition}\label{prop:composition_star_conv}
    For any element of the incidence algebra $\nu,\eta\in I$ and any function $f\in F$, we have 
    \begin{equation}
        \nu\star (\eta \star f)=(\nu \ast \eta) \star f
    \end{equation}
\end{proposition}

\begin{proof}
    Let $x\in \mathcal{P}([N])$, 
    \begin{align*}
        (\nu\star (\eta \star f))(x) &= \sum_{y \geq x}(\eta\star f )(y) \nu(x,y)\\
        &=\sum_{y \geq x}\sum_{z \geq y}f(z)\eta(y,z) \nu(x,y)\\
        &=\sum_{z\geq x}\left(\sum_{\substack{y \text{ s.t}\\x\leq y \leq z }}\eta(y,z) \nu(x,y)\right) f(z)\\
        &=\sum_{z\geq x}(\nu \ast \eta)(x,z) f(z)\\
        &= ((\nu \ast \eta)\star f)(x)
    \end{align*}
\end{proof}
Using this proposition, it will be easy to naturally deduce constraints over $f$ from the triorthogonality-like constraints \Cref{eq:cons_T} that are expressed as constraints over $\zeta \star f$ via the Möbius function.

\subsection{Canonical protocols proofs}\label{App:Proofs_of_lemmas_and_ths}
\subsubsection{Proof of \Cref{lemma:constraints}}\label{Proof_lemma_1}
The proof of \Cref{lemma:constraints} is direct and consists of reintroducing the bitvectors $\alpha^k$ associated with each possible gate and writing the definition of $g$ through those $\alpha^k$.

\begin{proof}[\Cref{lemma:constraints} Proof]
    By definition of $g$, for any $B \subseteq [N]$, 
    \[g(B)=|\{A\in \mathcal{F_N}, B \subseteq A\}| \mod 2.\]
    If for all $A\in \mathcal{F_N}$ we define a vector $\alpha^k\in \mathbb{F}_2^N$ by \[\alpha^k_i=\begin{cases}
        1 \text{ if } i \in A,\\
        0 \text{ otherwise.}
    \end{cases}\]
    then the condition exactly reads as \Cref{eq:cons_T}.
    
     For example, the conditions on $B\subseteq [N]$ such that $|B|=1$ become 
    \[\forall i\in [N], g(\{i\})=|\{A\in \mathcal{F_N},\{i\} \subseteq A\}|=\sum_k \alpha^k_i.\]
    For the doublets and triplets, the logic is the same because $\forall i<j \in [N], \alpha_i^k\alpha_j^k=\begin{cases}
        1 \text{ if } i,j \in A,\\
        0 \text{ otherwise.}
    \end{cases}$

    To invert the relation between $f$ and $g$, we use the fact that $g=\zeta \star f$ so that using \Cref{prop:composition_star_conv}, $\mu \star g=\mu \star (\zeta \star f)= (\mu \ast \zeta) \star f = f$
    which directly gives  \[\forall A \subseteq [N], f(A) = \sum_{B\supseteq A}g(B).\]
\end{proof}

\subsubsection{Proof of \Cref{th:canonical_T}}\label{Proof_th1}
We provide the proof of \Cref{th:canonical_T}.
The proof is straightforward and consists of applying the definition of the introduced tools; the only non trivial part of the proof is the distance formula.

\begin{proof}[\Cref{th:canonical_T} Proof]
    We begin by the expression of $f_0$.
    Let $A\subseteq [N]$, using \Cref{lemma:constraints}, we have $f_0(A)=\sum_{B\supseteq A}g_0(B)$.
    \begin{itemize}
        \item If $|A|>3$, $f_0(A)=0$ because $g(B)=0$ for any $B$ containing $A$.
        \item If $|A|=3$, $f_0(A)=\sum_{B\supseteq A}g_0(B)=g_0(A) =1$.
        \item If $|A|=2$, \begin{align*}
            f_0(A)&=\sum_{B\supseteq A}g_0(B)\\
            &=g_0(A) +\sum_{i\in [N]\backslash S}g_0(A\cup\{i\})\\
            &=1+ N-2 \\
            &= N-1 \mod 2 .
        \end{align*}
        \item If $|A|=1$, \begin{align*}
            f_0(A)&=\sum_{B\supseteq A}g_0(B)\\
            &=g_0(A) +\sum_{i<j\in [N]\backslash S}g_0(A\cup\{i,j\})+\sum_{i\in [N]\backslash S}g_0(A\cup\{i\})\\
            &=1+ \binom{N-1}{2} +N-1 \\
            &= N+ \binom{N-1}{2} \mod 2 .
        \end{align*}
    \end{itemize}
We introduce $a_1, a_2, a_3\in \mathbb{F}_2$ that are the parity of $f$ respectively on sets of cardinality $1, 2$ and $3$.
It is always the case that $a_3=1$.
In practice, the value of $a_1,a_2$ only depends on the value of $N \mod 4$.
Indeed it is easy to show that \begin{equation*}
    \begin{split}
        & a_1=\begin{cases}
1 & \text{if } N\equiv 0\mod 4 \text{ or }N\equiv 1\mod 4,\\
0 & \text{if } N\equiv 2\mod 4 \text{ or } N\equiv 3\mod 4,\\
\end{cases}\\
& a_2=\begin{cases}
1 & \text{if } N\equiv 0\mod 4 \text{ or }N\equiv 2\mod 4,\\
0 & \text{if } N\equiv 1\mod 4 \text{ or } N\equiv 3\mod 4.\\
\end{cases}
    \end{split}
\end{equation*}
Then, using the fact that $n= \binom{N}{1}a_1+\binom{N}{2}a_2+\binom{N}{3}a_3$, we deduce the $T$-count of the protocol.

Lastly, for the distance we can build patterns of gates in $\mathcal{F}_N^0$ for which the combined errors is a logical $Z$ operator:
\begin{itemize}
    \item when $N\equiv 0 \mod 4$ or $N\equiv 2 \mod 4$, 
    \begin{itemize}
        \item if $N$ can be written as $N=3q$ then we can build the pattern with $q$ disjoint triplets as $a_3=1$.
        \item if $N=3q+1$, we build it with $q-1$ triplets and two doublets as $a_2=a_3=1$.
        \item if $N=3q+2$, we build it with $q$ triplets and one doublet.
    \end{itemize}
    Overall the smallest pattern is always of cardinal $\lceil \frac{N}{3}\rceil$ when $N$ is even.
    \item when $N\equiv 1 \mod 4$, 
    \begin{itemize}
        \item if $N$ can be written as $N=3q$ then we can build the pattern with $q$ disjoint triplets as $a_3=1$.
        $q$ is indeed the smallest odd integer above $N/3$ because if $q$ is even then $N$ is even which is not possible.
        \item if $N=3q+1$, we build it with $q$ triplets and one singleton as $a_1=a_3=1$.
        $q+1$ is indeed the smallest odd integer above $N/3$ because if $q$ is odd then $N$ is even which is not possible.
        \item if $N=3q+2$, we build it with $q$ triplets and two singleton as $a_1=a_3=1$.
        $q+2$ is indeed the smallest odd integer above $N/3$ because if $q$ is even then $N$ is even which is not possible.
    \end{itemize}
    \item when $N\equiv 3 \mod 4$, there are only triplets in $\mathcal{F}_N^0$.
    \begin{itemize}
        \item if $N$ can be written as $N=3q$ then we can build the pattern with $q$ disjoint triplets as $a_3=1$.
        $q$ is indeed the smallest odd integer above $N/3$ because if $q$ is even then $N$ is even which is not possible.
        \item if $N=3q+1$, we build it with $q-2$ triplets of separated supports plus 3 triplets for the remaining 7 qubits.
        Indeed, a pattern where the three gates all share exactly one qubit in their support works.
        $q+1$ is indeed the smallest odd integer above $N/3$ because if $q$ is odd then $N$ is even which is not possible.
        \item if $N=3q+2$, we build it with $q-1$ triplets of separated supports plus two triplets sharing one qubit for the 5 remaining qubits.
        $q+2$ is indeed the smallest odd integer above $N/3$ because if $q$ is even then $N$ is even which is not possible.
    \end{itemize}.
\end{itemize}
\end{proof}

\subsection{\texorpdfstring{$\sqrt{T}$}{sqrt(T)} distillation and beyond}\label{App:sqrt(T)}

The analytical framework introduced in~\Cref{App:analytical_framework} allows us to go beyond the distillation of $\ket{T}$ states.
In particular, we explicitly develop the framework for $\ket{\sqrt{T}}$ in this Appendix and explain how this naturally expands to $\ket{\sqrt[L]{T}}=R_Z\left(\frac{\pi}{2^{L+2}}\right)\ket{+}$.

We start by introducing the logical $\sqrt{T}$ gate on the repetition code.
\begin{equation}
    \begin{quantikz}
        & \gate[4]{\mathcal{E}^\dagger} & \gate{\sqrt{T}} & \gate[4]{\mathcal{E}} & \\
        &&&&\\
        \setwiretype{n} & & \vdots & & \\
        & & & & \\
    \end{quantikz} =: \begin{quantikz}
        & \gate[3]{{\sqrt{T}}_L} & \\
        \setwiretype{n} &  & \\
        & &  \\
    \end{quantikz}
\end{equation}\label{eq:def_Z_pi/16}
As in \Cref{sec:FT_T_on_rep_code}, we can commute the $\sqrt{T}$ gate through the CNOT ladder and show that the circuit is equivalent to applying $\sqrt{T}$ on every line once, applying $\mathrm{C}T^{\dagger}$ on every pair of lines, applying $\mathrm{CC}S$ on every triplet, and applying $\mathrm{CCC}Z$ on every quadruplet; see \Cref{App:Commutations_relations} for the proof.

Similarly to \Cref{circ:Tinject} we can define the injection circuit of a $P_{\frac{\pi}{16}}$ gate as \begin{equation}\label{eq:sqrt(T)_injection}
    \begin{quantikz}
    	\ & \gate[5, disable auto height, style={fill=measurebg, rounded corners}]{\begin{array}{c}
     \\  \\ P \\ \\  \\ \\ \\ \\ Z
    \end{array}} & \gate[3]{P_{\frac{\pi}{8}}} &\gate[3]{P_{\frac{\pi}{2}}}& \\
        \ldots\setwiretype{n} &           &              &\\
                             &           &              &   &\\
        \setwiretype{n}&   &\cwbend{-1}     &\\
    	\ket{\sqrt{T}} \ &   & \gate{H}&\meter[style={fill=measurebg, rounded corners}]{} \wire[u][2]{c}
    \end{quantikz}
\end{equation}
Note that this circuit uses a $\ket{\sqrt{T}}$ and involves a $P_{\frac{\pi}{8}}$ correction with 50\% probability.
Such a correction is non-Clifford and requires either one $\ket{T}$-state or two $\ket{\sqrt{T}}$-states.
Therefore, we can say that the $\sqrt{T}$-count of one $P_{\frac{\pi}{16}}$ via this injection circuit is $2$.
Obviously, each of these $P_{\frac{\pi}{16}}$ could be implemented with a single $\ket{\sqrt{T}}$-state and many CNOT gates using the circuit of \Cref{eq:def_Z_pi/16}.
However, at the logical level, this could slow down the implementation of the circuit, as each logical CNOT will typically require two lattice surgery operations.

The conditions for a sequence of $P_{\frac{\pi}{16}}$ gates of $Z$ Pauli strings $\{P^1, \ldots, P^n \}$ defined through the binary vectors
\begin{equation}\label{eq:gate_sequence_P_from_alpha_2}
    \alpha^k_i = \begin{cases} 0 & \text{if } (P^k)_i = I, \\ 1 & \text{if } (P^k)_i = Z. \end{cases}
\end{equation} to implement a logical $\sqrt{T}$ gate on the repetition code, up to Clifford corrections, are:
\begin{subequations}
    \begin{align}
        & \forall i\in[\![1,N]\!], \sum_k \alpha_i^k \equiv 1 \mod 4, \\
        & \forall i<j \in[\![1,N]\!]^2, \sum_k \alpha_i^k\alpha_j^k \equiv 1 \mod 4, \\
        & \forall i<j<l\in[\![1,N]\!]^3, \sum_k \alpha_i^k\alpha_j^k\alpha_l^k \equiv 1 \mod 4\\
        & \forall i<j<l<m\in[\![1,N]\!]^4, \sum_k \alpha_i^k\alpha_j^k\alpha_l^k\alpha_m^k \equiv 1 \mod 2. 
    \end{align}
\end{subequations}
The first three conditions are modulo $4$ conditions because the three gates $\sqrt{T}^2=T$, $(\mathrm{C}T^\dagger)^2=\mathrm{C}S^\dagger$ and $(\mathrm{CC}S)^2=\mathrm{CC}Z$ are still non-Clifford gates when squared whereas ${\mathrm{CCC}Z}^2=\mathbb{I}^{\otimes 4}$.
However, $\sqrt{T}$, $\mathrm{C}T^\dagger$ and $\mathrm{CC}S$ become Clifford gates only when elevated to the fourth power.
As a consequence, it could be useful to apply the same gate two or three times.
In practice, this would result in a single gate injecting $\ket{T}$ or $\ket{T^{3/2}}$ respectively.
To identify whether a gate is supposed to inject $\ket{\sqrt{T}}$, $\ket{T}$, or $\ket{T^{3/2}}$, we use a function $f$ valued in the ring $\mathcal{R}=\mathbb{Z}\backslash 4\mathbb{Z}$.

\begin{lemma}\label{prop:f_from_g_2}
    Let $f,g:\mathcal{P}([N])\rightarrow\mathbb{Z}\backslash 4\mathbb{Z}$ be two functions such that \[\forall B \subseteq [N], g(B) = \sum_{A\supseteq B}f(A).\] with
    \[
g(B) = \begin{cases}
    1 \mod 4& \text{if } 1 \leq |B|\leq 3, \\
    1 \mod 2& \text{if } |B|= 4. \\
\end{cases}
\]
    Then, $f$ defines a family of circuits implementing a logical $\sqrt{T}$ gate on an $N$-qubit repetition code up to a Clifford correction. Additionally, we have  
    \[\forall A \in \mathcal{P}([N]), f(A) = \sum_{B\supseteq A}(-1)^{|B|-|A|}g(B) \mod 4.\]
\end{lemma}
\begin{proof}
    The proof is the same as the one of Lemma 1 provided in \Cref{Proof_lemma_1}.
\end{proof}

The remaining degrees of freedom on $f$ are given by the values of $g$ in sets of cardinality greater than $4$.
If we denote $\mathcal{E}_N=\{ A, f(A)\neq 0 \mod 4\}$ as the set of possible $Z$ error patterns arising from the gates in the protocol, the distance is still defined as the minimal number of erroneous gates in $\mathcal{E}_N$ that combine into the undetected pattern of $Z$ errors on every qubit.
The same way as for $\ket{T}$-state distillation, we can build a canonical family $\mathcal{F}^1_N$ for which $g(B)=1 \mod 4$ for all $1\leq |B|\leq 4$ and $g(B)=0 \mod 4$ for $|B|>4$. Note that one could also choose $g(B)=3 \mod 4$ for $|B|=4$.

\begin{theorem}[Canonical family of logical $\sqrt{T}$ circuits]\label{th:sqrt(T)}
    
    For $N\geq 5$, the family $\mathcal{F}^1_N$ is defined by:
    \begin{multline*}
        \forall A\subseteq[N], \\f_1(A)\equiv\begin{cases}
            (-1)^{|A|}\binom{N-|A|-1}{4-|A|} \mod 4&\text{ if } 1\leq |A|\leq 4,\\
             0 \mod 4 & \text{ otherwise.}
        \end{cases}
    \end{multline*} 
    The distance of this family of protocols is 
    \[d(\mathcal{F}_N^1)=\lceil \frac{N}{4}\rceil.\]
    The $\sqrt{T}$-count is 
    \[n= 2\sum_{s=1}^4\binom{N}{s}(1-\delta_{\binom{N-s-1}{4-s} \mod 4}),\]
    where the factor 2 accounts for the non-Clifford correction in the injection circuit.
\end{theorem}
\begin{proof}[\Cref{th:sqrt(T)} Proof]
    The proof is the same as for \Cref{th:canonical_T}.
    Using \Cref{prop:f_from_g_2}, given $A\subseteq[N]$, it directly follows that $f_1(A)=0$ if $|A|>4$ and 
    \[f_1(A)=\sum_{r=0}^{4-|A|}(-1)^r\binom{N-|A|}{r}\] otherwise.
    From Pascal's formula $\binom{X}{r}=\binom{X-1}{r}+\binom{X-1}{r-1}$, one can deduce by summing over $r$ on each side and telescoping the right-hand side that for any positive integer~$m$:
    \begin{align*}
    \sum_{r=0}^m(-1)^r&\binom{X}{r}=\sum_{r=0}^m(-1)^r\binom{X-1}{r}+\sum_{r=0}^m(-1)^r\binom{X-1}{r-1}\\
    &=\sum_{r=0}^m(-1)^r\binom{X-1}{r}+\sum_{r=0}^{m-1}(-1)^{r+1}\binom{X-1}{r}\\
    &= (-1)^m\binom{X-1}{m}.
    \end{align*}
    Applying this to $f$ above, we have
    \[f_1(A)=(-1)^{4-|A|}\binom{N-|A|-1}{4-|A|},\]
    giving the expected result as $(-1)^{4-|A|}=(-1)^{|A|}$.

    We now introduce $a_1,a_2,a_3,a_4\in \mathbb{F}_2$ as \[a_r=\begin{cases}
        1 &\text{if } f([r])\neq 0 \mod 4, \\
        0 &\text{otherwise. }
    \end{cases}\]
    Indeed, the error model is the same for a multi-qubit $Z$ rotation, no matter whether it rotates an angle $\frac{\pi}{8}$ (corresponding to $f(A)=1 \mod 4$), $\frac{\pi}{4}$ (corresponding to $f(A)=2 \mod 4$), or $\frac{3\pi}{8}$ (corresponding to $f(A)=3 \mod 4$).
    The $\sqrt{T}$-count is the same as well: 2 for any rotation.

    One can exhaustively compute that 
    \begin{itemize}
        \item $\forall N >5, a_4=1$.
        \item $a_3= \begin{cases}
            0 & \text{if } N\equiv 0 \mod 4\\
            1 &  \text{otherwise }
        \end{cases}$.
        \item $a_2= \begin{cases}
            0 & \text{if } N\equiv 3,4 \mod 8\\
            1 &  \text{otherwise }
        \end{cases}$. Indeed, $\binom{N-3}{2}=\frac{(N-3)(N-4)}{2}$ is divisible by 4 if and only if $N\equiv 3,4 \mod 8$ because only one of two consecutive integers can be even, so one of the two has to be divisible by $8$.
        \item $a_1= \begin{cases}
            0 & \text{if } N\equiv 0,2,3,4,6 \mod 8\\
            1 &  \text{otherwise }
        \end{cases}$. Indeed $\binom{N-2}{3}=\frac{(N-2)(N-3)(N-4)}{3\times 2}$. First, one of the three consecutive integers is divisible by 3.
        Then, if $N$ is even, either $N-4$ or $N-2$ is divisible by $4$ and the other one is divisible by $2$.
        If $N$ is odd, $(N-2)(N-3)(N-4)$ is divisible by 8 if and only if $(N-3)\equiv 0 \mod 8$.
        
    \end{itemize}
    The only non trivial thing to prove is the distance of the protocol.
    As no set with cardinality $>4$ belongs to the family, we know $d(\mathcal{F}^1_N)\geq\lceil N/4\rceil$ because we can't partition $[N]$ with less than $\lceil N/4\rceil$ sets of cardinality $\leq 4$.
    We can explicitly build patterns to ensure $d(\mathcal{F}^1_N)=\lceil N/4\rceil$ .
    \begin{itemize}
        \item if $N=4q$, we can take $q$ different quadruplets to partition $[N]$.
        \item if $N=4q+1$,  we can use $q-1$ quadruplets, 1 triplet, and 1 pair.
        \item if $N=4q+2$,  we can use $q$ quadruplets and 1 pair.
        \item if $N=4q+3$,  we can use $q$ quadruplets and 1 triplet.
    \end{itemize}
    which ends the proof.
\end{proof}

\begin{table}[]
    \centering
    \renewcommand{\arraystretch}{1.25}
    \setlength{\tabcolsep}{10pt}
    \begin{tabular}{@{}ccc@{}}
    \toprule
    \textbf{Number of qubits} $N$ & \textbf{$\sqrt{T}$-count} $n$ & \textbf{Distance} $d$ \\
    \midrule
    5  & $2 \times 60$   & 2 \\
    9  & $2 \times 510$  & 3 \\
    13 & $2 \times 2184$ & 4 \\
    \bottomrule
    \end{tabular}
    \caption{Parameters of the $\ket{\sqrt{T}}$ distillation protocols from the canonical family for small $N$.}
    \label{tab:canonical_family_sqrtT}
\end{table}
Unfortunately, the canonical family provides very deep circuits; see~\cref{tab:canonical_family_sqrtT}. Indeed, taking one column of weight $r$ implies taking all of them.
We believe they provide insights regarding what distance can be achieved for a given number $N$ of logical qubits, but they are not of practical use.

Regarding $\ket{\sqrt{T}}$ distillation circuits, these families are, to our knowledge, the first to distill such states for arbitrary order $d$, even though a quadratic order distillation scheme exists that could be concatenated to achieve distillation to order $2^k$ for some integer $k$.
In particular, Refs\@.~\cite{Campbell_2016, Duclos_Cianci_2015} provide protocols for distilling $\ket{\sqrt{T}}$ to quadratic order on $N=5$ qubits.
The quantum Reed-Muller code family yields distillation protocols for $\ket{\sqrt{T}}$-states with distances $d=2^r-1$ for any integer $r>0$~\cite{tiurev2026parityunfoldeddistillationarchitecturenoisebiased}.
Additionally, there exists a catalysis protocol~\cite{Gidney_2019} that consumes $5\ket{T}$ to produce $2\ket{\sqrt{T}}$ without consuming a preexisting $\ket{\sqrt{T}}$-state.
This reduces the $T$-cost of $\ket{\sqrt{T}}$-state to $2.5$, provided one can produce a first good quality $\ket{\sqrt{T}}$-state.

Contrary to the canonical family of $\ket{T}$ distillation using a SAT solver, we have been able to find a distance $d=3$ protocol for distilling $\ket{\sqrt{T}}$ on $N=6$ qubits, which is less than the first distance 3 canonical family protocol on $N=9$ qubits. 
This protocol of distance 3 over $N=6$ qubits coincides exactly with the quantum Reed-Muller code $QRM(1,5)$ that has $5$ $X$-stabilizers and $n=31$ qubits.
The SAT solvers return UNSAT on $N=4$ qubits for distance $2$, suggesting that the protocol from~\cite{Duclos_Cianci_2015,Campbell_2016} as well as $\mathcal{F}_5^1$ are optimal regarding the qubit footprint.
Yet, we have not found any distance $4$ protocol distilling $\ket{\sqrt{T}}$ for $N<13$.
In the general case of $\ket{\sqrt[L]{T}}=R_Z\left(\frac{\pi}{2^{L+2}}\right)\ket{+}$ distillation, the same procedure works using the ring $\mathcal{R}=\mathbb{Z}\backslash 2^{L+1}\mathbb{Z}$ and the conditions for $f$ to define a logical $\sqrt[L]{T}$ gate on the repetition code are 
\[
g(B) = \begin{cases}
    1 \mod 2^{L+1}& \text{if } 1 \leq |B|\leq 3 \\
    1 \mod 2^{L}& \text{if } |B|= 4  \\
    \vdots\\
    1 \mod 2& \text{if } |B|=L+3 \\
\end{cases}
\]

In the exact same way as for $\mathcal{F}^0_N$,$\mathcal{F}^1_N$, one can build a canonical family $\mathcal{F}^L_N$ that will have a distance $\geq \lceil \frac{N}{3+L}\rceil$.

%% file: Matrices.tex
\section{Matrices of protocols}\label{App:G_matrices}
We list here the matrices associated to the different new distillation protocols found with the SAT solver and discussed in the main text. The matrices are given in the formalism detailed in \Cref{sec:logicalT-to-distill}.

\begin{widetext}

\subsection{T distillation protocols}

The matrix associated to the $64T \rightarrow 1 T$ protocol on 10 qubits with error suppression in $495p^4$ is given in \Cref{eq:64T1T}
\begin{equation}\label{eq:64T1T}
\setlength{\arraycolsep}{2pt}
\resizebox{0.85\textwidth}{!}{$
\left(\begin{array}{cccccccccccccccccccccccccccccccccccccccccccccccccccccccccccccccc}
    1 & 0 & 0 & 0 & 0 & 0 & 0 & 0 & 0 & 0 & 1 & 1 & 1 & 1 & 1 & 1 & 1 & 1 & 1 & 0 & 0 & 0 & 0 & 0 & 0 & 0 & 0 & 0 & 0 & 0 & 0 & 0 & 0 & 0 & 0 & 0 & 0 & 0 & 0 & 0 & 0 & 0 & 0 & 0 & 0 & 0 & 1 & 1 & 1 & 1 & 1 & 1 & 1 & 1 & 1 &
 0 & 0 & 0 & 0 & 0 & 0 & 0 & 0 & 0 \\
    1 & 0 & 0 & 1 & 0 & 0 & 0 & 0 & 0 & 0 & 1 & 1 & 0 & 1 & 1 & 1 & 1 & 1 & 1 & 0 & 0 & 0 & 1 & 1 & 1 & 0 & 0 & 0 & 0 & 1 & 0 & 0 & 1 & 1 & 0 & 0 & 0 & 0 & 1 & 0 & 0 & 0 & 1 & 1 & 0 & 1 & 1 & 0 & 1 & 1 & 1 & 0 & 0 & 1 & 0 &
 0 & 1 & 0 & 1 & 1 & 1 & 0 & 0 & 1 \\
    0 & 0 & 0 & 1 & 0 & 1 & 0 & 0 & 0 & 0 & 0 & 0 & 1 & 0 & 1 & 0 & 0 & 0 & 0 & 0 & 0 & 0 & 1 & 0 & 0 & 0 & 0 & 1 & 0 & 1 & 1 & 0 & 1 & 1 & 0 & 1 & 1 & 1 & 0 & 0 & 0 & 1 & 1 & 1 & 1 & 1 & 1 & 0 & 0 & 0 & 1 & 1 & 1 & 1 & 1 &
 1 & 0 & 0 & 0 & 1 & 1 & 0 & 1 & 0 \\
    0 & 0 & 0 & 0 & 0 & 1 & 0 & 0 & 1 & 0 & 0 & 0 & 0 & 0 & 1 & 0 & 0 & 1 & 0 & 0 & 0 & 0 & 0 & 1 & 0 & 0 & 1 & 0 & 0 & 0 & 1 & 1 & 1 & 1 & 0 & 1 & 1 & 1 & 1 & 1 & 1 & 0 & 0 & 0 & 1 & 1 & 1 & 1 & 1 & 1 & 0 & 0 & 0 & 1 & 1 &
 1 & 1 & 1 & 0 & 0 & 0 & 1 & 0 & 0 \\
    0 & 0 & 0 & 0 & 0 & 0 & 0 & 0 & 1 & 1 & 0 & 0 & 0 & 0 & 0 & 0 & 0 & 1 & 1 & 0 & 0 & 0 & 0 & 0 & 1 & 1 & 0 & 0 & 0 & 0 & 0 & 1 & 1 & 1 & 1 & 1 & 1 & 0 & 0 & 1 & 1 & 1 & 1 & 1 & 1 & 0 & 0 & 0 & 1 & 1 & 1 & 1 & 1 & 1 & 0 &
 0 & 0 & 1 & 1 & 1 & 1 & 0 & 0 & 0 \\
    0 & 1 & 0 & 0 & 0 & 0 & 0 & 0 & 0 & 1 & 1 & 0 & 0 & 0 & 0 & 0 & 0 & 0 & 1 & 1 & 0 & 0 & 0 & 0 & 0 & 0 & 0 & 1 & 1 & 0 & 0 & 1 & 1 & 0 & 1 & 1 & 1 & 1 & 1 & 1 & 0 & 0 & 1 & 1 & 0 & 1 & 1 & 1 & 1 & 0 & 0 & 1 & 1 & 0 & 1 &
 1 & 1 & 0 & 0 & 0 & 1 & 0 & 0 & 1 \\
    0 & 1 & 1 & 0 & 0 & 0 & 0 & 0 & 0 & 0 & 1 & 1 & 0 & 0 & 0 & 0 & 0 & 0 & 0 & 0 & 1 & 0 & 0 & 0 & 0 & 0 & 1 & 0 & 1 & 1 & 0 & 1 & 1 & 0 & 1 & 1 & 0 & 1 & 1 & 0 & 1 & 1 & 1 & 0 & 1 & 0 & 1 & 1 & 0 & 1 & 1 & 1 & 0 & 1 & 0 &
 0 & 1 & 1 & 1 & 0 & 0 & 0 & 1 & 0 \\
    0 & 0 & 1 & 0 & 1 & 0 & 0 & 0 & 0 & 0 & 0 & 1 & 0 & 1 & 0 & 0 & 0 & 0 & 0 & 0 & 0 & 1 & 0 & 0 & 0 & 1 & 0 & 0 & 1 & 1 & 1 & 1 & 0 & 0 & 1 & 1 & 0 & 1 & 0 & 1 & 1 & 1 & 0 & 1 & 0 & 1 & 1 & 0 & 1 & 1 & 1 & 0 & 1 & 0 & 1 &
 0 & 0 & 0 & 1 & 0 & 1 & 1 & 0 & 1 \\
    0 & 0 & 0 & 0 & 1 & 0 & 1 & 0 & 0 & 0 & 0 & 0 & 0 & 1 & 0 & 1 & 0 & 0 & 0 & 1 & 0 & 0 & 1 & 0 & 0 & 0 & 0 & 0 & 1 & 1 & 1 & 1 & 0 & 1 & 1 & 0 & 0 & 0 & 1 & 0 & 1 & 0 & 1 & 1 & 1 & 1 & 0 & 1 & 0 & 1 & 0 & 1 & 1 & 1 & 1 &
 0 & 1 & 0 & 0 & 1 & 0 & 0 & 1 & 1 \\
    0 & 0 & 0 & 0 & 0 & 0 & 1 & 1 & 0 & 0 & 0 & 0 & 0 & 0 & 0 & 1 & 1 & 0 & 0 & 0 & 1 & 0 & 0 & 1 & 0 & 0 & 0 & 0 & 1 & 1 & 1 & 0 & 0 & 1 & 1 & 0 & 1 & 1 & 1 & 1 & 0 & 1 & 0 & 1 & 1 & 0 & 1 & 1 & 1 & 0 & 1 & 0 & 1 & 1 & 0 &
 1 & 0 & 0 & 1 & 0 & 0 & 1 & 1 & 0
\end{array}\right)
$}
\end{equation}

The matrix associated to the $65T \rightarrow 1 T$ protocol on 11 qubits with error suppression in $7947p^5$ is given in \Cref{eq:65T1T}
\begin{equation}\label{eq:65T1T}
\setlength{\arraycolsep}{2pt}
\resizebox{0.85\textwidth}{!}{$
\left(\begin{array}{ccccccccccccccccccccccccccccccccccccccccccccccccccccccccccccccccc}
    1 & 0 & 0 & 0 & 0 & 0 & 0 & 0 & 0 & 0 & 0 & 1 & 1 & 1 & 1 & 1 & 1 & 1 & 1 & 1 & 0 & 0 & 0 & 0 & 0 & 0 & 0 & 0 & 0 & 0 & 0 & 0 & 0 & 0 & 0 & 0 & 0 & 0 & 1 & 1 & 0 & 0 & 1 & 1 & 0 & 0 & 1 & 0 & 1 & 1 & 0 & 0 & 1 & 0 & 1 & 0 & 0 & 0 & 0 & 0 & 0 & 0 & 0 & 0 & 0 \\
    1 & 1 & 0 & 0 & 0 & 0 & 0 & 0 & 0 & 0 & 0 & 0 & 0 & 0 & 0 & 0 & 0 & 0 & 0 & 0 & 0 & 0 & 0 & 0 & 0 & 0 & 0 & 0 & 0 & 0 & 0 & 0 & 0 & 0 & 0 & 0 & 0 & 0 & 1 & 1 & 1 & 1 & 1 & 1 & 
1 & 1 & 1 & 1 & 1 & 1 & 1 & 1 & 1 & 1 & 1 & 1 & 0 & 0 & 0 & 0 & 0 & 0 & 0 & 0 & 0 \\
    0 & 1 & 0 & 0 & 0 & 0 & 1 & 0 & 0 & 0 & 0 & 1 & 1 & 1 & 1 & 0 & 1 & 1 & 1 & 1 & 0 & 0 & 0 & 0 & 1 & 1 & 0 & 0 & 1 & 0 & 0 & 1 & 0 & 0 & 0 & 0 & 1 & 1 & 1 & 1 & 0 & 0 & 0 & 0 & 
1 & 1 & 1 & 0 & 0 & 0 & 1 & 1 & 1 & 0 & 0 & 1 & 1 & 1 & 0 & 1 & 0 & 0 & 1 & 0 & 1 \\
    0 & 0 & 0 & 0 & 0 & 0 & 1 & 0 & 0 & 1 & 0 & 0 & 0 & 0 & 0 & 1 & 0 & 0 & 1 & 0 & 0 & 0 & 0 & 0 & 1 & 0 & 0 & 1 & 0 & 0 & 0 & 1 & 1 & 1 & 1 & 0 & 1 & 1 & 1 & 1 & 1 & 1 & 1 & 1 & 
1 & 1 & 0 & 0 & 0 & 0 & 0 & 0 & 1 & 1 & 1 & 1 & 1 & 1 & 1 & 0 & 0 & 0 & 0 & 1 & 0 \\
    0 & 0 & 0 & 0 & 0 & 0 & 0 & 0 & 0 & 1 & 1 & 0 & 0 & 0 & 0 & 0 & 0 & 0 & 1 & 1 & 0 & 0 & 0 & 0 & 0 & 1 & 1 & 0 & 0 & 0 & 0 & 0 & 1 & 1 & 1 & 1 & 1 & 1 & 0 & 0 & 0 & 0 & 1 & 1 & 
1 & 1 & 1 & 1 & 1 & 1 & 1 & 1 & 1 & 1 & 0 & 0 & 0 & 0 & 1 & 1 & 1 & 1 & 0 & 0 & 0 \\
    0 & 0 & 1 & 0 & 0 & 0 & 0 & 0 & 0 & 0 & 1 & 1 & 0 & 0 & 0 & 0 & 0 & 0 & 0 & 1 & 1 & 0 & 0 & 0 & 0 & 0 & 0 & 0 & 1 & 0 & 1 & 0 & 1 & 0 & 1 & 1 & 1 & 1 & 1 & 1 & 1 & 1 & 0 & 1 & 
1 & 0 & 0 & 0 & 1 & 1 & 1 & 1 & 0 & 0 & 1 & 1 & 1 & 1 & 0 & 0 & 0 & 1 & 1 & 0 & 0 \\
    0 & 0 & 1 & 1 & 0 & 0 & 0 & 0 & 0 & 0 & 0 & 1 & 1 & 0 & 0 & 0 & 0 & 0 & 0 & 0 & 0 & 1 & 0 & 0 & 0 & 0 & 0 & 1 & 0 & 1 & 1 & 0 & 1 & 0 & 1 & 1 & 1 & 0 & 1 & 1 & 1 & 1 & 1 & 0 & 
0 & 1 & 1 & 1 & 1 & 0 & 1 & 0 & 1 & 1 & 0 & 0 & 0 & 1 & 1 & 1 & 0 & 0 & 0 & 0 & 1 \\
    0 & 0 & 0 & 1 & 1 & 0 & 0 & 0 & 0 & 0 & 0 & 0 & 1 & 1 & 0 & 0 & 0 & 0 & 0 & 0 & 0 & 0 & 1 & 0 & 0 & 0 & 1 & 0 & 0 & 1 & 1 & 1 & 1 & 0 & 0 & 1 & 1 & 0 & 1 & 0 & 1 & 0 & 1 & 1 & 
1 & 1 & 1 & 1 & 0 & 1 & 0 & 1 & 0 & 0 & 1 & 1 & 0 & 0 & 0 & 1 & 0 & 1 & 1 & 1 & 0 \\
    0 & 0 & 0 & 0 & 1 & 1 & 0 & 0 & 0 & 0 & 0 & 0 & 0 & 1 & 1 & 0 & 0 & 0 & 0 & 0 & 1 & 0 & 0 & 1 & 0 & 0 & 0 & 0 & 0 & 1 & 1 & 1 & 1 & 1 & 0 & 1 & 0 & 0 & 0 & 1 & 0 & 1 & 1 & 0 & 
0 & 1 & 0 & 0 & 1 & 1 & 1 & 1 & 1 & 1 & 1 & 1 & 0 & 1 & 0 & 0 & 1 & 0 & 1 & 0 & 1 \\
    0 & 0 & 0 & 0 & 0 & 1 & 0 & 1 & 0 & 0 & 0 & 0 & 0 & 0 & 1 & 0 & 1 & 0 & 0 & 0 & 0 & 1 & 0 & 0 & 1 & 0 & 0 & 0 & 0 & 1 & 1 & 1 & 0 & 1 & 0 & 1 & 0 & 1 & 1 & 1 & 1 & 1 & 0 & 1 & 
1 & 0 & 1 & 1 & 0 & 1 & 0 & 1 & 1 & 1 & 0 & 0 & 1 & 0 & 0 & 1 & 0 & 0 & 0 & 1 & 1 \\
    0 & 0 & 0 & 0 & 0 & 0 & 0 & 1 & 1 & 0 & 0 & 0 & 0 & 0 & 0 & 0 & 1 & 1 & 0 & 0 & 0 & 0 & 1 & 0 & 0 & 1 & 0 & 0 & 0 & 1 & 1 & 1 & 0 & 1 & 1 & 0 & 0 & 1 & 0 & 1 & 0 & 1 & 1 & 1 & 
1 & 1 & 1 & 1 & 1 & 0 & 1 & 0 & 0 & 0 & 1 & 1 & 0 & 0 & 1 & 0 & 1 & 0 & 1 & 1 & 0
\end{array}\right)
$}
\end{equation}

\subsection{\texorpdfstring{$\CCZ$}{CCZ} distillation protocols}
The matrix associated to the $47 T \rightarrow 1 \CCZ$ protocol on 9 qubits with error suppression in $236p^3$ is given in \Cref{eq:47T1CCZ}
\begin{equation}\label{eq:47T1CCZ}
\setlength{\arraycolsep}{2pt}
\resizebox{0.85\textwidth}{!}{$
\left(\begin{array}{ccccccccccccccccccccccccccccccccccccccccccccccc}
0&1&1&1&1&1&0&1&0&1&1&0&1&1&0&1&1&1&1&1&0&1&0&0&1&1&1&0&1&0&0&1&1&0&0&0&0&1&0&0&1&0&1&1&1&0&1\\
1&0&0&1&0&0&1&1&1&1&0&0&0&0&1&0&1&1&0&0&1&1&1&1&0&1&1&1&0&1&0&0&1&1&1&0&1&0&0&1&0&1&1&0&0&0&0\\
0&0&0&1&0&0&1&1&1&0&0&1&0&1&1&0&0&0&0&0&0&1&0&0&0&1&1&1&0&1&1&1&1&0&1&1&1&1&1&0&1&0&1&0&1&1&0\\
1&0&0&1&1&1&0&1&0&0&1&0&1&0&0&1&1&1&0&1&0&0&1&0&1&0&0&1&1&1&0&1&0&0&1&0&1&0&0&1&1&1&0&1&0&0&1\\
0&1&0&0&1&0&1&1&0&1&1&0&0&1&0&0&1&0&1&1&0&1&1&0&0&1&0&0&1&0&1&1&0&1&1&0&0&1&0&0&1&0&1&1&0&1&1\\
0&0&1&1&1&0&0&0&1&1&1&0&0&0&1&1&1&0&0&0&1&1&1&0&0&0&1&1&1&0&0&0&1&1&1&0&0&0&1&1&1&0&0&0&1&1&1\\
0&0&0&0&0&1&1&1&1&1&1&0&0&0&0&0&0&1&1&1&1&1&1&0&0&0&0&0&0&1&1&1&1&1&1&0&0&0&0&0&0&1&1&1&1&1&1\\
0&0&0&0&0&0&0&0&0&0&0&1&1&1&1&1&1&1&1&1&1&1&1&0&0&0&0&0&0&0&0&0&0&0&0&1&1&1&1&1&1&1&1&1&1&1&1\\
0&0&0&0&0&0&0&0&0&0&0&0&0&0&0&0&0&0&0&0&0&0&0&1&1&1&1&1&1&1&1&1&1&1&1&1&1&1&1&1&1&1&1&1&1&1&1
\end{array}\right)
$}
\end{equation}

The matrix associated to the $48 T \rightarrow 1 \CCZ$ protocol on 10 qubits with error suppression in $2816p^4$ is given in \Cref{eq:48T1CCZ}

\begin{equation}\label{eq:48T1CCZ}
    \setlength{\arraycolsep}{2pt}
    \resizebox{0.85\textwidth}{!}{$
    \left(\begin{array}{cccccccccccccccccccccccccccccccccccccccccccccccc}
    0&0&1&1&1&1&1&0&1&1&1&0&1&1&1&0&1&1&1&1&1&1&0&1&1&0&1&1&0&1&1&1&1&0&1&0&0&1&0&0&0&1&1&0&1&0&1&1\\
    0&1&1&0&0&1&1&1&1&0&1&1&0&1&0&0&0&1&0&1&1&0&1&1&0&0&0&1&0&0&1&0&1&1&0&0&1&1&1&0&1&1&1&1&1&1&1&0\\
    0&0&1&1&1&1&0&1&0&0&1&0&1&1&0&1&0&0&1&1&1&1&0&1&0&0&0&1&0&1&1&0&1&1&1&0&0&1&1&0&1&1&1&0&1&1&1&0\\
    1&0&0&1&1&0&0&1&1&0&0&1&1&0&0&1&1&0&0&1&1&0&0&1&1&0&0&1&0&1&1&0&0&1&1&0&0&1&1&0&0&1&1&0&1&0&0&1\\
    0&1&0&1&0&1&0&1&0&0&1&1&0&0&1&1&0&1&0&1&0&1&0&1&1&1&0&0&0&1&0&1&0&1&0&1&0&1&0&1&0&1&0&1&1&1&0&0\\
    0&0&1&1&0&0&1&1&0&1&0&1&0&1&0&1&0&0&1&1&0&0&1&1&0&1&0&1&0&0&1&1&0&0&1&1&0&0&1&1&0&0&1&1&0&1&0&1\\
    0&0&0&0&1&1&1&1&0&0&1&1&0&0&1&1&0&0&0&0&1&1&1&1&0&0&1&1&0&0&0&0&1&1&1&1&0&0&0&0&1&1&1&1&0&0&1&1\\
    0&0&0&0&0&0&0&0&1&1&1&1&0&0&0&0&1&1&1&1&1&1&1&1&0&0&0&0&1&1&1&1&1&1&1&1&0&0&0&0&0&0&0&0&1&1&1&1\\
    0&0&0&0&0&0&0&0&0&0&0&0&1&1&1&1&1&1&1&1&1&1&1&1&0&0&0&0&0&0&0&0&0&0&0&0&1&1&1&1&1&1&1&1&1&1&1&1\\
    0&0&0&0&0&0&0&0&0&0&0&0&0&0&0&0&0&0&0&0&0&0&0&0&1&1&1&1&1&1&1&1&1&1&1&1&1&1&1&1&1&1&1&1&1&1&1&1\\
    \end{array}\right)
    $}
    \end{equation}
\end{widetext}

%% file: Recycling_app.tex
\section{Qubit recycling}
\label{app: recycling}

In this appendix, we detail the methods used to compress a distillation protocol by recycling qubits, i.e.~by minimizing the number of qubits that need to be simultaneously active.

\subsection{Definitions and invariant operations}

Recall that a $nT \to 1T$ distillation protocol is described by a triorthogonal matrix $\mathcal{G} \in \mathbb{F}_2^{N \times n}$, whose $N$ rows index the qubits and whose $n$ columns index the Pauli product rotations.
We fix the first row to the output row and the remaining $N-1$ rows the check rows.
For a given matrix $\mathcal{G}$, let $f_i$ and $\ell_i$ denote the column indices of the first and last $1$ of row $i$,
\begin{equation}
  f_i = \min\{\, j : \mathcal{G}_{ij} = 1 \,\}, \quad
  \ell_i = \max\{\, j : \mathcal{G}_{ij} = 1 \,\}.
\end{equation}
A qubit is \emph{active} at a given column if it has already been involved in a rotation and is still to be involved in a later one.
A check qubit $i$ is therefore active on the interval $[f_i, \ell_i]$, after which it can be measured out and reinitialized.
The output qubit remains active until the last column $n$ of the circuit, as it is never measured out, thus its active interval is $[f_i,n]$.
The number of qubits simultaneously active at column $j$ is thus
\begin{equation}
  a(j) = \bigl|\{\, i : f_i \le j \le \tilde{\ell}_i \,\}\bigr|,
  \quad
  \tilde{\ell}_i =
  \begin{cases}
    n & \text{if } i=1, \\
    \ell_i & \text{otherwise}.
  \end{cases}
\end{equation}
For compressing distillation circuit we thus aim to minimize the maximum active-qubit count over the whole circuit,
\begin{equation}
  A(\mathcal{G}) = \max_{1 \le j \le n} a(j).
  \label{eq:app-aq}
\end{equation}

As we stated in the main text, $A(\mathcal{G})$ depends on the specific matrix $\mathcal{G}$ implementing the protocol, not on the protocol itself.
Three operations transform $\mathcal{G}$ into another matrix implementing the same $nT \to 1T$ protocol while leaving the distilled state unchanged. These operations preserving triorthogonality and the protocol logic are:

\begin{itemize}
    \item \emph{Column permutations:}
    The columns of $\mathcal{G}$ correspond to Pauli product rotations that are all diagonal in the $Z$ basis and hence mutually commute.
    Reordering them therefore realizes the same protocol, but changes the intervals $[f_i, \ell_i]$ and thus $A(\mathcal{G})$.
    \item \emph{Row additions:}
    Adding one row of $\mathcal{G}$ to another over $\mathbb{F}_2$ is simply a change of generator basis of the underlying code, provided the added row is of even weight, to preserve parity.
    It leaves the distilled state invariant while altering the support of the target row, and hence its interval $[f_i, \ell_i]$.
    \item \emph{Row permutations:}
    Permuting the rows of $\mathcal{G}$ trivially relabels the qubits and leaves $A(\mathcal{G})$ unchanged. Therefore, we ignore this operation.
\end{itemize}

Minimizing $A(\mathcal{G})$ over the equivalence class of a protocol is thus a joint optimization over the row basis (row additions) and the column order (column permutations). 
We here tackled these two degree of freedom separately.

\subsection{Minimizing the generator weights}
\label{app: min generator weight}

The output qubit is the most costly of the protocol, since it is never measured out and thus stays active from its first rotation until the end of the circuit.
Its contribution to $A(\mathcal{G})$ is governed by $f_1$, the column of its first rotation: the earlier the output qubit is initialized, the longer it needs to coexist with the check qubits.
Reducing the Hamming weight of the output row lowers the number of columns in which it participates, which in turn allows its first rotation to be pushed later by using a suited column order, delaying the initialization of the output qubit.

We reduce this weight using the row-addition operation described above.
The output row is defined only up to the addition of check rows, so any subset $S$ of the $N-1$ check rows yields an equally valid output row
\begin{equation}
  \mathbf{r}_1 \;\oplus \bigoplus_{i \in S} \mathbf{r}_i,
\end{equation}
where $\mathbf{r}_i=\sum_k \beta_i^k$ denotes the $i$-th row of $\mathcal{G}$.
These subsets span the entire coset of valid output rows.
Considering small $N\leq15$, we select the subset $S$ minimizing the Hamming weight of the output row by explicit search over all $2^{N-1}$ possibilities.

Similarly, we reduce the weight of check rows.
Each check row may absorb any combination of the other check rows without affecting the parity structure that distinguishes it from the output row, and hence without altering the code.
Repeatedly, we look for a pair of check rows $(i, i')$ such that replacing $\mathbf{r}_i$ by $\mathbf{r}_i \oplus \mathbf{r}_{i'}$ strictly lowers the Hamming weight of $\mathbf{r}_i$, and apply the addition whenever such a pair exists, until no single addition reduces any check-row weight further.
This greedy approach brings $\mathcal{G}$ towards a lighter generator basis for the check subspace.

Unlike the output row, the weight of a check row acts on $A(\mathcal{G})$ only indirectly, through the length of the active window $[f_i, \ell_i]$, which also depends on the column order.
Therefore, while the greedy weight reduction of check rows may help minimize $A(\mathcal{G})$, we did not see any improvement in practice.

\subsection{Branch-and-bound search for column ordering}
\label{app: bnb}

With the rows of $\mathcal{G}$ fixed, we now tackle the ordering of the $n$ columns.
The active-qubit count $A(\mathcal{G})$ is a highly non-smooth function of this ordering.
Therefore, local-search heuristics such as simulated annealing tend to stall on the resulting plateaus.
We instead perform an exact search with a branch-and-bound algorithm, which returns both an ordering and the corresponding $A(\mathcal{G})$.

The algorithm builds the column order one column at a time, maintaining the following state for the partial ordering placed so far.
For each qubit we track whether it is currently active and, if so, how many of its $1$s remain to be placed, a check qubit being released once its count reaches zero while the output qubit stays active until the end.
We denote by $a$ the number of currently active qubits and by $a^\star$ the running peak of $a$ over the partial ordering, so that a completed ordering has $A(\mathcal{G}) = a^\star$.
At each step we extend the partial ordering by an unplaced column, updating the active set and the remaining counts accordingly.

To make the branch-and-bound search tractable in practice, we extensively prune the search space.
As columns are appended, $a^\star$ can only stay constant or increase, since placing further columns never releases a qubit earlier than in any completion.
Hence, if the current $a^\star$ already equals the best complete value found so far, no completion of the current partial ordering can improve on it, and the entire branch can be pruned.
To make pruning happen early in branch exploration, we explore columns in an order that tends to close open rows quickly, so that good complete orderings, and thus low incumbent values, are found near the start of the search.
We summarize our implementation in Algorithm~\ref{alg:bnb}.

\begin{algorithm}[H]
\caption{Branch-and-bound for $\min A(\mathcal{G})$}
\label{alg:bnb}
\begin{algorithmic}[1]
\State $a^\star_{\mathrm{best}} \gets$ value of a greedy initial ordering
\Procedure{Explore}{partial order $\pi$, peak $a^\star$}
  \If{$a^\star \ge a^\star_{\mathrm{best}}$}
    \State \textbf{return} \Comment{prune}
  \EndIf
  \If{$\pi$ places all $n$ columns}
    \State $a^\star_{\mathrm{best}} \gets a^\star$; record $\pi$
    \State \textbf{return}
  \EndIf
  \For{each unplaced column $c$, ordered to close rows first}
    \State extend $\pi$ by $c$; update active set; $a \gets |\text{active}|$
    \State \Call{Explore}{$\pi \cup \{c\}$, $\max(a^\star, a)$}
  \EndFor
\EndProcedure
\end{algorithmic}
\end{algorithm}

While the pruning is effective in practice, certifying optimality requires exhausting the search tree, whose size is worst-case $O(n!)$ in the number of columns.
The branch-and-bound therefore reliably returns good orderings quickly, but proving that a given value is optimal becomes intractable as $n$ grows.

\subsection{Certification of column ordering with an integer linear program}
\label{app: ilp}

To certify that an ordering is optimal, we frame the column-ordering problem as an integer linear program (ILP), which can be solved with a certificate.
Rather than optimizing over permutations directly, we assign each column to a position and encode the active-qubit count through per-qubit indicator variables.

Let $x_{c,t} \in \{0,1\}$ indicate that column $c$ is placed at position $t$, with $1 \le c, t \le n$.
A valid ordering is enforced by requiring that each column occupies exactly one position and each position holds exactly one column,
\begin{equation}
  \sum_{t} x_{c,t} = 1 \quad \forall c,
  \qquad
  \sum_{c} x_{c,t} = 1 \quad \forall t.
\end{equation}
For each qubit $i$ and position $t$ we introduce indicator variables $o_{i,t}$ and $u_{i,t}$, marking respectively whether qubit $i$ has already been opened by position $t$ (it has a $1$ at some position $\le t$) and whether it is still to be used at position $t$ (it has a $1$ at some position $\ge t$).
Both are prefix and suffix disjunctions of the placement indicators and are encoded with the standard linear inequalities.
Qubit $i$ is then active at position $t$ when it is both opened and still to be used, $a_{i,t} = o_{i,t} \wedge u_{i,t}$, again linearized in the standard way, with the output qubit kept active until the last position.
Introducing an integer variable $\mathbf{A}$ bounding the active count at every position,
\begin{equation}
  \sum_{i} a_{i,t} \le \mathbf{A} \quad \forall t,
\end{equation}
the program minimizes $\mathbf{A}$, whose optimal value is exactly $A(\mathcal{G})$.

We warm-start the solver with the ordering returned by the branch-and-bound, so that it begins from a known feasible value and spends its effort closing the gap to the optimum rather than searching for a good ordering.
The model comprises $O(n^2)$ binary variables, dominated by the $n \times n$ assignment block $x_{c,t}$.
Solving this ILP is worst-case exponential in $n$, but the polynomial model size lets solver find a certificate quickly in practice.

\subsection{Results}
\label{app: recycling results}

We run the full minimization pipeline, starting with the row weights minimization and followed by the branch-and-bound search and its ILP certification, on the $15$-to-$1$ and $49$-to-$1$ protocols.
For the $15$-to-$1$ protocol, the pipeline compresses the implementation from $N=5$ to $4$ active qubits, while for the $49$-to-$1$ protocol it reduces the footprint from $N=14$ to $5$ active qubits, both column ordering being certified optimal by the ILP.
To our knowledge this later results is new, going beyond $N=7$ proposed in \cite{ibm_qbit_recycling} and $N=6$ found in \cite{gidney_2024_13777072}. In ~\Cref{fig:Recycling}, we propose a visualization of the recycling pipeline for the compression of that $49$-to-$1$ protocol. See also \href{https://algassert.com/quirk#circuit={%22cols%22:[[%22H%22,%22H%22,%22H%22,%22H%22,%22H%22],[%22zpar%22,%22zpar%22,1,1,1,%22%E2%88%9Ai%22],[%22zpar%22,1,%22zpar%22,1,1,%22%E2%88%9Ai%22],[1,1,%22zpar%22,%22zpar%22,1,%22%E2%88%9Ai%22],[1,%22zpar%22,1,%22zpar%22,1,%22%E2%88%9Ai%22],[%22zpar%22,1,%22zpar%22,%22zpar%22,1,%22%E2%88%9Ai%22],[%22zpar%22,%22zpar%22,1,%22zpar%22,1,%22%E2%88%9Ai%22],[1,1,%22zpar%22,1,%22zpar%22,%22%E2%88%9Ai%22],[1,%22zpar%22,1,1,%22zpar%22,%22%E2%88%9Ai%22],[%22zpar%22,1,%22zpar%22,1,%22zpar%22,%22%E2%88%9Ai%22],[%22zpar%22,%22zpar%22,1,1,%22zpar%22,%22%E2%88%9Ai%22],[1,1,%22zpar%22,%22zpar%22,%22zpar%22,%22%E2%88%9Ai%22],[1,%22zpar%22,1,%22zpar%22,%22zpar%22,%22%E2%88%9Ai%22],[%22zpar%22,%22zpar%22,1,%22zpar%22,%22zpar%22,%22%E2%88%9Ai%22],[%22zpar%22,1,%22zpar%22,%22zpar%22,%22zpar%22,%22%E2%88%9Ai%22],[1,%22Z^%C2%BC%22],[%22H%22,%22H%22,1,%22H%22,%22H%22],[%22|0%E2%9F%A9%E2%9F%A80|%22,%22|0%E2%9F%A9%E2%9F%A80|%22,1,%22|0%E2%9F%A9%E2%9F%A80|%22,%22|0%E2%9F%A9%E2%9F%A80|%22],[%22H%22,%22H%22,1,%22H%22,%22H%22],[1,%22zpar%22,%22zpar%22,1,1,%22%E2%88%9Ai%22],[%22zpar%22,%22zpar%22,1,1,1,%22%E2%88%9Ai%22],[%22zpar%22,%22zpar%22,1,%22zpar%22,1,%22%E2%88%9Ai%22],[1,%22zpar%22,%22zpar%22,%22zpar%22,1,%22%E2%88%9Ai%22],[%22zpar%22,%22zpar%22,1,1,%22zpar%22,%22%E2%88%9Ai%22],[1,%22zpar%22,%22zpar%22,1,%22zpar%22,%22%E2%88%9Ai%22],[%22zpar%22,%22zpar%22,1,%22zpar%22,%22zpar%22,%22%E2%88%9Ai%22],[1,%22zpar%22,%22zpar%22,%22zpar%22,%22zpar%22,%22%E2%88%9Ai%22],[1,%22H%22],[1,%22|0%E2%9F%A9%E2%9F%A80|%22],[1,%22H%22],[%22zpar%22,%22zpar%22,1,%22zpar%22,1,%22%E2%88%9Ai%22],[1,%22zpar%22,%22zpar%22,%22zpar%22,1,%22%E2%88%9Ai%22],[%22zpar%22,%22zpar%22,1,1,%22zpar%22,%22%E2%88%9Ai%22],[1,%22zpar%22,%22zpar%22,1,%22zpar%22,%22%E2%88%9Ai%22],[%22zpar%22,%22zpar%22,1,%22zpar%22,%22zpar%22,%22%E2%88%9Ai%22],[1,%22zpar%22,%22zpar%22,%22zpar%22,%22zpar%22,%22%E2%88%9Ai%22],[1,1,1,%22H%22,%22H%22],[1,1,1,%22|0%E2%9F%A9%E2%9F%A80|%22,%22|0%E2%9F%A9%E2%9F%A80|%22],[1,1,1,%22H%22,%22H%22],[%22zpar%22,1,1,%22zpar%22,1,%22%E2%88%9Ai%22],[1,1,%22zpar%22,%22zpar%22,1,%22%E2%88%9Ai%22],[%22zpar%22,%22zpar%22,1,%22zpar%22,1,%22%E2%88%9Ai%22],[1,%22zpar%22,%22zpar%22,%22zpar%22,1,%22%E2%88%9Ai%22],[%22zpar%22,1,1,%22zpar%22,%22zpar%22,%22%E2%88%9Ai%22],[1,1,%22zpar%22,%22zpar%22,%22zpar%22,%22%E2%88%9Ai%22],[%22zpar%22,%22zpar%22,1,%22zpar%22,%22zpar%22,%22%E2%88%9Ai%22],[1,%22zpar%22,%22zpar%22,%22zpar%22,%22zpar%22,%22%E2%88%9Ai%22],[1,1,1,%22H%22],[1,1,1,%22|0%E2%9F%A9%E2%9F%A80|%22],[1,1,1,%22H%22],[%22zpar%22,1,1,%22zpar%22,%22zpar%22,%22%E2%88%9Ai%22],[1,1,%22zpar%22,%22zpar%22,%22zpar%22,%22%E2%88%9Ai%22],[%22zpar%22,%22zpar%22,1,%22zpar%22,%22zpar%22,%22%E2%88%9Ai%22],[1,%22zpar%22,%22zpar%22,%22zpar%22,%22zpar%22,%22%E2%88%9Ai%22],[1,1,1,1,%22H%22],[1,1,1,1,%22|0%E2%9F%A9%E2%9F%A80|%22],[1,1,1,1,%22H%22],[%22zpar%22,1,1,1,%22zpar%22,%22%E2%88%9Ai%22],[1,1,%22zpar%22,1,%22zpar%22,%22%E2%88%9Ai%22],[%22zpar%22,%22zpar%22,1,1,%22zpar%22,%22%E2%88%9Ai%22],[1,%22zpar%22,%22zpar%22,1,%22zpar%22,%22%E2%88%9Ai%22],[%22zpar%22,1,1,%22zpar%22,%22zpar%22,%22%E2%88%9Ai%22],[1,1,%22zpar%22,%22zpar%22,%22zpar%22,%22%E2%88%9Ai%22],[1,%22zpar%22,%22zpar%22,%22zpar%22,%22zpar%22,%22%E2%88%9Ai%22],[%22zpar%22,%22zpar%22,1,%22zpar%22,%22zpar%22,%22%E2%88%9Ai%22],[1,%22H%22,%22H%22,%22H%22,%22H%22],[1,%22|0%E2%9F%A9%E2%9F%A80|%22,%22|0%E2%9F%A9%E2%9F%A80|%22,%22|0%E2%9F%A9%E2%9F%A80|%22,%22|0%E2%9F%A9%E2%9F%A80|%22],[%22%E2%80%A6%22]]}}{this link} for the corresponding quantum circuit in Quirk.

We further apply the pipeline to every $nT \to 1T$ protocol reported in Table~\ref{tab:all_results}.
For all distance 4 and above, the minimization returned an active-qubit count equal to their number of rows, i.e.~no reduction was found.
While some distance 3 protocols can be compressed, none of them can be compressed below $4$ active qubits.
Additionally, we extend our recycling pipeline to $\CCZ$ distillation protocols, recovering this way an implementation of the $64T$-to-$2\CCZ$ protocol on $10$ active qubits, instead of $N=17$, already reported in~\cite{ibm_qbit_recycling}.
No improvements has been found on other $\CCZ$ protocols discussed in this work.
\begin{figure*}[ht]
    \centering
    \includegraphics[width=0.87\linewidth]{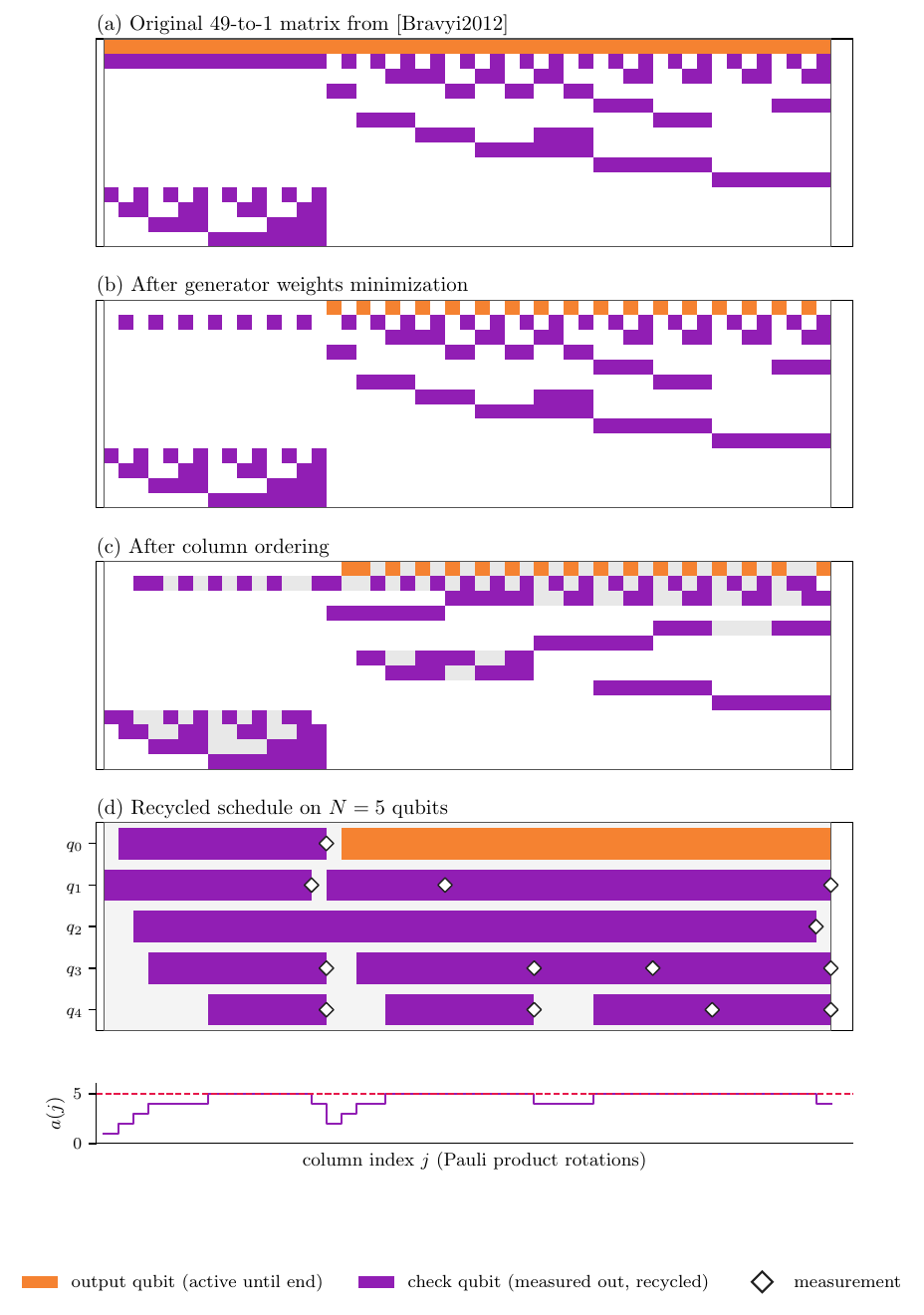}
    \label{fig:Recycling}
    \caption{\textbf{Compressing the $49$-to-$1$ protocol using qubit recycling.}
    (a) The original $49$-to-$1$ distillation matrix of Ref.~\cite{Bravyi_2012} on $N=14$ rows, with the output row (orange) of Hamming weight $49$ and the $13$ check rows (purple).
    (b) After minimizing the generator weights with the method of ~\Cref{app: min generator weight}, bringing down the output row weight from $49$ to $17$.
    (c) After reordering the columns, which reorder the commuting Pauli product rotations. The active window of each qubit, from its first to its last rotation, is shaded in grey.
    (d) The corresponding recycled schedule on $5$ active qubits. Each check qubit is measured out once its last rotation has been applied (diamond) and its qubit is reused by a later row, while the output qubit remains active until the end. The trace below shows the number of active qubits $a(j)$ at each column $j$, which never exceeds $5$.
    }
\end{figure*}

%% file: sample.bib
@article{Bravyi_2013,
  title = {Classification of Topologically Protected Gates for Local Stabilizer Codes},
  author = {Bravyi, Sergey and K\"onig, Robert},
  journal = {Phys. Rev. Lett.},
  volume = {110},
  issue = {17},
  pages = {170503},
  numpages = {5},
  year = {2013},
  month = {Apr},
  publisher = {American Physical Society},
  doi = {10.1103/PhysRevLett.110.170503},
  url = {https://link.aps.org/doi/10.1103/PhysRevLett.110.170503}
}

@misc{gidney2024magicstatecultivation,
      title={Magic state cultivation: growing T states as cheap as CNOT gates}, 
      author={Craig Gidney and Noah Shutty and Cody Jones},
      year={2024},
      eprint={2409.17595},
      archivePrefix={arXiv},
      primaryClass={quant-ph},
      url={https://arxiv.org/abs/2409.17595}, 
}

@article{jones_low-overhead_2013,
	title = {Low-overhead constructions for the fault-tolerant {Toffoli} gate},
	volume = {87},
	url = {https://link.aps.org/doi/10.1103/PhysRevA.87.022328},
	doi = {10.1103/PhysRevA.87.022328},
	number = {2},
	urldate = {2025-07-26},
	journal = {Physical Review A},
	publisher = {American Physical Society},
	author = {Jones, Cody},
	month = feb,
	year = {2013},
	pages = {022328},
}

@article{gidney_halving_2018,
	title = {Halving the cost of quantum addition},
	volume = {2},
	url = {https://quantum-journal.org/papers/q-2018-06-18-74/},
	doi = {10.22331/q-2018-06-18-74},
	urldate = {2025-07-26},
	journal = {Quantum},
	publisher = {Verein zur Förderung des Open Access Publizierens in den Quantenwissenschaften},
	author = {Gidney, Craig},
	month = jun,
	year = {2018},
	pages = {74},
}

@article{barizien2025accessiblequantumgatesclassical,
  title = {Accessible Quantum Gates on Classical Stabilizer Codes},
  author = {Barizien, Victor and Jacinto, Hugo and Sangouard, Nicolas},
  journal = {Phys. Rev. Lett.},
  volume = {136},
  issue = {3},
  pages = {030602},
  numpages = {7},
  year = {2026},
  month = {Jan},
  publisher = {American Physical Society},
  doi = {10.1103/h741-4678},
  url = {https://link.aps.org/doi/10.1103/h741-4678}
}

@article{Haah_2018,
   title={Codes and Protocols for Distilling {$T$}, controlled-{$S$}, and Toffoli Gates},
   volume={2},
   ISSN={2521-327X},
   url={http://dx.doi.org/10.22331/q-2018-06-07-71},
   DOI={10.22331/q-2018-06-07-71},
   journal={Quantum},
   publisher={Verein zur Forderung des Open Access Publizierens in den Quantenwissenschaften},
   author={Haah, Jeongwan and Hastings, Matthew B.},
   year={2018},
   month=jun,
   pages={71}
}

@article{Litinski_2019,
   title={Magic State Distillation: Not as Costly as You Think},
   volume={3},
   ISSN={2521-327X},
   url={http://dx.doi.org/10.22331/q-2019-12-02-205},
   DOI={10.22331/q-2019-12-02-205},
   journal={Quantum},
   publisher={Verein zur Forderung des Open Access Publizierens in den Quantenwissenschaften},
   author={Litinski, Daniel},
   year={2019},
   month=dec, pages={205} }

@article{Bravyi_2012,
  title = {Magic-state distillation with low overhead},
  author = {Bravyi, Sergey and Haah, Jeongwan},
  journal = {Phys. Rev. A},
  volume = {86},
  issue = {5},
  pages = {052329},
  numpages = {10},
  year = {2012},
  month = {Nov},
  publisher = {American Physical Society},
  doi = {10.1103/PhysRevA.86.052329},
  url = {https://link.aps.org/doi/10.1103/PhysRevA.86.052329}
}

@article{Bravyi_2005,
  title = {Universal quantum computation with ideal Clifford gates and noisy ancillas},
  author = {Bravyi, Sergey and Kitaev, Alexei},
  journal = {Phys. Rev. A},
  volume = {71},
  issue = {2},
  pages = {022316},
  numpages = {14},
  year = {2005},
  month = {Feb},
  publisher = {American Physical Society},
  doi = {10.1103/PhysRevA.71.022316},
  url = {https://link.aps.org/doi/10.1103/PhysRevA.71.022316}
}

@article{Horsman_2012,
   title={Surface code quantum computing by lattice surgery},
   volume={14},
   ISSN={1367-2630},
   url={http://dx.doi.org/10.1088/1367-2630/14/12/123011},
   DOI={10.1088/1367-2630/14/12/123011},
   number={12},
   journal={New Journal of Physics},
   publisher={IOP Publishing},
   author={Horsman, Dominic and Fowler, Austin G and Devitt, Simon and Meter, Rodney Van},
   year={2012},
   month=dec, pages={123011} }

@article{Nezami_2022,
   title={Classification of small triorthogonal codes},
   volume={106},
   ISSN={2469-9934},
   url={http://dx.doi.org/10.1103/PhysRevA.106.012437},
   DOI={10.1103/physreva.106.012437},
   number={1},
   journal={Physical Review A},
   publisher={American Physical Society (APS)},
   author={Nezami, Sepehr and Haah, Jeongwan},
   year={2022},
   month=jul }

@misc{Shi_2024,
      title={Triorthogonal Codes and Self-dual Codes}, 
      author={Minjia Shi and Haodong Lu and Jon-Lark Kim and Patrick Sole},
      year={2024},
      eprint={2408.09685},
      archivePrefix={arXiv},
      primaryClass={cs.IT},
      url={https://arxiv.org/abs/2408.09685}, 
}

@article{Litinski_2019_GoS,
   title={A Game of Surface Codes: Large-Scale Quantum Computing with Lattice Surgery},
   volume={3},
   ISSN={2521-327X},
   url={http://dx.doi.org/10.22331/q-2019-03-05-128},
   DOI={10.22331/q-2019-03-05-128},
   journal={Quantum},
   publisher={Verein zur Forderung des Open Access Publizierens in den Quantenwissenschaften},
   author={Litinski, Daniel},
   year={2019},
   month=mar, pages={128} }

@article{campbell_magic-state_2012,
	title = {Magic-{State} {Distillation} in {All} {Prime} {Dimensions} {Using} {Quantum} {Reed}-{Muller} {Codes}},
	volume = {2},
	url = {https://link.aps.org/doi/10.1103/PhysRevX.2.041021},
	doi = {10.1103/PhysRevX.2.041021},
	number = {4},
	urldate = {2026-05-04},
	journal = {Physical Review X},
	publisher = {American Physical Society},
	author = {Campbell, Earl T. and Anwar, Hussain and Browne, Dan E.},
	month = dec,
	year = {2012},
	pages = {041021},
}

@article{Campbell_2017,
  title = {Unified framework for magic state distillation and multiqubit gate synthesis with reduced resource cost},
  author = {Campbell, Earl T. and Howard, Mark},
  journal = {Phys. Rev. A},
  volume = {95},
  issue = {2},
  pages = {022316},
  numpages = {22},
  year = {2017},
  month = {Feb},
  publisher = {American Physical Society},
  doi = {10.1103/PhysRevA.95.022316},
  url = {https://link.aps.org/doi/10.1103/PhysRevA.95.022316}
}

@article{jones_multilevel_2013,
	title = {Multilevel distillation of magic states for quantum computing},
	volume = {87},
	url = {https://link.aps.org/doi/10.1103/PhysRevA.87.042305},
	doi = {10.1103/PhysRevA.87.042305},
	number = {4},
	urldate = {2026-05-04},
	journal = {Physical Review A},
	publisher = {American Physical Society},
	author = {Jones, Cody},
	month = apr,
	year = {2013},
	pages = {042305},
}

@article{Campbell_2016,
   title={An efficient magic state approach to small angle rotations},
   volume={1},
   ISSN={2058-9565},
   url={http://dx.doi.org/10.1088/2058-9565/1/1/015007},
   DOI={10.1088/2058-9565/1/1/015007},
   number={1},
   journal={Quantum Science and Technology},
   publisher={IOP Publishing},
   author={Campbell, Earl T and O’Gorman, Joe},
   year={2016},
   month=Dec, pages={015007} }

@article{Duclos_Cianci_2015,
  title = {Reducing the quantum-computing overhead with complex gate distillation},
  author = {Duclos-Cianci, Guillaume and Poulin, David},
  journal = {Phys. Rev. A},
  volume = {91},
  issue = {4},
  pages = {042315},
  numpages = {9},
  year = {2015},
  month = {Apr},
  publisher = {American Physical Society},
  doi = {10.1103/PhysRevA.91.042315},
  url = {https://link.aps.org/doi/10.1103/PhysRevA.91.042315}
}

@article{Gidney_2019,
   title={Efficient magic state factories with a catalyzed $\ket{\CCZ}$ to $2 \ket{T}$ transformation},
   volume={3},
   ISSN={2521-327X},
   url={http://dx.doi.org/10.22331/q-2019-04-30-135},
   DOI={10.22331/q-2019-04-30-135},
   journal={Quantum},
   publisher={Verein zur Forderung des Open Access Publizierens in den Quantenwissenschaften},
   author={Gidney, Craig and Fowler, Austin G.},
   year={2019},
   month=Apr, pages={135} }

@misc{tiurev2026parityunfoldeddistillationarchitecturenoisebiased,
      title={Parity-unfolded distillation architecture for noise-biased platforms}, 
      author={Konstantin Tiurev and Christoph Fleckenstein and Christophe Goeller and Paul Schnabl and Matthias Traube and Nitica Sakharwade and Anette Messinger and Josua Unger and Wolfgang Lechner},
      year={2026},
      eprint={2604.15436},
      archivePrefix={arXiv},
      primaryClass={quant-ph},
      url={https://arxiv.org/abs/2604.15436}, 
}

@BOOK{Nielsen2010,
  title     = "Quantum Computation and Quantum Information",
  author    = "Nielsen, Michael A and Chuang, Isaac L",
  publisher = "Cambridge University Press",
  month     =  dec,
  year      =  2010,
  address   = "Cambridge, England"
}

@article{itogawa_efficient_2025,
	title = {Efficient {Magic} {State} {Distillation} by {Zero}-{Level} {Distillation}},
	volume = {6},
	url = {https://link.aps.org/doi/10.1103/thxx-njr6},
	doi = {10.1103/thxx-njr6},
	number = {2},
	urldate = {2026-05-04},
	journal = {PRX Quantum},
	publisher = {American Physical Society},
	author = {Itogawa, Tomohiro and Takada, Yugo and Hirano, Yutaka and Fujii, Keisuke},
	month = jun,
	year = {2025},
	pages = {020356},
}

@misc{cain_shors_2026,
	title = {Shor's algorithm is possible with as few as 10,000 reconfigurable atomic qubits},
	url = {http://arxiv.org/abs/2603.28627},
	doi = {10.48550/arXiv.2603.28627},
	urldate = {2026-05-04},
	publisher = {arXiv},
	author = {Cain, Madelyn and Xu, Qian and King, Robbie and Picard, Lewis R. B. and Levine, Harry and Endres, Manuel and Preskill, John and Huang, Hsin-Yuan and Bluvstein, Dolev},
	month = mar,
	year = {2026},
	note = {arXiv:2603.28627 [quant-ph]},
}

@article{eastin2009,
  title = {Restrictions on Transversal Encoded Quantum Gate Sets},
  author = {Eastin, Bryan and Knill, Emanuel},
  journal = {Phys. Rev. Lett.},
  volume = {102},
  issue = {11},
  pages = {110502},
  numpages = {4},
  year = {2009},
  month = {Mar},
  publisher = {American Physical Society},
  doi = {10.1103/PhysRevLett.102.110502},
  url = {https://link.aps.org/doi/10.1103/PhysRevLett.102.110502}
}

@inproceedings{Gottesman1998,
    author = "Gottesman, Daniel",
    title = "{The Heisenberg representation of quantum computers}",
    booktitle = "{22nd International Colloquium on Group Theoretical Methods in Physics}",
    eprint = "quant-ph/9807006",
    archivePrefix = "arXiv",
    reportNumber = "LAUR-98-2848, LA-UR-98-2848",
    pages = "32--43",
    month = "7",
    year = "1998"
}

@article{aaronson2004,
  title = {Improved simulation of stabilizer circuits},
  author = {Aaronson, Scott and Gottesman, Daniel},
  journal = {Phys. Rev. A},
  volume = {70},
  issue = {5},
  pages = {052328},
  numpages = {14},
  year = {2004},
  month = {Nov},
  publisher = {American Physical Society},
  doi = {10.1103/PhysRevA.70.052328},
  url = {https://link.aps.org/doi/10.1103/PhysRevA.70.052328}
}

@misc{google_cultivation_experiment,
      title={Magic state cultivation on a superconducting quantum processor}, 
      author={Emma Rosenfeld and Craig Gidney and Gabrielle Roberts and Alexis Morvan and Nathan Lacroix and Dvir Kafri and Jeffrey Marshall and Ming Li and Volodymyr Sivak et al.},
      year={2025},
      eprint={2512.13908},
      archivePrefix={arXiv},
      primaryClass={quant-ph},
      url={https://arxiv.org/abs/2512.13908}, 
}

@misc{beverland_resource_estimation,
      title={Assessing requirements to scale to practical quantum advantage}, 
      author={Michael E. Beverland and Prakash Murali and Matthias Troyer and Krysta M. Svore and Torsten Hoefler and Vadym Kliuchnikov and Guang Hao Low and Mathias Soeken and Aarthi Sundaram and Alexander Vaschillo},
      year={2022},
      eprint={2211.07629},
      archivePrefix={arXiv},
      primaryClass={quant-ph},
      url={https://arxiv.org/abs/2211.07629}, 
}

@misc{gidney2025factor2048bitrsa,
      title={How to factor 2048 bit RSA integers with less than a million noisy qubits}, 
      author={Craig Gidney},
      year={2025},
      eprint={2505.15917},
      archivePrefix={arXiv},
      primaryClass={quant-ph},
      url={https://arxiv.org/abs/2505.15917}, 
}

@misc{kasai2026breakingorthogonalitybarrierquantum,
      title={Breaking the Orthogonality Barrier in Quantum LDPC Codes}, 
      author={Kenta Kasai},
      year={2026},
      eprint={2601.08824},
      archivePrefix={arXiv},
      primaryClass={quant-ph},
      url={https://arxiv.org/abs/2601.08824}, 
}

@misc{quera_qldpc_VHER,
      title={Towards Ultra-High-Rate Quantum Error Correction with Reconfigurable Atom Arrays}, 
      author={Chen Zhao and Casey Duckering and Andi Gu and Nishad Maskara and Hengyun Zhou},
      year={2026},
      eprint={2604.16209},
      archivePrefix={arXiv},
      primaryClass={quant-ph},
      url={https://arxiv.org/abs/2604.16209}, 
}

@misc{ibm_qbit_recycling,
      title={Distilling Magic States in the Bicycle Architecture}, 
      author={Shifan Xu and Kun Liu and Patrick Rall and Zhiyang He and Yongshan Ding},
      year={2026},
      eprint={2602.20546},
      archivePrefix={arXiv},
      primaryClass={quant-ph},
      url={https://arxiv.org/abs/2602.20546}, 
}

@article{Jones_2013_64_to_2,
  title = {Composite Toffoli gate with two-round error detection},
  author = {Jones, Cody},
  journal = {Phys. Rev. A},
  volume = {87},
  issue = {5},
  pages = {052334},
  numpages = {8},
  year = {2013},
  month = {May},
  publisher = {American Physical Society},
  doi = {10.1103/PhysRevA.87.052334},
  url = {https://link.aps.org/doi/10.1103/PhysRevA.87.052334}
}

@article{Gidney_2021,
   title={How to factor 2048 bit RSA integers in 8 hours using 20 million noisy qubits},
   volume={5},
   ISSN={2521-327X},
   url={http://dx.doi.org/10.22331/q-2021-04-15-433},
   DOI={10.22331/q-2021-04-15-433},
   journal={Quantum},
   publisher={Verein zur Forderung des Open Access Publizierens in den Quantenwissenschaften},
   author={Gidney, Craig and Ekerå, Martin},
   year={2021},
   month=Apr, pages={433} }

@misc{baldelli2026constructingdecodingquantumtriorthogonal,
      title={On Constructing and Decoding Quantum Triorthogonal Codes}, 
      author={Alessio Baldelli and Olai \r{A}. Mostad and Hsuan-Yin Lin and Eirik Rosnes and Massimo Battaglioni},
      year={2026},
      eprint={2605.24519},
      archivePrefix={arXiv},
      primaryClass={quant-ph},
      url={https://arxiv.org/abs/2605.24519}, 
}

@misc{litinski2023compute256bitellipticcurve,
      title={How to compute a 256-bit elliptic curve private key with only 50 million Toffoli gates}, 
      author={Daniel Litinski},
      year={2023},
      eprint={2306.08585},
      archivePrefix={arXiv},
      primaryClass={quant-ph},
      url={https://arxiv.org/abs/2306.08585}, 
}

@article{Gouzien_2023,
  title = {Performance Analysis of a Repetition Cat Code Architecture: Computing 256-bit Elliptic Curve Logarithm in 9 Hours with 126 133 Cat Qubits},
  author = {Gouzien, \'Elie and Ruiz, Diego and Le R\'egent, Francois-Marie and Guillaud, J\'er\'emie and Sangouard, Nicolas},
  journal = {Phys. Rev. Lett.},
  volume = {131},
  issue = {4},
  pages = {040602},
  numpages = {7},
  year = {2023},
  month = {Jul},
  publisher = {American Physical Society},
  doi = {10.1103/PhysRevLett.131.040602},
  url = {https://link.aps.org/doi/10.1103/PhysRevLett.131.040602}
}

@article{Gouzien_2021,
  title = {Factoring 2048-bit RSA Integers in 177 Days with 13 436 Qubits and a Multimode Memory},
  author = {Gouzien, \'Elie and Sangouard, Nicolas},
  journal = {Phys. Rev. Lett.},
  volume = {127},
  issue = {14},
  pages = {140503},
  numpages = {6},
  year = {2021},
  month = {Sep},
  publisher = {American Physical Society},
  doi = {10.1103/PhysRevLett.127.140503},
  url = {https://link.aps.org/doi/10.1103/PhysRevLett.127.140503}
}

@inproceedings{Zhou_2025, series={SIGARCH ’25},
   title={Resource Analysis of Low-Overhead Transversal Architectures for Reconfigurable Atom Arrays},
   url={http://dx.doi.org/10.1145/3695053.3731039},
   DOI={10.1145/3695053.3731039},
   booktitle={Proceedings of the 52nd Annual International Symposium on Computer Architecture},
   publisher={ACM},
   author={Zhou, Hengyun and Duckering, Casey and Zhao, Chen and Bluvstein, Dolev and Cain, Madelyn and Kubica, Aleksander and Wang, Sheng-Tao and Lukin, Mikhail D.},
   year={2025},
   month=Jun, pages={1432–1448},
   collection={SIGARCH ’25} }

@article{Jacinto_2026,
  title = {Network requirements for distributed quantum computation},
  author = {Jacinto, Hugo and Gouzien, \'Elie and Sangouard, Nicolas},
  journal = {Phys. Rev. Res.},
  volume = {8},
  issue = {1},
  pages = {013205},
  numpages = {7},
  year = {2026},
  month = {Feb},
  publisher = {American Physical Society},
  doi = {10.1103/v9ln-c4v2},
  url = {https://link.aps.org/doi/10.1103/v9ln-c4v2}
}

@article{Preskill_2025,
   title={Beyond NISQ: The Megaquop Machine},
   volume={6},
   ISSN={2643-6817},
   url={http://dx.doi.org/10.1145/3723153},
   DOI={10.1145/3723153},
   number={3},
   journal={ACM Transactions on Quantum Computing},
   publisher={Association for Computing Machinery (ACM)},
   author={Preskill, John},
   year={2025},
   month=Apr, pages={1–7} }

@article{Dennis_2002,
   title={Topological quantum memory},
   volume={43},
   ISSN={1089-7658},
   url={http://dx.doi.org/10.1063/1.1499754},
   DOI={10.1063/1.1499754},
   number={9},
   journal={Journal of Mathematical Physics},
   publisher={AIP Publishing},
   author={Dennis, Eric and Kitaev, Alexei and Landahl, Andrew and Preskill, John},
   year={2002},
   month=Sep, pages={4452–4505} }

@misc{fowler2019lowoverheadquantumcomputation,
      title={Low overhead quantum computation using lattice surgery}, 
      author={Austin G. Fowler and Craig Gidney},
      year={2019},
      eprint={1808.06709},
      archivePrefix={arXiv},
      primaryClass={quant-ph},
      url={https://arxiv.org/abs/1808.06709}, 
}

@article{Bravyi_2024,
   title={High-threshold and low-overhead fault-tolerant quantum memory},
   volume={627},
   ISSN={1476-4687},
   url={http://dx.doi.org/10.1038/s41586-024-07107-7},
   DOI={10.1038/s41586-024-07107-7},
   number={8005},
   journal={Nature},
   publisher={Springer Science and Business Media LLC},
   author={Bravyi, Sergey and Cross, Andrew W. and Gambetta, Jay M. and Maslov, Dmitri and Rall, Patrick and Yoder, Theodore J.},
   year={2024},
   month=Mar, pages={778–782} }

@misc{koutsioumpas2022smallestcodetransversalt,
      title={The Smallest Code with Transversal T}, 
      author={Stergios Koutsioumpas and Darren Banfield and Alastair Kay},
      year={2022},
      eprint={2210.14066},
      archivePrefix={arXiv},
      primaryClass={quant-ph},
      url={https://arxiv.org/abs/2210.14066}, 
}

@misc{yoder2025tourgrossmodularquantum,
      title={Tour de gross: A modular quantum computer based on bivariate bicycle codes}, 
      author={Theodore J. Yoder and Eddie Schoute and Patrick Rall and Emily Pritchett and Jay M. Gambetta and Andrew W. Cross and Malcolm Carroll and Michael E. Beverland},
      year={2025},
      eprint={2506.03094},
      archivePrefix={arXiv},
      primaryClass={quant-ph},
      url={https://arxiv.org/abs/2506.03094}, 
}

@article{serraperalta2025decodingtransversalcliffordgates,
  title = {Decoding across Transversal Clifford Gates in the Surface Code},
  author = {Serra-Peralta, Marc and Shaw, Mackenzie H. and Terhal, Barbara M.},
  journal = {PRX Quantum},
  volume = {7},
  issue = {1},
  pages = {010335},
  numpages = {41},
  year = {2026},
  month = {Feb},
  publisher = {American Physical Society},
  doi = {10.1103/sk5y-25b1},
  url = {https://link.aps.org/doi/10.1103/sk5y-25b1}
}

@article{Lee_2021,
  title = {Even More Efficient Quantum Computations of Chemistry Through Tensor Hypercontraction},
  author = {Lee, Joonho and Berry, Dominic W. and Gidney, Craig and Huggins, William J. and McClean, Jarrod R. and Wiebe, Nathan and Babbush, Ryan},
  journal = {PRX Quantum},
  volume = {2},
  issue = {3},
  pages = {030305},
  numpages = {62},
  year = {2021},
  month = {Jul},
  publisher = {American Physical Society},
  doi = {10.1103/PRXQuantum.2.030305},
  url = {https://link.aps.org/doi/10.1103/PRXQuantum.2.030305}
}

@article{Campbell_2017_Terhal,
   title={Roads towards fault-tolerant universal quantum computation},
   volume={549},
   ISSN={1476-4687},
   url={http://dx.doi.org/10.1038/nature23460},
   DOI={10.1038/nature23460},
   number={7671},
   journal={Nature},
   publisher={Springer Science and Business Media LLC},
   author={Campbell, Earl T. and Terhal, Barbara M. and Vuillot, Christophe},
   year={2017},
   month=Sep, pages={172–179} }

@misc{compact_distillation,
  title        = {compact\_distillation: {SAT} problem to find compact magic state distillation factories},
  year         = {2026},
  publisher    = {GitHub},
  howpublished = {\url{https://github.com/xvalcarce/compact_distillation}},
  license      = {BSD-2-Clause},
  note         = {Code accompanying the manuscript},
}

@inproceedings{z3, author = {De Moura, Leonardo and Bj\o{}rner, Nikolaj}, title = {Z3: an efficient SMT solver}, year = {2008}, isbn = {3540787992}, publisher = {Springer-Verlag}, address = {Berlin, Heidelberg}, booktitle = {Proceedings of the Theory and Practice of Software, 14th International Conference on Tools and Algorithms for the Construction and Analysis of Systems}, pages = {337–340}, numpages = {4}, location = {Budapest, Hungary}, series = {TACAS'08/ETAPS'08} }

@software{cpsatlp,
  title = {CP-SAT},
  version = { v9.10 },
  author = {Laurent Perron and Frédéric Didier},
  organization = {Google},
  url = {https://developers.google.com/optimization/cp/cp_solver/},
  year = {2024},
  month=May,
  date = { 2024-05-07 }
}

@misc{londe2026localdistillationreedmuller,
      title={Local distillation from Reed Muller codes unfolding}, 
      author={Vivien Londe},
      year={2026},
      eprint={2605.06284},
      archivePrefix={arXiv},
      primaryClass={quant-ph},
      url={https://arxiv.org/abs/2605.06284}, 
}

@misc{schrottenloher2026optimizedpointadditioncircuits,
      title={Optimized Point Addition Circuits for Elliptic Curve Discrete Logarithms}, 
      author={André Schrottenloher},
      year={2026},
      eprint={2606.02235},
      archivePrefix={arXiv},
      primaryClass={quant-ph},
      url={https://arxiv.org/abs/2606.02235}, 
}

@book{Review_applications,
   title={Quantum Algorithms: A Survey of Applications and End-to-end Complexities},
   ISBN={9781009639668},
   url={http://dx.doi.org/10.1017/9781009639651},
   DOI={10.1017/9781009639651},
   publisher={Cambridge University Press},
   author={Dalzell, Alexander M. and McArdle, Sam and Berta, Mario and Bienias, Przemyslaw and Chen, Chi-Fang and Gilyén, András and Hann, Connor T. and Kastoryano, Michael J. and Khabiboulline, Emil T. and Kubica, Aleksander and Salton, Grant and Wang, Samson and Brandão, Fernando G. S. L.},
   year={2025},
   month=Apr }

@inproceedings{MS_injection_Surface_code, author = {Lao, Lingling and Criger, Ben}, title = {Magic state injection on the rotated surface code}, year = {2022}, isbn = {9781450393386}, publisher = {Association for Computing Machinery}, address = {New York, NY, USA}, url = {https://doi.org/10.1145/3528416.3530237}, doi = {10.1145/3528416.3530237}, booktitle = {Proceedings of the 19th ACM International Conference on Computing Frontiers}, pages = {113–120}, numpages = {8}, location = {Turin, Italy}, series = {CF '22} }

@article{Li_2015,
   title={A magic state’s fidelity can be superior to the operations that created it},
   volume={17},
   ISSN={1367-2630},
   url={http://dx.doi.org/10.1088/1367-2630/17/2/023037},
   DOI={10.1088/1367-2630/17/2/023037},
   number={2},
   journal={New Journal of Physics},
   publisher={IOP Publishing},
   author={Li, Ying},
   year={2015},
   month=Feb, pages={023037} }

@misc{low2026denserplanarsurfacecode,
      title={A Denser Planar Surface Code}, 
      author={Guang Hao Low and William J. Huggins and Dominic W. Berry and Tanuj Khattar and Alec F. White and Nicholas C. Rubin and Ryan Babbush},
      year={2026},
      eprint={2605.30455},
      archivePrefix={arXiv},
      primaryClass={quant-ph},
      url={https://arxiv.org/abs/2605.30455}, 
}

@misc{gidney_2024_13777072,
  author       = {Gidney, Craig and
                  Jones, Cody and
                  Shutty, Noah},
  title        = {Data for "Magic state cultivation: growing T
                   states as cheap as CNOT gates"
                  },
  month        = sep,
  year         = 2024,
  publisher    = {Zenodo},
  doi          = {10.5281/zenodo.13777072},
  url          = {https://doi.org/10.5281/zenodo.13777072},
}

@misc{singh2026borrowed,
  title        = {Borrowed Identities: Malleable Distillation Factories and a Unified Numerical Search},
  author       = {Singh, Shraddha and Gidney, Craig and Jones, Cody},
  year         = {2026},
  eprint       = {2606.28518},
  archivePrefix = {arXiv},
  primaryClass = {quant-ph}
}
